\documentclass[11pt]{article}
\usepackage[utf8]{inputenc}
\usepackage{graphicx}
\usepackage{subcaption}
\usepackage{amsfonts}
\usepackage{amsmath}
\usepackage{amssymb}
\usepackage{bm}
\usepackage{url}
\usepackage{fancyhdr}
\usepackage{float}
\usepackage{indentfirst}
\usepackage{enumerate}
\usepackage{dsfont,footmisc}
\usepackage{comment}
\usepackage{caption}
\usepackage{multirow}
\usepackage{tabularx}

\usepackage{setspace}

\usepackage[misc]{ifsym}

\usepackage{amsthm}
\usepackage{color}
\usepackage{natbib}
\usepackage[british]{babel}
\usepackage[colorlinks=true,citecolor=blue]{hyperref}
\usepackage[right,pagewise,mathlines]{lineno}

\newcommand{\E}{\mathbb{E}}
\newcommand{\R}{\mathbb{R}}

\newcommand{\p}{\mathbb{P}}

\renewcommand{\ge}{\geqslant}
\renewcommand{\le}{\leqslant}
\renewcommand{\geq}{\geqslant}
\renewcommand{\leq}{\leqslant}
\renewcommand{\epsilon}{\varepsilon}

\theoremstyle{plain}
\newtheorem{theorem}{Theorem}[section]

\newtheorem{lemma}{Lemma}[section]
\newtheorem{proposition}{Proposition}[section]

\theoremstyle{definition}
\newtheorem{definition}{Definition}[section]

\theoremstyle{remark}
\newtheorem{remark}{Remark}[section]

\theoremstyle{definition}

\numberwithin{equation}{section} \numberwithin{theorem}{section}
\renewcommand{\cite}{\citet}

\DeclareMathOperator*{\argmin}{arg\,min}
\newcommand{\be}{\begin{eqnarray}}
\newcommand{\ee}{\end{eqnarray}}
\newcommand{\by}{\begin{eqnarray*}}
\newcommand{\ey}{\end{eqnarray*}}
\newcommand{\bn}{\begin{enumerate}}
\newcommand{\en}{\end{enumerate}}
\newcommand{\bi}{\begin{itemize}}
\newcommand{\ei}{\end{itemize}}

\title{Conditional value-at-risk   under  reward-penalty mechanism   with applications to robust portfolio management}
\author{Jun Cai\thanks{Department of Statistics and Actuarial Science, University of Waterloo,  Canada. \Letter~{\url{jcai@uwaterloo.ca}}}, \  \  \ \  Tiantian Mao\thanks{Department of Statistics and Finance, School of Management,
University of Science and Technology of China, China. \Letter~{\url{tmao@ustc.edu.cn}}},\ \ \ \  Zhiqiao Song\thanks{Department of Statistics and Actuarial Science, University of Waterloo,  Canada. \Letter~{\url{z22song@uwaterloo.ca}}}
}

\date{\today}

\begin{document}

\maketitle

	\begin{abstract}

In this paper, we present robust portfolio selection models by incorporating a reward and penalty mechanism into portfolio management. We assume that the joint distribution of the losses of the underlying risky assets in a portfolio is uncertain but lies within a multivariate distribution set. Our goal is to identify optimal portfolio allocations by minimizing the worst-case conditional value-at-risk (CVaR) of portfolio loss under the reward and penalty mechanism and distribution uncertainty. Our models can also be used to investigate the problem of how to balance portfolio losses with  their   associated downside risk in portfolio management.  The core of such a robust portfolio selection model is the worst-case CVaR.
  We first derive an explicit closed-form expression for the worst-case CVaR under the reward-penalty mechanism, which generalizes several existing models and results regarding the worst-case CVaR, such as those studied in \cite{jagannathan1977minimax}, \cite{chen2011tight}, and \cite{cai2024worst}. 
  We then apply this expression to obtain optimal portfolio allocations that minimize the worst-case CVaR under both a classical mean-covariance-based multivariate distribution set and a generalized mean-covariance-based multivariate distribution set  introduced in \cite{kang2019data}.
  Additionally, we utilize real market data to illustrate the application of the proposed models and the corresponding optimal portfolio allocations in portfolio management. Our empirical experiments show that portfolios based on the proposed models have the potential to outperform those based on several existing related models. Furthermore, the results demonstrate that incorporating downside risk into portfolio loss helps better manage risk and can achieve higher investment returns than considering either downside risk or portfolio loss alone. {\color{black} Moreover, our experiments reveal the trade-off between improving expected portfolio return and controlling the worst-case portfolio CVaR.}  \\

\textbf{Keywords}: Conditional value-at-risk,   reward, penalty, downside risk, distribution uncertainty,   worst-case risk measure, convex optimization,   robust portfolio allocation.
  	\end{abstract}

\maketitle

\newpage

\section{Introduction}\label{sec-intro}

Optimal portfolio selection 
%, or the optimal allocation of wealth across various risky assets, 
has been an important research topic in quantitative risk management for finance.
 The classical mean-variance portfolio selection model, 
 %introduced by Harry Markowitz in the 1950s, 
 as described in \citet{markowitz1952}, aims to minimize portfolio variance subject to a targeted expected return.
  However, the risk measure of variance does not account for the tail or downside risk of a portfolio. To address this issue, several moment-based risk measures, such as semi-variance, lower partial moment (LPM), and upper partial moment (UPM), have been employed in optimal portfolio selection models (see, for example, \cite{roy1952safety}, \cite{bawa1975optimal}, \cite{MTXY93}, \cite{K11}, \cite{KLMY17}, \cite{LMYZ22}, and references therein). Additionally, tail-based mean-variance portfolio selection models have been proposed in the literature, including those in \cite{LZ10}, \cite{OL13}, \cite{EK21}, \cite{HWW24}, and references therein. With the advancement of modern risk management tools (see \cite{NFE15}), many quantile-based risk measures have been extensively applied to optimal portfolio selection. In particular, the quantile-based risk measures of value-at-risk (VaR) and conditional value-at-risk (CVaR), also referred to as tail value-at-risk (TVaR), have been widely used in optimal portfolio selection due to their ability to capture tail risk (see, for example, \cite{basak2001value}, \cite{rockafellar2002conditional}, and references therein).

Suppose  a portfolio manager at a financial institution invests  a certain amount of wealth   into 
 $n$ different risky assets.  Let $\boldsymbol{X} =(X_1,...,X_n)\in\mathbb{R}^n$ represent 
  the random losses 
  of these  $n$ assets.  The portfolio loss is computed as $\boldsymbol{\omega}^\top\boldsymbol{X}=\sum_{i=1}^{n}\omega_{i}X_{i},$ 
where $\boldsymbol{\omega} =(\omega_1,..., \omega_n)\in\mathbb{R}^n$,  
 $\omega_i$ represents the proportion of the   wealth allocated to  asset $i$  at the beginning of an investment period, and $X_i$ is the loss of asset $i$ at the end of the investment period for $i=1,...,n$.  
 We have  that    $\boldsymbol{\omega}^{\top}\boldsymbol{e}=1$, where $\boldsymbol{e}=(1,...,1)$ is an $n$-dimensional  vector 
with all elements equal to 1. 
 Note that in this paper, positive values of a loss random variable or vector indicate  losses, while  negative values represent gains or returns.

     An optimal  portfolio selection model based on a risk measure  is often formulated as the following optimization model:  
  \be
  \label{model-with-df}
  \min_{\boldsymbol{\omega}  \in \cal W} \, \rho(\boldsymbol{\omega}^{\top}\boldsymbol{X}),
  \ee
where $\rho$ is a risk measure and  ${\cal W}\subseteq \mathbb{R}^n$ is a feasible set of 
$\boldsymbol{\omega}$. To apply model \eqref{model-with-df} and calculate  
the risk measure $\rho(\boldsymbol{\omega}^{\top}\boldsymbol{X})$, the complete information about 
the distribution of  $\boldsymbol{X} $  is usually required.    However, the true distribution of $\boldsymbol{X} $  is often unknown. 
This challenge leads to the following robust portfolio selection model:
\be
  \label{model-no-df}
  \min_{\boldsymbol{\omega}  \in   \mathcal{W}} \, \sup_{G \in \mathcal{S}}  \, \rho(\boldsymbol{\omega}^{\top}\boldsymbol{X}^G),
  \ee
where $\cal S$ is the set of  all  possible  distributions of $\boldsymbol{X}$  consistent with the given information, and is often  referred 
to as an {\it uncertainty set} for $\boldsymbol{X}$. Throughout this paper, for a random variable (or vector)  $X$ (resp. $\boldsymbol{X}$), the notation     
 $X^F$ (resp. $\boldsymbol{X}^G$) indicates  that the distribution of  $X$ (resp. $\boldsymbol{X}$) is specified as $F$ (resp. $G$). 
 Additionally, for a univariate (multivariate) function   $h$,  
   the notation 
 $\rho(h(X^F))$ (resp. $\rho(h(\boldsymbol{X}^G))$)   signifies that  the risk measure $\rho$ of $h(X)$ (resp. $h(\boldsymbol{X})$) is  evaluated under the 
 distribution $F$ (resp. $G)$. 
In the robust portfolio selection  model \eqref{model-no-df},  the manager 
  assumes that complete information about the joint distribution of the risky assets is unavailable and that only partial information is known.
  She/he  assumes  
   that all possible joint distributions of the risky assets belong to  an  uncertainty set $\cal S$. An optimal  portfolio   is one that 
minimizes the supremum    value of the risk measure of portfolio loss. This  supremum   value is often referred to as the {\it  worst-case} risk measure, such as worst-case VaR and  worst-case CVaR.  

The robust optimization problem \eqref{model-no-df} has been extensively studied in the literature. For example, \cite{ZF09} investigated this problem using the CVaR risk measure and an uncertainty set comprising all convex combinations of a finite number of known distributions. \citet{bertsimas2009constructing} explored a polyhedral uncertainty set generated by a class of coherent risk measures.
 \citet{chen2011tight} addressed the issue by employing LPM and CVaR risk measures, with an uncertainty set defined by the given mean and covariance. \citet{Li2018} derived closed-form solutions for worst-case law-invariant risk measures and discussed the applications of the worst-case values to robust portfolio optimization. \cite{kang2019data} introduced a generalized mean-covariance-based uncertainty set and investigated robust mean-CVaR portfolio selection.  \citet{lotfi2018robust} examined the problem using VaR and CVaR risk measures under joint ambiguity in distribution, mean returns, and the covariance matrix. \citet{blanchet2022distributionally} proposed an optimal distributionally robust portfolio model that utilizes Wasserstein distance to measure discrepancies from the empirical distribution. \cite{B23} derived expressions for worst-case distortion risk measures and discussed their applications to robust portfolio optimization.  Additionally,  models and  supremum values of risk measures similar to those in \eqref{model-no-df}  also appear in many studies related to insurance. See, for example, \cite{KG86b}, \cite{H02}, \cite{BBBH15}, \cite{H15}, 
 \cite{A19}, 
 \cite{CRV18}, \cite{Liu20}, \cite{LM22}, and references therein.

The portfolio selection models \eqref{model-with-df} and \eqref{model-no-df} primarily focus on the risk measure of the portfolio loss, $\boldsymbol{\omega}^{\top}\boldsymbol{X}$, and do not account for the downside risk of  the portfolio loss.
%, which may arise in  different scenarios.
For example, in practice, a portfolio manager may receive  a reward or bonus if the portfolio performs well.
 Assume that the   manager  will  receive a reward or bonus  at a reward  rate $\beta \geq 0$ if the  return of the  
portfolio exceeds  a pre-specified target return of  $-d$,  specifically, if  $-\boldsymbol{\omega}^{\top}\boldsymbol{X}\geq-d$, or equivalently, if $\boldsymbol{\omega}^{\top}\boldsymbol{X}\leq d$. This  indicates that the  loss of the portfolio does not exceed the  threshold loss level $d$, where $d$ is a real number.    The reward amount  is determined  as $ \beta(-\boldsymbol{\omega}^{\top}\boldsymbol{X}-(-d))_{+} 
 = \beta(\boldsymbol{\omega}^{\top}\boldsymbol{X}-d)_{-}.$ This means the reward equals   $100\beta\%$ of the  return over the target return of $-d$.  Here,  and throughout this paper, for any  $x \in \R=(-\infty, \infty)$, $(x)_+=\max\{x, 0\}$ and  
 $(x)_-=\max\{-x, 0\}$. 
     Conversely,  if the portfolio's 
 performance  does not meet the target (i.e.,  $-\boldsymbol{\omega}^{\top}\boldsymbol{X}<-d$), the manager 
  incurs     a penalty at  a  rate $\lambda \ge 0$.  
 The penalty is computed as $\lambda(\boldsymbol{\omega}^{\top}\boldsymbol{X}-d)_{+},$
  indicating that  the penalty  amount is 
 equal to the $100\lambda \%$ of the portfolio  loss exceeding  the threshold loss level  of $d$.  Note that 
$(\boldsymbol{\omega}^{\top}\boldsymbol{X}- d)_+ $ is referred to as the downside risk of the portfolio.
Thus, with rewards and penalties, the net portfolio loss, denoted by  $L(\boldsymbol{\omega}, \boldsymbol{X})$, for the financial institution at allocation  $\boldsymbol{\omega}$ is given by 
\be
\label{institutional}
      L(\boldsymbol{\omega}, \boldsymbol{X})  & = & \boldsymbol{\omega}^{\top}\boldsymbol{X}+\beta (\boldsymbol{\omega}^{\top}\boldsymbol{X}-d)_{-}-\lambda(\boldsymbol{\omega}^{\top}\boldsymbol{X}-d)_{+} \notag\\
             &=&(1-\beta)\boldsymbol{\omega}^{\top}\boldsymbol{X}+(\beta-\lambda)(\boldsymbol{\omega}^{\top}\boldsymbol{X}-d)_{+}+\beta d.
    \ee
In this setting, it is interesting  to consider the risk measure of  the net portfolio loss  $L(\boldsymbol{\omega}, \boldsymbol{X})$. 
Therefore, to incorporate both rewards and penalties in portfolio management while addressing distributional uncertainty, we propose 
the following robust portfolio selection model:
  \be
  \label{model-no-df-RP}
\min_{\boldsymbol{\omega}  \in   \mathcal{W}} \sup_{G \in \mathcal{S}} 
   \rho(L(\boldsymbol{\omega}, \boldsymbol{X})) =   \min_{\boldsymbol{\omega}  \in   \mathcal{W}} \sup_{G \in \mathcal{S}} 
   \rho((1-\beta)\boldsymbol{\omega}^{\top}\boldsymbol{X}^G+(\beta-\lambda)(\boldsymbol{\omega}^{\top}\boldsymbol{X}^G-d)_{+}+\beta d).
  \ee 
If the risk measure $\rho$  satisfies translation invariance, then solving  problem \eqref{model-no-df-RP} is equivalent to solving the following problem: 
   \be
  \label{model-no-df-RP-Final}
  \min_{\boldsymbol{\omega}  \in \cal W} \sup_{G \in \cal S}  \rho((1-\beta)\boldsymbol{\omega}^{\top}\boldsymbol{X}^G+(\beta-\lambda)(\boldsymbol{\omega}^{\top}\boldsymbol{X}^G-d)_{+}).
  \ee  
Model \eqref{model-no-df-RP-Final} generalizes model \eqref{model-no-df} by incorporating rewards and penalties  into portfolio management.  %thereby enabling a more comprehensive analysis of portfolio performance. 
  
 Additionally, if the risk measure $\rho$ satisfies positive  homogeneity,
   by setting   $1-\beta >0$ and $\beta >  \lambda$  in  model \eqref{model-no-df-RP-Final}, we see that 
   model \eqref{model-no-df-RP-Final}  is equivalent to the following problem: 
     \be
  \label{model-no-df-RP-Final-kappa}
  \min_{\boldsymbol{\omega}  \in \cal W} \sup_{G \in \cal S}  \rho(\boldsymbol{\omega}^{\top}\boldsymbol{X}^G+\kappa (\boldsymbol{\omega}^{\top}\boldsymbol{X}^G-d)_{+}),
  \ee    
where,    throughout this paper,  for $0 \le \lambda < \beta < 1$,  we define  the parameter $\kappa > 0  $ as 
\be
\label{kappa}
\kappa=\frac{\beta-\lambda}{1-\beta}.
\ee

{\color{black} A reward-penalty mechanism, also known as a performance-based incentive, is a common approach in the finance industry for evaluating the performance of fund managers or investors. Typically, a fund manager receives a bonus or reward if the portfolio generates profits or exceeds a target value. This bonus is usually a percentage of the amount by which the portfolio surpasses the target, or the profits achieved, as represented by $ \beta (\boldsymbol{\omega}^{\top}\boldsymbol{X} - d)_{-} $ in our model \eqref{institutional}. 
    Although penalties are less common, they do exist in certain contexts. For example, in private funds, if the portfolio falls below a target level, the fund manager may be required to return a portion of previously awarded bonuses.
  This practice is known as a clawback provision. In particular,    when $\lambda = 0$, our  model \eqref{institutional}  reflects a scenario without penalties. When $\lambda > 0$, the model provides a potential framework for quantifying penalties. There is extensive literature on the implementation of bonuses and penalties in the finance industry and their impact on organizational behavior. 
 See, for example,  \cite{ADN09}, \cite{CGS23}, \cite{EGB03}, \cite{MY10}, \cite{SLT87}, and references therein.

It is also worth noting that, although our proposed model \eqref{model-no-df-RP-Final}, or its equivalent form  model \eqref{model-no-df-RP-Final-kappa},   originates from a reward-penalty mechanism, model \eqref{model-no-df-RP-Final-kappa} can be also formulated independently of this context. Specifically, it may arise from the broader problem of balancing portfolio losses with their associated downside risk. Indeed, model \eqref{model-no-df-RP-Final-kappa} can be directly interpreted from a risk management perspective by treating $\kappa > 0$ as a weight assigned to downside risk. In this interpretation, $\kappa$ is not necessarily linked to the reward rate $\beta$ or the penalty rate $\lambda$. Indeed,  it can be viewed as an endogenous parameter of the model.

Within the reward-penalty framework, the parameter $-d$ serves as a threshold that determines eligibility for a bonus payment. In contrast, in risk management, the parameter $d$ typically represents a risk metric such as the VaR or CVaR of the portfolio loss, 
which is a   
 quantity   of concern to the investor. The investor focuses  on both the downside risk, represented by $(\boldsymbol{\omega}^{\top}\boldsymbol{X}^G - d)_{+}$, and the overall portfolio loss $\boldsymbol{\omega}^{\top}\boldsymbol{X}^G$.
 }

 In most portfolio selection models, the objective is to minimize either a risk measure of the portfolio loss $\boldsymbol{\omega}^{\top}\boldsymbol{X}$
 or a risk measure of the downside risk $(\boldsymbol{\omega}^{\top}\boldsymbol{X} - d)_+$.
By employing model \eqref{model-no-df-RP-Final-kappa}, both the portfolio loss and its associated downside risk can be simultaneously considered in portfolio selection. 
As illustrated by the numerical experiments using real data in Section \ref{sec-num-exam}, portfolios constructed using model \eqref{model-no-df-RP-Final-kappa} have the potential to yield better performance compared to those based solely on model \eqref{model-no-df}.

% or the downside risk model above.

In this paper, we investigate the robust portfolio selection problem \eqref{model-no-df-RP-Final}, focusing on the CVaR risk measure and mean-covariance-based uncertainty sets. In Section \ref{sec-risk-set}, we introduce the mean-covariance-based uncertainty sets employed in this paper. In Section \ref{sec-WC-TVaR}, we derive the explicit expression for the worst-case CVaR under the reward-penalty mechanism. This expression recovers several known results regarding worst-case CVaR, such as those  in \cite{jagannathan1977minimax}, \cite{chen2011tight}, and \cite{cai2024worst}, and plays a crucial role in solving the robust portfolio selection problem \eqref{model-no-df-RP-Final}. In Section \ref{sec-optimal-selection}, we derive the explicit expression for the optimal allocation $\boldsymbol{\omega}^*$ in problem \eqref{model-no-df-RP-Final} when the uncertainty set for risky assets consists of distributions with a given mean vector and covariance matrix. 
In Section \ref{sec-num-exam}, we use real market data to compare investment strategies based on the proposed models and the optimal allocation $\boldsymbol{\omega}^*$
  derived from these models with those based on other optimal selection models related to ours, such as the classical mean-variance model and the models proposed in \cite{ZF09} and \cite{chen2011tight}. {\color{black} Concluding remarks are provided in Section \ref{sec-con}, and proofs of all results derived  in this paper are included in the appendix.   
    
  Throughout this  paper, we use $\mathds{1}_A$ to denote   the indicator function of a set or event $A$,  and $a\triangleq b$ to mean that $a$ is defined as $b$.  }

\section{Preliminaries} \label{sec-risk-set}

Two key components of the  robust  portfolio selection model \eqref{model-no-df-RP-Final}   are the  risk measure $\rho$ and the  uncertainty set ${\cal S}$. 
   In real life, the  distribution  of a loss random  vector $\boldsymbol{X}$ is often unknown,  and typically only partial information about the vector is available. Consequently, adopting a distributionally robust model that  utilizes   the available information becomes essential for addressing this uncertainty. 
  The following uncertainty set is commonly used in the literature to indicate that only the mean and covariance of $\boldsymbol{X}$ are available: 
\begin{equation}
\label{MV uncertainty set}
    \mathcal{M}(\boldsymbol{\mu},\boldsymbol{\Sigma})=\left\{G\in\mathcal{F}(\mathbb{R}^n): \ \mathbb{E}[\boldsymbol{X}^G]=\boldsymbol{\mu}, \ {\rm Cov}[\boldsymbol{X}^G]=\boldsymbol{\Sigma}\right\},
\end{equation}
where  $\boldsymbol{\mu} \in \mathbb{R}^n$, $\boldsymbol{\Sigma}$ is an $n\times n$ positive definite matrix, and $\mathcal{F}(\mathbb{R}^n)$ is the set of all $n$-dimensional distributions defined on $\mathbb{R}^n$.

In the uncertainty set   $ \mathcal{M}(\boldsymbol{\mu},\boldsymbol{\Sigma})$, the mean vector $\boldsymbol{\mu}$ and the covariance matrix $\boldsymbol{\Sigma} $ of the loss random  vector $\boldsymbol{X}$ are assumed to be  known, typically as their estimates   based on observed data. However, due to estimation errors, these estimates  may deviate from the true mean vector and covariance matrix of  $\boldsymbol{X}$. To address this issue,  generalized versions of $\mathcal{M}(\boldsymbol{\mu},\boldsymbol{\Sigma})$ have  been proposed in the literature. One such generalization was  presented in \citet{kang2019data} and is defined  as follows: 
\begin{equation}
\label{set-D}
    \mathbb{D}_n(\gamma_1,\gamma_2)=\left\{\begin{array}{ll}
    &\left(\mathbb{E}[\boldsymbol{X}^G]-\hat{\boldsymbol{\mu}}\right)^{\top} \hat{\boldsymbol{\Sigma}}^{-1}\left(\mathbb{E}[\boldsymbol{X}^G]-\hat{\boldsymbol{\mu}}\right) \leq \gamma_1, \\
G\in\mathcal{F}(\mathbb{R}^n): \ &\|{\rm Cov}(\boldsymbol{X}^G)-\hat{\boldsymbol{\Sigma}}\|_F \leq \gamma_2, \\
& {\rm Cov}(\boldsymbol{X}^G) \succ 0. 
\end{array}\right\},
\end{equation}
where $\gamma_1 \ge 0$ and $\gamma_2 \ge 0$ are constants. Furthermore,  let $\hat{\boldsymbol{\mu}} \in \R^n$ and 
 $\hat{\boldsymbol{\Sigma}}$ be an $n \times n$ positive definite matrix, which can be  interpreted as 
the inferred or estimated mean vector and covariance matrix of $\boldsymbol{X}$. Additionally, 
  $ \| {\rm Cov}(\boldsymbol{X}^G) -\hat{\boldsymbol{\Sigma}}\|_{F}$ represents 
 the Frobenius norm of the difference between the matrices ${\rm Cov}(\boldsymbol{X}^G) $ and $\hat{\boldsymbol{\Sigma}}$. 
  The Frobenius norm for an $n \times n$  matrix $\boldsymbol{A} = (a_{ij})_{n \times n}$  is defined as
 $ \|\boldsymbol{A}\|_{F}=\sqrt{\sum_{i=1}^n \sum_{j=1}^n  a^2_{ij}}=\sqrt{\text{tr}(\boldsymbol{A}^\top \boldsymbol{A})}.$

 Intuitively, in the uncertainty  set  $\mathbb{D}_n (\gamma_1,\gamma_2)$,  one  assumes  the mean vector  of $ \boldsymbol{X}$ lies in an ellipsoid of size $\gamma_1$ centered   
  at the inferred mean vector  $\hat{\boldsymbol{\mu}}$. Meanwhile,  the discrepancy between the covariance matrix  of 
  $\boldsymbol{X}$ 
  and its  inferred covariance matrix $\hat{\boldsymbol{\Sigma}}$ is 
bounded by a certain value $\gamma_2$ under the  Frobenius norm. These two parameters, 
 $\gamma_1$ and $\gamma_2$,  measure the level of ambiguity regarding the estimated mean vector and covariance matrix.
 For other  forms of uncertainty sets applicable to the loss vector, we refer 
    readers  to \cite{delage2010distributionally}, \cite{CLM25}, and the references therein.

      In this paper, we investigate problem \eqref{model-no-df-RP-Final}, or equivalently problem \eqref{model-no-df-RP-Final-kappa}, by  focusing  on the   commonly used   moment-based uncertainty 
       sets
      $ \mathcal{M}(\boldsymbol{\mu},\boldsymbol{\Sigma}) $ and  $ \mathbb{D}_n(\gamma_1,\gamma_2)$, as defined in \eqref{MV uncertainty set} and \eqref{set-D}.
     
The key to solving problem \eqref{model-no-df-RP-Final} lies in  solving the inner optimization problem of \eqref{model-no-df-RP-Final}, or equivalently, in determining the following  worst-case risk measure:  
\by
%\label{sup-risk-RP}
 \sup_{G \in \cal S}  \rho\big ((1-\beta)\boldsymbol{\omega}^{\top}\boldsymbol{X}^G+(\beta-\lambda)(\boldsymbol{\omega}^{\top}\boldsymbol{X}^G-d)_{+}\big).
 \ey

When ${\cal S}= \mathcal{M}(\boldsymbol{\mu},\boldsymbol{\Sigma})$, by applying Lemma \ref{le:A=B}, we establish the following equivalence, which relates the worst-case risk measure  of the portfolio loss  under the reward-penalty mechanism over  the uncertainty 
set  $\mathcal{M}(\boldsymbol{\mu},\boldsymbol{\Sigma})$, characterized by $n$-dimensional distributions,  to that   over  the uncertainty set 
$\mathcal{S}(\boldsymbol{\omega}^{\top}\boldsymbol{\mu}, \ \boldsymbol{\omega}^{\top}\boldsymbol{\Sigma}\boldsymbol{\omega})$, which consists of   univariate distributions:
\be
\label{Sup-risk-m=s,t}
    &&  \sup_{G \in \mathcal{M} (\boldsymbol{\mu},\boldsymbol{\Sigma})} \rho \big ((1-\beta)\boldsymbol{\omega}^{\top}\boldsymbol{X}^G+(\beta-\lambda)(\boldsymbol{\omega}^{\top}\boldsymbol{X}^G-d)_{+}\big ) \notag \\
     &=& \sup_{F \in \mathcal{S}(\boldsymbol{\omega}^{\top}\boldsymbol{\mu}, \ \boldsymbol{\omega}^{\top}\boldsymbol{\Sigma}\boldsymbol{\omega})} \rho \big (
     (1-\beta)X^F+(\beta-\lambda)(X^F-d)_{+}     
     \big ),
     \ee
where, throughout this paper, for any $\mu \in \R$ and $\sigma >0$, the (univariate) uncertainty set $ \mathcal{S}(\mu, \sigma)$  is defined as the set of all (univariate) distributions with mean $\mu$ and variance $\sigma^2$:
\begin{equation*}
    \begin{aligned}
        \mathcal{S}(\mu, \sigma) 
        &=\left\{F\in \mathcal{F}(\mathbb{R}): \ \mathbb{E}[X^F]=\mu, \ \mathbb{E}[(X^F)^2]=\mu^2+\sigma^2 \right\}.
    \end{aligned}
\end{equation*}

 The key to solving problem \eqref{Sup-risk-m=s,t} is to derive 
 the  expression for the following optimization  problem: 
\be
\label{WC-risk-model}
 \sup_{F \in \mathcal{S}(\mu, \, \sigma)} \rho\left(
     (1-\beta)X^F+(\beta-\lambda)(X^F-d)_{+}     
     \right)=
\sup_{F\in\mathcal{S}(\mu,\sigma)} \rho(\ell_d(X^F)).
  \ee
  Throughout this paper,  the  function $\ell_d(x)$ is defined for any $x \in \R$ as
\be
\label{ell}
    \ell_d(x) &= &(1-\beta) x+(\beta-\lambda)(x-d)_{+} \label{ldx1} 
    \ = \ (1-\beta)\min\left\{x,d\right\}+(1-\lambda)(x-d)_{+}. 
    \ee
This expression for problem  \eqref{WC-risk-model}  is the key to solving the robust portfolio selection problem  \eqref{model-no-df-RP-Final} 
 when 
${\cal S} =\mathcal{M}(\boldsymbol{\mu},\boldsymbol{\Sigma})$ or  ${\cal S}= \mathbb{D}_n(\gamma_1,\gamma_2)$.
Problem \eqref{WC-risk-model} itself is also of interest, as it represents the worst-case risk measure  of a loss random variable $X$,  weighted by its associated downside risk $(X-d)_+$. 

Two popular risk measures used in robust portfolio selection problems are  VaR and CVaR.   See, for instance,  \cite{ghaoui2003worst}, \cite{ZF09}, and references therein. 
For any loss random variable $X$ with distribution function  $F(x)=\Pr\left\{X\leq x\right\}$, the value-at-risk (VaR) of $X$ 
at a confidence level $\alpha \in (0,1)$ is defined as ${\rm VaR}_{\alpha}(X)=\inf\left\{x: F(x)\geq \alpha \right\},$ 
which represents the largest possible loss at this  confidence level.  
The  conditional  value-at-risk (CVaR)  or the tail value-at-risk (TVaR) of $X$ at a confidence level $\alpha \in (0,1)$ is defined as
$
{\rm CVaR}_{\alpha}(X) = \frac{1}{1-\alpha}\int_{\alpha}^{1}{\rm VaR}_{q}(X) dq, 
$
which represents  the average  magnitude of potential losses that exceed the  $\text{VaR}$ at the $\alpha$-confidence level.

When $\rho$ is  VaR, problem \eqref{WC-risk-model}  is reduced to 
 the following  worst-case VaR problem:
    \begin{equation}
\label{WC-VaR-model}
\sup_{F\in\mathcal{S}(\mu,\sigma)}{\rm VaR}_{\alpha}\left((1-\beta) X^F+(\beta-\lambda)(X^F-d)_{+}\right) =
\sup_{F\in\mathcal{S}(\mu,\sigma)}{\rm VaR}_{\alpha}^F(\ell_d(X^F)).
\end{equation}  
Note that for $ 0 \le  \beta \le   1$, and $0 \le \lambda < 1$,
   the function   $\ell_d(x)$ is continuous and increasing in  $x\in\mathbb{R}$.
 Therefore, 
    \by
 \sup_{F\in\mathcal{S}(\mu,\sigma)}{\rm VaR}_{\alpha}((1-\beta) X^F+(\beta-\lambda)(X^F-d)_{+})    = \sup_{F\in\mathcal{S}(\mu,\sigma)}\ell_d({\rm VaR}_{\alpha}(X^F))
       = \ell_d\Big (\sup_{F\in\mathcal{S}(\mu,\sigma)}{\rm VaR}_{\alpha}(X^F) \Big ), 
      \ey 
which  arises from the invariance property of VaR under continuous and increasing transformations.  
 A notable  special case  of problem  \eqref{WC-VaR-model}  is given by           
\cite{ghaoui2003worst}  who derived the explicit expressions for problem  \eqref{WC-VaR-model} in the case where
 $\beta=0$ and $\beta=\lambda$. That is,
  \be
  \label{VaR-model}  \sup_{F\in \mathcal{S}(\mu, \sigma)} {\rm VaR}_{\alpha}(X^F) = \mu + \sigma \,  \sqrt{\frac{\alpha}{1-\alpha}}.
\ee
 Thus, we have
 \by
      \sup_{F\in\mathcal{S}(\mu,\sigma)}{\rm VaR}_{\alpha}((1-\beta) X^F+(\beta-\lambda)(X^F-d)_{+})  
       &=&\ell_d\Big (\mu+\sigma\sqrt{\frac{\alpha}{1-\alpha}} \, \Big ).  \label{WC-VaR-RP}
      \ey

In the rest of this paper,  we focus on  problems \eqref{WC-risk-model} and \eqref{model-no-df-RP-Final} or its equivalent form 
\eqref{model-no-df-RP-Final-kappa} when $\rho$ is taken as  CVaR.

\section{Worst-case CVaR  under reward-penalty mechanism}\label{sec-WC-TVaR}

In this section, we investigate the worst-case CVaR of the portfolio loss  under reward-penalty mechanism and
 consider the following optimization problem: 
\begin{equation}
\label{TVaR-stop-loss}
\sup_{F\in\mathcal{S}(\mu,\sigma)}{\rm CVaR}_{\alpha}\left((1-\beta) X^F+(\beta-\lambda)(X^F-d)_{+}\right)=
\sup_{F\in\mathcal{S}(\mu,\sigma)}{\rm CVaR}_{\alpha}(\ell_d(X^F)).
\end{equation}
 Furthermore, in this section,  we assume that   
\be
\label{0-b-l-1}
0 \le  \beta \le   1, \ \ \ \ \  0 \le \lambda < 1.
\ee
The assumption in  \eqref{0-b-l-1} reflects the fact that, in practice,  the reward and penalty rates are  typically  between 0 and 1. Additionally, the condition $\lambda <1$ in \eqref{0-b-l-1}  avoids the trivial case where  both  $\beta=1$ and $\lambda=1$ hold, which would result in $\ell_d(x)=0$ for all $x \in \R$.     
 A notable  special case  of problem   \eqref{TVaR-stop-loss} with applications  to  optimal robust portfolio selection problems  has been explored in the literature.          
\cite{chen2011tight} derived the explicit expressions for problem  \eqref{TVaR-stop-loss} in the case where
 $\beta=0$ and $\beta=\lambda$. Their results show  that 
  \be
  \label{CVaR-model}
\sup_{F\in \mathcal{S}(\mu, \sigma)} {\rm CVaR}_{\alpha}(X^F) = \mu + \sigma \,  \sqrt{\frac{\alpha}{1-\alpha}}.
\ee

{\color{black} It is a well-known result that the worst-case values of VaR and CVaR over the set of all distributions with a given mean and variance are equal, 
as presented in 
\eqref{VaR-model}  and \eqref{CVaR-model}. More specifically, the worst-case value of VaR in \eqref{VaR-model}    is given in \cite{ghaoui2003worst}, while the worst-case value of CVaR in \eqref{CVaR-model}  is provided by  \cite{chen2011tight}. This equivalence is further discussed in \cite{PWW13}.

}

 Recently,     \citet{cai2024worst}
  studied an optimal  robust reinsurance problem and derived the explicit expression  for the worst-case CVaR of the 
   downside risk  $(X-d)_{+}$, or  stop-loss random variable, over the uncertainty set $\mathcal{S}(\mu,\sigma) $ for $d \ge 0$.
    Theorem 5.1 of \citet{cai2024worst} shows  that  
 \begin{equation}
 \label{WC-TVaR-CLY}
    \begin{aligned}
      &\sup_{F\in\mathcal{S}(\mu,\sigma)} {\rm CVaR}_{\alpha}\big((X^F-d)_{+}\big)
      =\begin{cases}
          \mu+\sigma\sqrt{\frac{\alpha}{1-\alpha}}, \quad &d\leq \mu-\sigma\frac{1-2\alpha}{2\sqrt{\alpha(1-\alpha)}},\\
          \frac{1}{2(1-\alpha)}\big (\mu-d+\sqrt{(\mu-d)^2+\sigma^2} \, \big), \quad &d>\mu-\sigma\frac{1-2\alpha}{2\sqrt{\alpha(1-\alpha)}},
      \end{cases}
    \end{aligned}
\end{equation}
for any $d\geq0$ and  $\alpha\in(0,1)$. 
 
         To solve problem  \eqref{TVaR-stop-loss},   we use the following   notation to denote the probability (mass) function 
         for   a
$k$-point (discrete) random variable.
 \begin{definition} 
 \label{de-k-ds}
 Let $
 [x_1,p_1;\ldots; x_k,p_k]
$
denote the probability function of   a $k$-point   random variable $X$,  where $k \ge 2$, 
   $0 \le \p(X=x_i)=p_i  <1$ for  
  $i=1,\ldots,k$, with    $\sum_{i=1}^k p_i=1$. Furthermore,    there exists  at least a pair $(i, j)$ such that $0 < p_i <1$ and 
$0 < p_j <1$ with  $1 \le i < j \le k$. This  implies that a
$k$-point distribution (with $k \ge 3$) may also  be considered as 
 $l$-point distribution in this paper, where $2 \le l  \le  k$.   Conversely, a $k$-point probability function  $[x_1,p_1;\ldots; x_k,p_k]$ 
  can be viewed as a $(k+1)$-point  probability function $[x_1,p_1;\ldots; x_k,p_k;x_{k+1}, 0]$.  
 \hfill $\Box$
 \end{definition} 

Our approach for solving problem \eqref{TVaR-stop-loss} involves reformulating this infinite-dimensional optimization into a finite-dimensional maximization problem. To achieve this, we must address the following maximization problem \eqref{eq:0401-1}, which is detailed in the following lemma:

\begin{lemma} \label{le-3.1}
For any $ 0 < \alpha < 1$, $\mu \in \R$, $\sigma >0$, $d \in \R$,   $0 \le \lambda  < 1$,   and $0 \le  \beta <  1$.  The maximum value of the objective function  in the  following  four-variable constrained maximization  problem: 
\begin{align}
\max_{(x_1, x_2, x_3, s) \in \R^ 3 \times [\alpha, 1]} ~~~&\, \frac{s-\alpha}{1-\alpha}(1- \beta) x_2 + \frac{1-s}{1-\alpha} \left((1-\beta) d + (1-\lambda)(x_3-d) \right) \label{eq:0401-1}\\
{\rm s.t.}~~~& \alpha x_1+ (s-\alpha) x_2 + (1-s) x_3=\mu,\notag\\
& \alpha x_1^2+ (s-\alpha) x_2 ^2+ (1-s) x_3^2=\mu^2+\sigma^2,\notag \\
& x_1\le x_2\le d\le x_3 
\notag 
\end{align}
is equal to the  maximum value of the objective function in the  following one-variable unconstrained maximization 
  problem: 
\be
\label{max-f}
\max_{ s \in [\alpha,1]} \, f(s),
\ee
where the function $f(s)$ for $\alpha \le s \le 1$  is defined  by 
\by
f(s) =  \frac{1-s}{1-\alpha} (\lambda-\beta) d+ \frac{1}{1-\alpha} \left(   f^*(s)+\sigma\sqrt{  (s-\alpha) (1-\beta)^2 + (1-s)(1-\lambda)^2 -(c^*(s))^2 }\right),
\ey
and the functions $f^*(s)$ and $c^*(s)$  are defined as  $f^*(s)=\mu [(1-\beta)(s-\alpha) + (1-\lambda)(1-s)]$ and 
$c^*(s)=(1-\beta)(s-\alpha) + (1-\lambda)(1-s)$.   \end{lemma}
 
To present solutions for problem  \eqref{TVaR-stop-loss},  in this paper,  for $0 <\alpha<1$ and $0 \le  \lambda  < \beta \le  1$,  we define the  parameters $d_1$
and $d_2$
 as \begin{equation}
    \label{d1d2}
        d_1=\frac{1-\beta+(1-\lambda)(2\alpha-1)}{2(1-\lambda)\sqrt{\alpha(1-\alpha)}}, \quad \  \  d_2=\frac{1-\lambda+(1-\beta)(2\alpha-1)}{2(1-\beta)\sqrt{\alpha(1-\alpha)}}. 
    \end{equation} 
Note that  $d_1 \in (-\infty, \infty)$, $d_1 < d_2$, and  $d_2 \in (0, \infty]$. Additionally,   $d_2=\infty$ when $\beta=1$.  Furthermore,      denote the function $f_{\lambda, \beta}(s, t)$ for $s \in \R$ and $t > 0$ as
    \be
    \label{TVaR-f}
f_{\lambda, \beta}(s, t)=\sup_{F\in\mathcal{S}(s, t)}{\rm CVaR}_{\alpha}((1-\beta) X^F+(\beta-\lambda)(X^F-d)_{+}).
\ee

\begin{theorem} 
\label{Prop-wcTVaR-institutional}  For  $0 \le  \beta \le 1$ and $0 \le \lambda <1$,    the worst-case {\rm CVaR} of $(1-\beta) X+(\beta-\lambda)(X-d)_{+}$ at confidence 
level $\alpha \in (0, 1)$ and over the uncertainty set 
$\mathcal{S}(\mu, \sigma)$  or the function 
\be
\label{f-wcTVaR}
 f_{\lambda, \beta}(\mu, \sigma)=\sup_{F\in\mathcal{S}(\mu, \sigma)}{\rm CVaR}_{\alpha}((1-\beta) X^F+(\beta-\lambda)(X^F-d)_{+})
  \ee
   has the following expressions:  
        \begin{itemize}
\item [(a)]   Assume  $ \lambda \ge  \beta $. Then, 
\be
\label{ga}
  f_{\lambda, \beta}(\mu, \sigma) 
  =    \ell_d\Big (\mu+\sigma\sqrt{\frac \alpha{1-\alpha} }  \, \Big ),
\ee
where $\ell_d$ is defined in \eqref{ell}.

\item [(b)] Assume   $ \lambda< \beta$. 
\begin{itemize}
\item[(b1)] If $d-\mu \le \sigma d_1$,  then 
\by
%\label{gb1}
  f_{\lambda, \beta}(\mu, \sigma)
  =(\lambda-\beta)d+	(1-\lambda)\Big  (\mu+\sigma \, \sqrt{\frac{\alpha}{1-\alpha}} \,  \Big ).
\ey

 \item[(b2)]  If $d-\mu \in   (\sigma  d_1, \, \sigma d_2 )$, then 
 \be
\label{gb2} 
f_{\lambda, \beta}(\mu, \sigma)  
=(1-\beta) d+\frac{\Delta(\mu, \sigma)-(2-\lambda-\beta)(d-\mu)}{2(1-\alpha)},
\ee
where  $\Delta(s,t)$ is a function defined for any $s \in \R$ and $t>0$ as
\be
\label{Delta}
\Delta(s, t)=\sqrt{(t^2+(d-s)^2)((\beta-\lambda)^2+4\alpha(1-\beta)(1-\lambda))}.
 \ee
 
\item[(b3)]   If $d-\mu \ge   \sigma  d_2$,  then
\by
  %   \label{gb3} 
f_{\lambda, \beta}(\mu, \sigma) 
=(1-\beta)\Big  (\mu+\sigma \, \sqrt{\frac{\alpha}{1-\alpha}} \, \Big ).
\ey
\end{itemize}
\end{itemize}
\end{theorem}

\begin{remark}
\bn[(i)]

\item  When  $\lambda=0$ and $\beta = 0$,   Theorem  \ref{Prop-wcTVaR-institutional}   recovers the 
 well-known worst-case CVaR expression \eqref{CVaR-model} derived in  \cite{chen2011tight}.
    Furthermore, our proof differs from that used in \cite{chen2011tight} and allows us to derive explicit expressions for problem \eqref{TVaR-stop-loss}.    

\item We point out that the result given in  \eqref{WC-TVaR-CLY} obtained by  \citet{cai2024worst} is a special case of Theorem  \ref{Prop-wcTVaR-institutional}  when 
 $\lambda=0$ and $\beta = 1$. More importantly,   Theorem  \ref{Prop-wcTVaR-institutional} allows $d$ to be any real number. Thus, the result  in  \eqref{WC-TVaR-CLY} holds for any $d \in \R$. 

\item  Theorem  \ref{Prop-wcTVaR-institutional} also  demonstrates   that   under the  reward-penalty mechanism, if $0 \le \beta  \le   \lambda < 1$, 
then, by \eqref{WC-VaR-RP} and \eqref{ga}, 
we have 
\by
  && \sup_{F\in\mathcal{S}(\mu,\sigma)}{\rm VaR}_{\alpha}((1-\beta) X^F+(\beta-\lambda)(X^F-d)_{+}) \\
         &=&  \sup_{F\in\mathcal{S}(\mu,\sigma)}{\rm CVaR}_{\alpha}((1-\beta) X^F+(\beta-\lambda)(X^F-d)_{+}) 
           =\ell_d\Big (\mu+\sigma\sqrt{\frac{\alpha}{1-\alpha}} \, \Big ). 
\ey
However, if $0 \le \lambda <  \beta \le   1$, the worst-case VaR of the  loss $X$ weighted by the downside risk $(X-d)_+$  differs from the corresponding  worst-case CVaR. Furthermore, in the absence of the  reward-penalty mechanism, the worst-case VaR and  the worst-case CVaR of  the loss $X$ over the uncertainty set 
$\mathcal{S}(\mu,\sigma)$  are equal,  as shown  in \eqref{CVaR-model}.    
\hfill $\Box$
 \en
\end{remark}

 Note that  $\text{CVaR}_\alpha \to  \E$ as $\alpha \to 0$.  Thus, 
by letting $\alpha \to 0$ in Theorem  \ref{Prop-wcTVaR-institutional} and noting that $d_1 \to -\infty$, 
$d_2 \to \infty$, and $\Delta(\mu,\sigma) \to |\beta-\lambda| \sqrt{\sigma^2+(d-\mu)^2}  $  as $\alpha \to 0$,
we obtain the following  proposition:

\begin{proposition} 
\label{pro-WC-TVaR-RP-a=0}
For $0 \le \lambda <1$ and $0 \le \beta \le 1$, the following results hold:  
\begin{itemize}
\item [(a)] If  $ \lambda\geq \beta $, then, 
\by
 \sup_{F\in \mathcal S (\mu, \sigma)} \E [(1-\beta) X^F+(\beta-\lambda)(X^F-d)_{+}]  =    (1-\beta) \mu+(\beta-\lambda)(\mu-d)_{+}.
\ey

\item [(b)] If  $ \lambda< \beta $, then
\by
&& \sup_{F\in \mathcal S (\mu, \sigma)} \E [(1-\beta) X^F +(\beta-\lambda)(X^F-d)_{+}]  
 \nonumber\\
&=& (1-\beta) d+\frac{(\beta-\lambda)\sqrt{\sigma^2+(d-\mu)^2}  -(2-\lambda-\beta)(d-\mu)}{2}.
\ey
\end{itemize}
%\hfill $\Box$ 
\end{proposition}

\begin{remark}
\bn[(i)]
\item When $\lambda=0$ and $\beta = 1$, Proposition \ref{pro-WC-TVaR-RP-a=0}(b) recovers the well-known worst-case stop-loss premium:
\by
 \sup_{F\in \mathcal S (\mu, \sigma)} \E [(X^F-d)_{+}]   = \frac{\mu-d + \sqrt{\sigma^2+(d-\mu)^2}}{2},
\ey  
which was first proved in \citet{jagannathan1977minimax}.

\item 
For $0 \le \lambda  <    \beta < 1$,  
define  the  function $f_\kappa(s, t)$ for any $s \in \R$  and $t>0$  as the worst-case CVaR of  
$X+\kappa(X-d)_{+}$ at confidence level $\alpha$ and  over the uncertainty set $\mathcal{S}(s, t)$, namely, 
\by
%\label{fk}
f_\kappa(s, t) =\sup_{F \in \mathcal{S}(s, t)}{\rm CVaR}_\alpha (X^F+\kappa(X^F-d)_{+} ).
\ey
Then, the worst-case CVaR 
\by
%\label{f_kappa}
f_\kappa(\mu, \sigma) = \sup_{F \in \mathcal{S}(\mu, \sigma)}{\rm CVaR}_\alpha (X^F+\kappa(X^F-d)_{+}) = \frac{1}{1-\beta} \, f_{\lambda, \beta}(\mu,\, \sigma), 
\ey
where $ f_{\lambda, \beta}(\mu,\, \sigma) =
 \sup_{F\in \mathcal S (\mu, \sigma)} {\rm CVaR}_\alpha ((1-\beta) X^F+(\beta-\lambda)(X^F-d)_{+})
     $ is given in Theorem  \ref{Prop-wcTVaR-institutional}(b).
  \en
\end{remark} 
To end this section,  we give a proposition about the convexity of the function $f_{\lambda, \beta}(\mu,\sigma)$ which 
 will be used to determine the convexity of objective functions in optimal robust portfolio problems, as discussed in the next two sections. 
\begin{proposition}  
\label{f-convex}
For any $0 \le \lambda < 1$ and $0 \le \beta \le 1$, the function $f_{\lambda, \beta}(\mu,\sigma)$ defined by \eqref{TVaR-f} is convex in $(\mu,\sigma)\in\R\times\R^+$. 
 \end{proposition} 

{\color{black} At the end of this section, we point out that there are multiple approaches to model distributional
uncertainty. In our paper, we adopt moment-based uncertainty sets, 
  primarily due to  their tractability and broad applicability. As illustrated by the references cited therein,  
   moment-based sets (and related variants) are widely used in distributionally robust optimization (DRO), particularly in insurance,  finance, and operations research. Their practical relevance is further supported by \cite{RVB24}, which provides an extensive discussion of moment-based DRO formulations.   We also point out that  divergence-based uncertainty sets,  such as those defined via Wasserstein balls,    are also important 
for  modelling distributional uncertainty.  For recent developments, we refer to \cite{B23} and \cite{TY23}, who introduce 
 new uncertainty sets based on Wasserstein balls. 
 
 It is possible to extend our framework to  divergence-based uncertainty sets. For instance,   consider 
     the  worst--case  risk-measure of a loss function   $\ell (X)$ over  a Wasserstein ball:  
                      \begin{equation} \label{WC-g-l-d-M}
  \sup_{F \in  {\cal S }(\hat{F}; k, \varepsilon)} \rho(\ell(X^F)),
  \end{equation}    
  where  $ {\cal S} ( \hat{F};  k, \varepsilon ) =   \{F:  \, W_k ( F, \hat{F} ) \leq \varepsilon \}$
       is  the $k$-th order Wasserstein ball centered at a reference distribution $\hat{F}$. Here,   $ k \geq 1$, $ \varepsilon \geq 0$, 
  and $W_k$ is the Wasserstein distance with order $k$. 
                 
  Problem \eqref{WC-g-l-d-M}  has been addressed in \cite{CLY25}  when   $\rho$ is a distortion risk measure and        $\ell(x)$ takes  one of 
   the following  three forms: i)  
  $ \ell(x) = (x-d)_+$;  ii)  
              $\ell (x) = \min\{x, \, m\}$;  iii) 
            $  \ell(x) = \min\{(x-d)_+, \,  m\}$.  
   In our model setting,   similar to problem \eqref{Sup-risk-m=s,t}, 
     replacing the moment-based uncertainty sets with the Wasserstein ball  would require solving the following  problem:           
                  \by
&&
 \sup_{F \in  {\cal S }(\hat{F}; k, \varepsilon)} \rho \left(
     (1-\beta)X^F+ (\beta-\lambda) (X^F-d)_{+}     
     \right) \nonumber \\
     &=&  \sup_{F \in  {\cal S }(\hat{F}; k, \varepsilon)} \rho \left(
     (1-\beta)\min\{X^F, \, d\} + (1-\lambda) (X^F-d)_{+} \right). 
   \ey
This formulation differs from those considered in \cite{CLY25}, and their techniques cannot be directly applied. Specifically, their approach relies on transforming the risk measure 
 of $\ell(X)$ into an equivalent form involving a new random variable $X^*$, while our framework employs dimension reduction techniques that convert infinite-dimensional DRO problems into finite-dimensional optimization programs. Therefore, extending our model to accommodate  divergence-based uncertainty sets 
 would require the development of entirely new analytical tools, which is beyond the scope of the current paper. Nevertheless,  this is a compelling direction for future research. 
}

\section{Portfolio selection under   a 
  mean-covariance-based uncertainty set}\label{sec-optimal-selection} 

In this section, we consider the robust portfolio selection problem that minimizes the worst-case CVaR of the portfolio loss under the reward-penalty mechanism. The problem is formulated as follows:
\begin{equation}
\label{portfolio CVaR}
    \begin{aligned}
        &\min_{\boldsymbol{\omega}\in\mathbb{R}^n}\sup_{G \in \mathcal{M}(\boldsymbol{\mu},\boldsymbol{\Sigma})} {\rm CVaR}_{\alpha}\big((1-\beta)\boldsymbol{\omega}^{\top}\boldsymbol{X}^G+(\beta-\lambda)(\boldsymbol{\omega}^{\top}\boldsymbol{X}^G-d)_{+}\big)\\
        & \ {\rm s.t.} \quad \boldsymbol{\omega}^{\top}\boldsymbol{e}=1.
    \end{aligned}
\end{equation}

In practice, the penalty rate $\lambda$ is typically less than the reward rate $\beta$, which, together with the conditions in \eqref{0-b-l-1},  leads to  
 the condition 
 \be
 \label{l<b}
 0 \le \lambda <  \beta < 1.
 \ee
 Under this condition, problem  \eqref{portfolio CVaR} is equivalent to the following problem:
  \begin{equation}
\label{portfolio CVaR-kappa}
    \begin{aligned}
        &\min_{\boldsymbol{\omega}\in\mathbb{R}^n}\sup_{G \in \mathcal{M}(\boldsymbol{\mu},\boldsymbol{\Sigma})} {\rm CVaR}_{\alpha}\big(\boldsymbol{\omega}^{\top}\boldsymbol{X}^G+\kappa (\boldsymbol{\omega}^{\top}\boldsymbol{X}^G-d)_{+}\big)\\
        & \ \text{s.t.} \quad \boldsymbol{\omega}^{\top}\boldsymbol{e}=1,
    \end{aligned}
\end{equation}  
 where $\kappa >0$ is defined in \eqref{kappa}.

 We assume that the condition \eqref{l<b} holds in this section. In addition,  following
 \citet{constantinides1995portfolio}, we assume that the covariance matrix $\boldsymbol{\Sigma}$ of $ \boldsymbol{X}$ is positive definite and  
  that $\boldsymbol{\mu}$ and $\boldsymbol{e}$ are not linearly   dependent, namely, for any  $c \in \R$,  
  $\boldsymbol{\mu} \not= c \boldsymbol{e}$. Further, for notational brevity, in this section, as in \cite{chen2011tight},  we define the parameters $c_0$, $c_1$, $c_2$, $b_0$, $b_1$, and $b_2$ as follows:
  \by
 % \label{def-b123}
\textit{c}_{0} = \boldsymbol{e}^\mathrm{\top}\boldsymbol{\Sigma}^{-1}\boldsymbol{e},\,\,  \textit{c}_{1} = \boldsymbol{e}^\mathrm{\top}\boldsymbol{\Sigma}^{-1}\boldsymbol{\mu}, \,\,  \textit{c}_{2} = \boldsymbol{\mu}^\mathrm{\top}\boldsymbol{\Sigma}^{-1}\boldsymbol{\mu}, \,\,
\textit{b}_{0} = \frac{\textit{c}_{0}}{\textit{c}_{0}\textit{c}_{2}-\textit{c}_{1}^2},  \,\,  \textit{b}_{1} = \frac{\textit{c}_{1}}{\textit{c}_{0}\textit{c}_{2}-\textit{c}_{1}^2}, \,\,    \textit{b}_{2} = \frac{\textit{c}_{2}}{\textit{c}_{0}\textit{c}_{2}-\textit{c}_{1}^2}.
\ey
It holds that $c_0c_2-c_1^2>0$, $b_0 >0$, $b_2>0$, and $b_0b_2-b_1^2 >0$   since $\boldsymbol{\Sigma} \succ 0$ and the assumption that $\boldsymbol{\mu}$ and $\boldsymbol{e}$ are not linearly   dependent.

\begin{lemma} 
 \label{sol-MV-model}
The optimal allocation $\boldsymbol{\omega}^*$ to the  mean-variance  
 portfolio selection problem of 
     \begin{equation}
    \label{model-MV}
        \begin{aligned}
            &\min_{\boldsymbol{\omega}\in\mathbb{R}^n} \boldsymbol{\omega}^{\top}\boldsymbol{\Sigma}\boldsymbol{\omega},\\
            & \ {\rm s.t.}    \quad \ \boldsymbol{\omega}^{\top}\boldsymbol{e}=1, \ \ \   \quad  -\boldsymbol{\omega}^{\top}\boldsymbol{\mu}= r_0, 
                    \end{aligned}
    \end{equation}
 is unique and given by 
 \by
 %  \label{w-M-V-model}
    \boldsymbol{\omega}^* =         \begin{pmatrix}
\boldsymbol{\Sigma}^{-1}\boldsymbol{\mu} & \boldsymbol{\Sigma}^{-1}\boldsymbol{e}
        \end{pmatrix}\begin{pmatrix}
            b_0 & -b_1\\-b_1 & b_2
        \end{pmatrix}\begin{pmatrix}
           -r_0 \\ 1
        \end{pmatrix}. \label{w-solution-MV-model}
   \ey
\end{lemma}
\begin{proof} This is a well-known result for the classical mean-variance  model \eqref{model-MV}, and its proof follows from the application of the Lagrangian method. See, for example, \citet{constantinides1995portfolio}  for details.
\end{proof}

\begin{lemma} 
\label{le-for-g(s)}
For $\alpha \in (0,1)$ and $0 \le \lambda < \beta < 1$,  define 
\be
\label{gs-TVaR-v2}
   g_{\lambda, \beta}(s)  \triangleq   f_{\lambda, \beta}(s,\sqrt{b_0s^2-2b_1s+b_2}),\quad  s \in \R,
\ee
where the function $f_{\lambda, \beta}(s, t)$ is defined   in \eqref{TVaR-f}  and solved in  Theorem 
\ref{Prop-wcTVaR-institutional}(b). 
Then $g_{\lambda, \beta}(s)$ is  strictly convex in $s \in \R$. Furthermore, if 
 the following   condition  holds: 
\be
\label{cons-for-g}
b_0+1 > b^* \triangleq  \max \Big \{\frac{1}{\alpha}, \, \ \frac{(2-\lambda-\beta)^2}{(\beta-\lambda)^2+4\alpha(1-\beta)(1-\lambda)} \Big \},
\ee
 then $g_{\lambda, \beta}(s)$ is also coercive in $s \in \R$ and thus $g_{\lambda, \beta}(s)$ has a unique minimizer on $\R$.      
    \end{lemma} 

{\color{black} Based on the proof of Lemma \ref {le-for-g(s)},      the condition in \eqref{cons-for-g}   is a technical requirement that ensures the function $g_{\lambda, \beta} (s)$  has a unique minimizer.  We also point out that  inequality  \eqref{cons-for-g}   is typically satisfied 
when the confidence level   $\alpha$ is high and the penalty rate  $\lambda$ is small.
 In fact, as $\alpha \to 1$ and $\lambda \to 0$,
 we have that $b^* \to \max \big \{1, \frac{(2-\beta)^2}{\beta^2+4(1-\beta)} \big \} =1$. Note that $b_0 > 0$. In practice, the value of $\alpha$ used in CVaR is typically large,  with common choices of 0.90. 0.95, or 0.99. The penalty rate $ \lambda $ is typically small and may even be 
  zero, corresponding to a situation without penalties.  Therefore,  the  condition in \eqref{cons-for-g}  is mild 
   and  is  generally satisfied in  typical  practical situation.
   }

\begin{theorem}\label{CVaR-MV-PR}
 For   $\alpha\in(0,1)$  and  $0 \le \lambda <   \beta < 1$, if  the condition in \eqref {cons-for-g} holds, then   the optimal allocation 
      $\boldsymbol{\omega}^*$ to the robust portfolio selection problem    \eqref{portfolio CVaR}  or problem 
      \eqref{portfolio CVaR-kappa} with $\kappa=\frac{\beta-\lambda}{1-\beta}$  is unique and  given by   
      \by
        \boldsymbol{\omega}^*
          &= & \begin{pmatrix}
            \boldsymbol{\Sigma}^{-1}\boldsymbol{\mu} & \boldsymbol{\Sigma}^{-1}\boldsymbol{e}
        \end{pmatrix}\begin{pmatrix}
            b_0 & -b_1\\-b_1 & b_2
        \end{pmatrix}\begin{pmatrix}
            s^* \\ 1
        \end{pmatrix}, \label{w*-solution-MV-model-TVaR}
    \ey
  where     
    $
        s^*=    \argmin_{s\in\mathbb{R}} g_{\lambda, \beta}(s) $ is the unique minimizer of the function $g_{\lambda, \beta}(s)$ on $\R$ and   
  the function $g_{\lambda, \beta}(s)$
   is given in \eqref{gs-TVaR-v2}.
\end{theorem}

By letting $\alpha \to 0$ in model \eqref{portfolio CVaR},
in the following theorem,  we obtain the explicit solution to the following  robust 
optimization problem:  
    \begin{equation}
\label{model m=1 MV uncertainty 1d}
\begin{aligned}
    &
    \min_{\boldsymbol{\omega} \in \mathbb{R}^n}  \, \sup_{G \in \mathcal{M}(\bm{\mu}, {\bf \Sigma})} \, \mathbb{E}[\boldsymbol{\omega}^{\top}\boldsymbol{X}^G+ \kappa (\boldsymbol{\omega}^\top \bm{X}^G -d)_+],\\
    & {\rm s.t.} \quad \boldsymbol{\omega}^{\top}\boldsymbol{e}=1,
\end{aligned}
\end{equation}
where $\kappa >0$ is defined in \eqref{kappa}.   

\begin{theorem}
\label{TVaR-MV-PR-a=0}
       If  $\kappa >  \kappa^*   \triangleq  \frac{2+2\sqrt{b_0+1}}{b_0}$, then the  optimal allocation   $\boldsymbol{\omega}^*$ to  the robust portfolio selection problem    \eqref{model m=1 MV uncertainty 1d}  is unique and 
has the following analytical expression:
    \begin{equation*}
  %  \label{optimalwupm1}
        \boldsymbol{\omega}^*=\begin{pmatrix}
            \boldsymbol{\Sigma}^{-1}\boldsymbol{\mu} & \boldsymbol{\Sigma}^{-1}\boldsymbol{e}
        \end{pmatrix}\begin{pmatrix}
            b_0 & -b_1\\-b_1 & b_2
        \end{pmatrix}\begin{pmatrix}
            s^* \\ 1 
        \end{pmatrix},
    \end{equation*}
where 
\be
\label{s*}
s^*=\frac{b_1+d}{b_0+1}+\frac{(\kappa+2)\sqrt{(4(\kappa+1)-b_0\kappa^2)(b_1^2+2b_1d-b_0b_2-b_0d^2-b_2)}}{(b_0+1)(4(\kappa+1)-b_0\kappa^2)}.
 \ee
 \end{theorem}

\begin{remark}
 {\color{black} 
 As pointed out in Section \ref{sec-intro},   the parameter $\kappa$ in  problems  \eqref{model-no-df-RP-Final-kappa} and   \eqref{model m=1 MV uncertainty 1d}  
 can be  directly interpreted   
  as a  weight assigned to the  downside risk,  and it  is
  not necessarily linked to the reward rate $\beta$ or the penalty rate $\lambda$. In fact,  
   $\kappa$  can be viewed as an  endogenous parameter of the model.
        Similar to  condition  \eqref{cons-for-g}  in Lemma  \ref{le-for-g(s)} and Theorem \ref{CVaR-MV-PR}, the condition  $\kappa > \kappa^*$ in Theorem \ref{TVaR-MV-PR-a=0}  is 
 a technical requirement that ensures  problem  \eqref{model m=1 MV uncertainty 1d}   admits  a   unique optimal allocation.
   Indeed,   if $\kappa \le \kappa^*$,    it is easy to verify  that the function $g_2(s)$, defined in the proof of Theorem  \ref{TVaR-MV-PR-a=0},  is 
   strictly increasing  on $s \in \R$, since      $g^\prime_2(s) >0$ holds  for all $s \in \R$. This implies 
   that  problem  \eqref{model m=1 MV uncertainty 1d}   has no solution when  $\kappa \le \kappa^*$.  In particular, problem 
     \eqref{model m=1 MV uncertainty 1d}     has no solution when  $\kappa=0$.

 We further point out that   if $\kappa$ is interpreted within the framework of a reward-penalty mechanism,  where
$\kappa = \frac{\beta - \lambda}{1 - \beta}$ as defined in \eqref{kappa}, then it may fall below the required threshold $\kappa^*$, particularly  for  small $\beta$.  This indicates  
a limitation in the applicability of the theorem when $\kappa$ is interpreted strictly within the reward-penalty framework. 
     On the other hand, our numerical results (Figures 1-2) reveal that when a portfolio is based 
on      \eqref{model m=1 MV uncertainty 1d}, there exists  a clear trade-off between the cumulative portfolio  wealth   and the worst-case portfolio CVaR: as $\kappa$ increases, portfolio wealth 
 declines, while portfolio CVaR improves.
Thus, if $\kappa$ is interpreted as the weight assigned to downside risk, then from a risk management perspective, it should be 
chosen to be larger. In this sense, the condition $\kappa > \kappa^*$ is meaningful, as it identifies a threshold required for the 
weight $\kappa$ to effectively control downside risk. 

 }
 \end{remark} 

\begin{remark} 
{\color{black} 
It is well known that the classical mean-variance (M-V) portfolio selection model \eqref{model-MV} has the following 
alternative formulation:
\begin{equation*}
\begin{aligned}
&\min_{\boldsymbol{\omega} \in \mathbb{R}^n} \ \boldsymbol{\omega}^{\top} \boldsymbol{\mu} + \frac{\gamma}{2} \boldsymbol{\omega}^{\top} \boldsymbol{\Sigma} \boldsymbol{\omega}, \
&\text{s.t.} \quad \boldsymbol{\omega}^{\top} \boldsymbol{e} = 1,
\end{aligned}
\end{equation*}
where $\gamma > 0$. This model balances expected portfolio losses with variance (volatility); see, for example, \cite{constantinides1995portfolio}. In contrast, model \eqref{model m=1 MV uncertainty 1d} essentially balances expected portfolio loss with the worst-case  expected downside risk. This motivates a comparison between model \eqref{model m=1 MV uncertainty 1d} and the following robust mean--worst-case risk-measure  portfolio selection model:
  \begin{equation}
\label{model-risk measure}
\begin{aligned}
&\min_{\boldsymbol{\omega} \in \mathbb{R}^n} \ \boldsymbol{\omega}^{\top} \boldsymbol{\mu} + \kappa  \sup_{G \in \mathcal{M}(\boldsymbol{\mu}, \boldsymbol{\Sigma})} \rho(\boldsymbol{\omega}^{\top} \boldsymbol{X}^G), \
&\text{s.t.} \quad \boldsymbol{\omega}^{\top} \boldsymbol{e} = 1,
\end{aligned}
\end{equation}
where $\kappa > 0$, $\boldsymbol{\mu}$ is the mean vector of the loss vector $\boldsymbol{X}$, and $\rho$ is a risk measure such as VaR, CVaR, or another measure commonly used to capture right-tail risks.

In Section 6, we compare our proposed models with model \eqref{model-risk measure}, focusing on the case where $\rho$ is CVaR, along with  several related to ours.
Solving problem \eqref{model-risk measure} is technically straightforward if the worst-case value
$\sup_{G \in \mathcal{M}(\boldsymbol{\mu}, \boldsymbol{\Sigma})} \rho(\boldsymbol{\omega}^{\top} \boldsymbol{X}^G)$
is available in closed form.
We note that for any distortion risk measure $\rho$, a closed-form expression for this worst-case value has been derived in \cite{CLM25}.
}
 \end{remark}

\section{Portfolio selection under   a generalized 
  mean-covariance-based uncertainty set}\label{sec-optimal-selection-g}

In this section, we study the optimal robust portfolio selection  problem  \eqref{model-no-df-RP-Final} when the uncertainty set $\cal S$  is the generalized mean-covariance-based uncertainty set $\mathbb{D}_n(\gamma_1,\gamma_2)$ defined in \eqref{set-D} and the risk measure $\rho$ is CVaR.  The problem is formulated as:
\begin{equation}
\label{CVaR-D}
    \begin{aligned}
        \min_{\boldsymbol{\omega}\in\mathbb{R}^n}&\sup_{G \in \mathbb{D}_n(\gamma_1,\gamma_2)} 
        {\rm CVaR}_{\alpha}((1-\beta)\boldsymbol{\omega}^{\top}\boldsymbol{X}^G+(\beta-\lambda)(\boldsymbol{\omega}^{\top}\boldsymbol{X}^G-d)_{+}),\\
        {\rm s.t.} & \ \min_{G \in \mathbb{D}_n(\gamma_1,\gamma_2)}\mathbb{E}[-\boldsymbol{\omega}^{\top}\boldsymbol{X}^G]\geq r,\ \ \ \boldsymbol{\omega}^{\top}\boldsymbol{e}=1.
    \end{aligned}
\end{equation}

In this formulation, we introduce the constraint $ \min_{G \in \mathbb{D}_n(\gamma_1,\gamma_2)}\mathbb{E}[-\boldsymbol{\omega}^{\top}\boldsymbol{X}^G]\geq r.$
 This constraint 
 ensures that the worst-case expected portfolio return over the uncertainty set is at least $r$. Such a constraint is commonly employed in robust portfolio selection problems (see, for example,   \cite{blanchet2022distributionally},  \citet{kang2019data}, and references therein). Due to the complexity of the uncertainty set  $\mathbb{D}_n(\gamma_1,\gamma_2)$, explicit solutions to problem \eqref{CVaR-D} are generally unavailable, except in several special cases. However, we demonstrate that the optimal robust portfolio selection problem \eqref{CVaR-D} can be reformulated as a convex optimization problem that can be solved numerically. To this end, 
 we define the following  vector and matrix  spaces:  
\by
\label{uncertainty set in short}
      \mathcal{U}_{\hat{\boldsymbol{\mu}}} &=&  \left\{\boldsymbol{\mu}\in \mathbb{R}^n:\ (\boldsymbol{\mu}-\hat{\boldsymbol{\mu}})^{\top}\hat{\boldsymbol{\Sigma}}^{-1}(\boldsymbol{\mu}-\hat{\boldsymbol{\mu}})\leq \gamma_{1}\right\},\\
        \mathcal{U}_{\hat{\boldsymbol{\Sigma}}}  &=&  \left\{\boldsymbol{\Sigma}\in S_{+}^n: \  \ \|\boldsymbol{\Sigma}-\hat{\boldsymbol{\Sigma}}\|_{F}\leq \gamma_2\right\}, \\
  \mathcal{U}_{(\hat{\boldsymbol{\mu}},\hat{\boldsymbol{\Sigma}})}  &=&  \left\{(\boldsymbol{\mu},\boldsymbol{\Sigma})\in \mathbb{R}^n \times S_{+}^n: \,  (\boldsymbol{\mu}-\hat{\boldsymbol{\mu}})^{\top}\hat{\boldsymbol{\Sigma}}^{-1}(\boldsymbol{\mu}-\hat{\boldsymbol{\mu}})\leq \gamma_{1}, \  \|\boldsymbol{\Sigma}-\hat{\boldsymbol{\Sigma}}\|_{F}\leq \gamma_2\right\},  
 \ey
where  $S_{+}^n$ denotes the set of all $n \times n$ positive definite matrices. Note that it holds that 
$
  \mathcal{U}_{(\hat{\boldsymbol{\mu}},\hat{\boldsymbol{\Sigma}})}  
  =\mathcal{U}_{\hat{\boldsymbol{\mu}}} \times  \mathcal{U}_{\hat{\boldsymbol{\Sigma}}}.   
$

\begin{theorem}
\label{th:ROPS-TVaR}
  For    $0 \le \lambda <  \beta <  1$. The optimal allocation $\boldsymbol{\omega}^*$  to the robust portfolio selection  
   problem \eqref{CVaR-D} is the solution to  the following convex optimization problem:
      \begin{equation}
\label{TVaR-convex-expression}
    \begin{aligned}
        \min_{\boldsymbol{\omega} \in \R^n, \, s \in \R^+, \,  t \in \R^+} \ & f_{\lambda, \beta}\big (\boldsymbol{\omega}^{\top} \hat{\boldsymbol{\mu}} + \sqrt{\gamma_1} \,  t, \  s \big ), \\        
      {\rm s.t.} \ &\sqrt{\gamma_1}\sqrt{\boldsymbol{\omega}^{\top}\hat{\boldsymbol{\Sigma}}\boldsymbol{\omega}}\leq -\hat{\boldsymbol{\mu}}^{\top}\boldsymbol{\omega}- r,\\
        &{\color{black} \sqrt{\boldsymbol{\omega}^{\top}(\hat{\boldsymbol{\Sigma}}+\gamma_2 \boldsymbol{I}_n)\boldsymbol{\omega}}} \leq s,\\
        & \sqrt{\boldsymbol{\omega}^{\top}\hat{\boldsymbol{\Sigma}}\boldsymbol{\omega}}\leq t, \ \ \ \  \boldsymbol{\omega}^{\top}\boldsymbol{e}=1,
    \end{aligned}
\end{equation}    
    where the function $f_{\lambda, \beta}(s, t)$ is defined in \eqref{TVaR-f} and its expressions are given in Theorem \ref{Prop-wcTVaR-institutional}(b).  
      \end{theorem}

Let $0 \le \lambda < \beta < 1$ and  $\alpha \to 0$ in problem \eqref{CVaR-D}. In this case, problem \eqref{CVaR-D} simplifies to 
 the   following equivalent robust portfolio selection  model:
\begin{equation}
\label{E-D}
\begin{aligned}
    \min_{\bm{w} \in \mathbb{R}^n}  \, &\sup_{G \in \mathcal{\mathbb{D}}_{n}(\gamma_1,\gamma_2)} \, 
    \E[ \boldsymbol{\omega}^{\top}\boldsymbol{X}^G+ \kappa (\boldsymbol{\omega}^{\top}\boldsymbol{X}^G-d)_{+} ] \\
   \textrm{s.t.}
  &\min_{G \in \mathcal{\mathbb{D}}_{n}(\gamma_1,\gamma_2)}\mathbb{E}[-\boldsymbol{\omega}^{\top}\boldsymbol{X}^G]\geq r,\ \ \ \ \  \boldsymbol{\omega}^{\top}\boldsymbol{e}=1,
\end{aligned}
\end{equation}
where $\kappa$ is defined in \eqref{kappa}.  
Then, by letting $\alpha \to 0$ in Theorem \ref{th:ROPS-TVaR}    we obtain the following proposition.
\begin{proposition}
\label{Pro-5.1}
%For   $0 \le \lambda <   \beta < 1$. 
The optimal allocation  $\boldsymbol{\omega}^*$ to the robust portfolio selection problem  \eqref{E-D}   is the solution to  the following  convex optimization problem: 
    \begin{equation}
    \label{prop 1}
    \begin{aligned}
        \min_{\boldsymbol{\omega} \in \R^n, \, s \in \R^+, \, t \in \R^+} \ &(1+\frac{\kappa}{2})(\hat{\boldsymbol{\mu}}^{\top}\boldsymbol{\omega}+\sqrt{\gamma_1}t)-\frac{\kappa}{2}d+\frac{\kappa}{2}\sqrt{s^2+(\hat{\boldsymbol{\mu}}^{\top}\boldsymbol{\omega}+\sqrt{\gamma_1}t-d)^2},\\
                 {\rm s.t.} \ &\sqrt{\gamma_1}\sqrt{\boldsymbol{\omega}^{\top}\hat{\boldsymbol{\Sigma}}\boldsymbol{\omega}} \le 
                 -\hat{\boldsymbol{\mu}}^{\top}\boldsymbol{\omega}  - r,\\
        &{\color{black}\sqrt{\boldsymbol{\omega}^{\top}(\hat{\boldsymbol{\Sigma}}+\gamma_2 \boldsymbol{I}_n)\boldsymbol{\omega}}} \leq s,\\
        & \sqrt{\boldsymbol{\omega}^{\top}\hat{\boldsymbol{\Sigma}}\boldsymbol{\omega}}\leq t, \ \ \ \ \boldsymbol{\omega}^{\top}\boldsymbol{e}=1.
    \end{aligned}
\end{equation} 
%where $\kappa =\frac{\beta-\lambda}{1-\beta} $ as  defined in \eqref{kappa} . 
\end{proposition}

\section{Numerical illustrations of   the portfolio selection models  using   real data}\label{sec-num-exam}

In this section, we present empirical experiments using real market data to demonstrate the application of the portfolio selection models proposed in Sections \ref{sec-optimal-selection} and \ref{sec-optimal-selection-g} of this paper. Specifically, we conduct two sets of experiments. First, we compare the performance of portfolios constructed using the models introduced in this paper, which are based on the mean-covariance uncertainty set $\mathcal{M}(\boldsymbol{\mu},\boldsymbol{\Sigma})$ defined in \eqref{MV uncertainty set}, with those obtained from several related portfolio selection models in the existing literature. Second, we compare the performance of portfolios under the uncertainty set $\mathcal{M}(\boldsymbol{\mu},\boldsymbol{\Sigma})$ with those under the generalized mean-covariance uncertainty set $\mathbb{D}_{n}(\gamma_1, \gamma_2)$ defined in \eqref{set-D}, and discuss the estimation of the ambiguity size parameters $\gamma_1$ and $\gamma_2$.  
{\color{black} Performance comparisons are conducted using cumulative portfolio  wealth  as well as
 the worst-case portfolio  CVaR.}

In the   empirical experiments conducted in this section, we selected the three largest companies (denoted by their stock trading symbols) by market capitalization from each of five distinct sectors: Technology (AAPL, NVDA, MSFT), Healthcare (JNJ, ABBV, LLY), Financial Services (BRK-B, JPM, V), Communication Services (GOOG, META, NFLX), and Consumer Defensive (WMT, COST, KO). Historical price data for these stocks were collected from Yahoo! Finance, covering eight years from March 2017 to February 2025. 
The first three years (March 2017 to February  2020), representing the pre-COVID period, were used as the training set. The subsequent three years (March 2020 to 
 February  2023), corresponding to the COVID period, together with the final two years (March 2023
to February 2025), representing the post-COVID period, constitute the test set. 
   The daily loss of a stock on (trading) day $t$, denoted by $\mu_t$, is computed as $\mu_t = -(P_{t+1} - P_{t}) / P_{t}$, where $P_{t}$ is the adjusted closing price of the stock on trading day $t$.

Throughout the experiments, we adopt a rebalancing strategy to evaluate and compare the performance of various portfolios using a rolling horizon with a 756-day window, which corresponds to approximately three years of trading data, assuming 252 trading days per year. For each portfolio selection model, the initial optimal portfolio
 $\boldsymbol{\omega}_{0}^*$ is obtained using   the first window in the training period. The first out-of-sample expected 
 portfolio return and the corresponding worst-case  portfolio CVaR  are then calculated as  
   $-\boldsymbol{\omega}_{0}^{*\top}\boldsymbol{\mu}_{1}$ and 
     $\sup_{G \in \mathcal{M}(\boldsymbol{\mu},\boldsymbol{\Sigma})}  \text{CVaR}_{\alpha} \big( \boldsymbol{\omega_0}^{*\top}\boldsymbol{X}^G\big)$, 
    where $\boldsymbol{\mu}_1$ denotes the expected daily loss vector of the selected  stocks on the first 
    trading day  in the test period. All mean vectors and covariance matrices used in the empirical experiments are computed as sample estimates based on the data within the corresponding window. 
    We then roll the window  forward one trading day at a time and 
     obtain  the out-of-sample expected portfolio returns and worst-case portfolio CVaRs for all subsequent trading days from March 2020 to February 2025.    
    
    The numerical results for these performance metrics are presented in Figures 1--3, corresponding to the two sets of empirical experiments described in Sections 6.1 and 6.2, respectively.

\subsection{First set of empirical experiments} \label{sec-6.1}
In the first set of experiments, we evaluate the performance of the portfolios derived from the models proposed in this paper, which are formulated under the mean-covariance uncertainty set $\mathcal{M}(\boldsymbol{\mu},\boldsymbol{\Sigma})$, in comparison with those constructed using several existing models that are related to ours. The models included in the comparison are as follows:

\begin{enumerate}[(a)]
          \item The mean--worst-case expected downside risk  ({\color{black}  M-WC-EDR}) model: This is  model  
     \eqref{model m=1 MV uncertainty 1d} that    
       minimizes  
     the mean loss weighted by   worst-case expected  downside risk.     
     The optimal  allocation $\boldsymbol{\omega}^*$ under this model  is given  in 
    Theorem \ref{TVaR-MV-PR-a=0}.

    \item The worst-case expected downside risk ({\color{black} WC-EDR}) model:    This is 
     \cite{chen2011tight}'s model that   
        minimizes     
      the worst-case  expected downside risk and is formulated as: 
          \by
   % \label{UPM-1}
       \begin{aligned}
   &  \min_{\boldsymbol{\omega\in\mathbb{R}^n}}\sup_{G\in\mathcal{M}(\boldsymbol{\mu},\boldsymbol{\Sigma})}\mathbb{E}[(\boldsymbol{\omega}^{\top}\boldsymbol{X}^G-d)_{+}], \  \  \  \  \text{s.t.}  \ \ \ \  \boldsymbol{\omega}^{\top}\boldsymbol{e}=1, 
                    \end{aligned}    
    \ey
where $d \in \R$.  The  optimal  allocation  $\boldsymbol{\omega}^*$  under this model     is given  in  \cite{chen2011tight}.

    \item The worst-case CVaR of the balanced portfolio loss ({\color{black} WC-CVaR-BPL}) model: This is 
       model \eqref{portfolio CVaR-kappa}, which       minimizes 
     the worst-case CVaR of the balanced portfolio loss. 
     The  optimal allocation   $\boldsymbol{\omega}^*$ under this model   is given in  
    Theorem \ref{CVaR-MV-PR}.

   \item The worst-case CVaR of the  portfolio loss ({\color{black} WC-CVaR}) model: This is     \cite{ZF09}'s model, which  minimizes   the worst-case CVaR  of the portfolio loss and is formulated as:  
    \by
     %\label{model-CVaR}
     \begin{aligned} 
     & \min_{\boldsymbol{\omega\in\mathbb{R}^n}}\sup_{G\in\mathcal{M}(\boldsymbol{\mu},\boldsymbol{\Sigma})}{\rm CVaR}_{\alpha}(\boldsymbol{\omega}^{\top}\boldsymbol{X}^G),  \  \ \ \   \text{s.t.}  \ \ \ \ \boldsymbol{\omega}^{\top}\boldsymbol{e}=1, 
                    \end{aligned}    
     \ey
where $0 < \alpha < 1$. The  optimal $\boldsymbol{\omega}^*$  under this model    is given  in  \cite{ZF09}.

\item 
 {\color{black}  The mean--worst-case CVaR (M-WC-CVaR)  model:
 \begin{equation}
\label{model-mean-CVaR}
\begin{aligned}
&\min_{\boldsymbol{\omega} \in \mathbb{R}^n} \ \boldsymbol{\omega}^{\top} \boldsymbol{\mu} + \kappa  \sup_{G \in \mathcal{M}(\boldsymbol{\mu}, \boldsymbol{\Sigma})} \text{CVaR}_\alpha(\boldsymbol{\omega}^{\top} \boldsymbol{X}^G), \
&\text{s.t.}  \ \ \ \ \boldsymbol{\omega}^{\top} \boldsymbol{e} = 1,
\end{aligned}
\end{equation}
where $\kappa > 0$ and  $0< \alpha < 1$.
 This model is a special case of model \eqref{model-risk measure}, with the risk measure specifically given by CVaR.
  Using techniques similar to those employed in the proof of Theorem 4.2,  and under the condition $\frac{\alpha}{1-\alpha} \, b_0 >1$, one can show that  if $\kappa > \kappa_\alpha^* =  \Big ( \sqrt{\frac{\alpha}{1-\alpha} \, b_0}-1 \Big )^{-1}$, problem \eqref{model-mean-CVaR} admits a unique solution; and 
if $\kappa \le \kappa_\alpha^*$, then the problem has no solution. We omit the proof of this result in this paper. 
Since $b_0 >0$ and   $\frac{\alpha}{1-\alpha} \, b_0 \to \infty$ as $\alpha \to 1$,  the condition $\frac{\alpha}{1-\alpha} \, b_0 >1$ typically holds in practical applications. This is because the confidence level $\alpha$ 
 used in CVaR is usually close to one -- for example, 0.90, 0.95, or 0.99.  
 }

    \item  The  mean-variance (M-V) model: This is a classical mean-variance model   \eqref{model-MV}, which     minimizes  the variance of the 
     portfolio loss  subject to a given expected portfolio  return $r_0$.  The  optimal allocation   $\boldsymbol{\omega}^*$  under this model     
     is given  in Lemma \ref{sol-MV-model}.
  \end{enumerate}

{\color{black} 
In the comparison, we used the M-V model as the benchmark, while the objectives of 
all other models in this section under 
comparison  are based 
on the expected portfolio loss,  the expected downside risk, or  the CVaR of portfolio loss and/or balanced  portfolio loss. These 
 portfolio selection models  involve several key parameters, 
including: the expected portfolio return $r_0$ in the M-V model;  
the target portfolio return $-d$  in the M-WC-EDR, WC-EDR, and WC-CVaR-BPL models; and 
 the weighting parameter $\kappa$  in  the M-WC-EDR,  WC-CVaR-BPL, and  M-WC-CVaR models.
 
The parameter $r_0$  
represents the expected portfolio return desired by the investor. Various 
methods exist for selecting $r_0$; see, for example, \cite{EGBG09} and references therein. 
A practical approach to determine
$r_0$ is to use historical market returns.  
In our numerical analysis, we set 
$r_0$ to the average daily return of the S\&P 500 index over the training period, since the stocks in our portfolio are selected  from this index.

In the context of the reward-penalty framework, the parameter $-d$ 
  can be interpreted as the threshold that  triggers  
 a  bonus payment. A natural choice in this setting is to set $-d \ge r_0$, 
 ensuring that the target return is at least as high as the investor's expected return. 
  On the other hand, as pointed out in Section \ref{sec-intro}, in the context of risk management, the parameter $d$ 
 can be interpreted as  a risk metric, such as the VaR or CVaR of the portfolio loss, which is critical for   the investor. 
      In our empirical results,  the parameter  $-d$ is selected  as the bonus threshold in Figures \ref{Figure1} and \ref{Figure3}, 
       whereas  in Figure \ref{Figure2}, 
  the parameter $d$   is set equal to the average of the CVaRs of the selected stocks.
      
  The parameter $\kappa$ 
   serves as a weighting factor on the downside risk, defined as $(\boldsymbol{\omega}^\top \bm{X}^G - d)_+$, of the portfolio. 
   For the M-WC-EDR and M-WC-CVaR models,  there exists a  threshold value
     $\kappa^*$ such that  the models admit  solutions only when $\kappa > \kappa^*$.    
    Our numerical results (Figure~\ref{Figure2}(a)-\ref{Figure2}(d)) indicate that, for the 
    M-WC-EDR and M-WC-CVaR models,  increasing $\kappa$  reduces 
  cumulative portfolio wealth, while improving the worst-case portfolio CVaR. 
 From the perspective of maximizing  cumulative portfolio wealth, it is therefore preferable to choose $\kappa$
  just slightly above the threshold $\kappa^*$. 
  However, from a risk management standpoint,
     specifically with regard to controlling CVaR, 
  a much larger $\kappa$  may be desirable.
     In contrast, for the WC-CVaR-BPL model, increasing $\kappa$ 
 results in improved cumulative portfolio wealth, but at the cost of a worse worst-case portfolio CVaR.

  Based on our empirical experiments, we recommend that when using the M-WC-EDR and M-WC-CVaR models, risk-seeking investors choose smaller $\kappa$-values close to the threshold, 
  while risk-averse investors choose  larger  $\kappa$-values. 
 In contrast,  when using the WC-CVaR-BPL model, risk-seeking investors  
  should select   larger  $\kappa$-values, whereas risk-averse investors are advised to  
   choose smaller $\kappa$-values near the threshold.  In general, it is unlikely that a single value of $\kappa$ 
   can  simultaneously satisfy both objectives. The problem of identifying the optimal $\kappa$ according to a well-defined decision criterion
remains an important question for future research.   
           
In addition, the parameter $\alpha$  represents the confidence level used in calculating the CVaR of the portfolio loss or the balanced portfolio loss. In risk management, one typically considers high confidence levels, with 
 common choices being  
  0.90, 0.95, and 0.99 (see, for example, \cite{NFE15}. In our empirical experiments, we set 
$\alpha =0.90$. 
   }

{\color{black} 

From Figures~\ref{Figure1}(a), \ref{Figure1}(c), and \ref{Figure1}(e), we observe that across all three cases where  $-d>r_0$, $-d=r_0$, and $-d < r_0$, 
the cumulative portfolio wealth of the M-WC-EDR model  is significantly  higher 
than that of the WC-EDR model and all other models   for most of the test period.
 Furthermore, the cumulative portfolio wealth of the WC-CVaR-BPL model is slightly but consistently better than that of 
the WC-CVaR model across all three cases. These results indicate that incorporating downside risk as a weighted component of portfolio loss,  rather than focusing solely on downside risk, can lead to improved expected return performance.
 Moreover, from the same figures, we find that the cumulative portfolio wealth of the M-V model is lower than all other models, while its worst-case portfolio CVaR is smaller than that of the M-WC-EDR model but slightly larger than those of the other models.

  However,    as shown in Figures~\ref{Figure1}(b), \ref{Figure1}(d), and \ref{Figure1}(f), models that exhibit higher  expected returns  tend to incur   larger  worst-case portfolio CVaRs. In particular, 
the M-WC-EDR  model, which achieves the highest expected return, also exhibits the highest worst-case CVaR. The worst-case CVaRs of the other models are relatively close to each other. These findings highlight  the inherent trade-off between expected return and worst-case CVaR in robust portfolio optimization.

 Furthermore, to examine the impact of the weighting factor $\kappa$
 on portfolio performance, we compute the cumulative  portfolio wealth curves and worst-case portfolio  CVaRs for models that include the 
$\kappa$-weighted downside risk term, namely   the M-WC-EDR, M-WC-CVaR, and WC-CVaR-BPL 
 models,    as shown in Figure~\ref{Figure2}. These models are evaluated  under  various $\kappa$ 
 values of   1, 5, 10, and 15. In particular, for the WC-CVaR-BPL model, we also include the case $\kappa=0$, 
which corresponds to the WC-CVaR model. In this numerical analysis, we treat $d$ 
as a risk metric and set it equal to the average of the CVaRs of the selected stocks over the in-sample period. Accordingly, we have 
$d=0.03.$
         
 Figure~\ref{Figure2} demonstrates  that for the  M-WC-EDR and  M-WC-CVaR  models, smaller 
 values of $\kappa$ 
  lead to higher cumulative wealth, but also result in higher worst-case CVaRs for   most of the test period.
   In contrast, for the WC-CVaR-BPL model, larger values of 
$\kappa$  produce higher cumulative wealth and correspondingly higher worst-case CVaRs over the same periods.

It is also worth noting that the cumulative wealth curves and   worst-case CVaR values in Figures~\ref{Figure2}(c)-\ref{Figure2}(f) are relatively close to each other, though small differences, as observed in our numerical results,   do exist. 
%For example, at the terminal time (the last trading  day 
%of February 2025) of the test period, the worst-case CVaRs for the 
%M-WC-CVaR and  WC-CVaR-BPL models 
% are 0.03042 and 0.03038, respectively, while the corresponding cumulative wealth values are 2.14335 and 2.14063.
These results further illustrate the trade-off between return and risk: portfolios with higher returns tend to be associated with 
larger  worst-case CVaRs.        

The results  in Figures~\ref{Figure2}(d) and \ref{Figure2}(f) also  imply that  M-WC-CVaR and  WC-CVaR-BPL are  
robust with respect to  the weighting factor  $\kappa$, whereas  
 the M-WC-EDR model  is sensitive  to the weighting factor  $\kappa$, in particular, for smaller  values of $\kappa$. 
 However, as $\kappa$ increases, the  M-WC-EDR model appears to become robust to changes in   $\kappa$.         
       
           Figures \ref{Figure1} and \ref{Figure2} also show that the wealth  and worst-case CVaR values for all models exhibit distinct patterns across different market phases. Specifically, during the initial  outbreak 
             phase of the  COVID period (March 2020), cumulative wealth decreases sharply while CVaR increases significantly. These metrics then stabilize throughout the remainder of the COVID period.  However, in the post-COVID period, the 
             pattern  reverses: cumulative wealth increases, while CVaR values gradually decline.
    }

\subsection{Second set of  empirical experiments} \label{sec-6.2}
In the second set of experiments, we compare the performance of portfolios based on the mean-covariance uncertainty set $\mathcal{M}(\boldsymbol{\mu},\boldsymbol{\Sigma})$ with those based on the generalized mean-covariance uncertainty set $\mathbb{D}_{n}(\gamma_1, \gamma_2)$, and discuss how to estimate the ambiguity size parameters $\gamma_1$ and $\gamma_2$.
 
As pointed out in \citet{kang2019data}, one of the key challenges in applying the uncertainty set $\mathbb{D}_{n}(\gamma_1, \gamma_2)$ is the choice of the ambiguity size parameters $\gamma_1$ and $\gamma_2$. \citet{kang2019data} introduced a bootstrap method to estimate these parameters using observed data. 
However, this method relies on a resampling procedure that randomly draws observations with replacement, which may overlook the fact that stock data is collected on a day-by-day basis.  To address this issue, we employ an alternative approach to estimate $\gamma_1$ and $\gamma_2$ as follows:

First, we estimate $\hat{\boldsymbol{\mu}}$ and $\hat{\boldsymbol{\Sigma}}$ in $\mathbb{D}_{n}(\gamma_1, \gamma_2)$ using the sample mean vector and sample covariance matrix of the daily losses of the selected  stocks. Next, to estimate $\gamma_1$ and $\gamma_2$, we apply a rolling window approach over all 
the data. More specifically,    for the $i$th 756-day rolling window,
we compute the sample mean vector $\hat{\boldsymbol{\mu}}_i$ and sample covariance matrix $\hat{\boldsymbol{\Sigma}}_i$ of the daily losses within that window. We then calculate the deviations of  $\hat{\boldsymbol{\mu}}_i$ and $\hat{\boldsymbol{\Sigma}}_i$ relative to $\hat{\boldsymbol{\mu}}$ and $\hat{\boldsymbol{\Sigma}}$, respectively, using the following metrics specified in the uncertainty set  $\mathbb{D}_{n}(\gamma_1, \gamma_2)$:
$\gamma_{1i} = (\hat{\boldsymbol{\mu}}_i - \hat{\boldsymbol{\mu}})^\top \hat{\boldsymbol{\Sigma}}^{-1} (\hat{\boldsymbol{\mu}}_i - \hat{\boldsymbol{\mu}})$ and $\gamma_{2i} = ||\hat{\boldsymbol{\Sigma}}_i - \hat{\boldsymbol{\Sigma}}||_F$.
Finally, we estimate $\gamma_1$ and $\gamma_2$ by averaging all $\gamma_{1i}$ and $\gamma_{2i}$ values, respectively.  This method accounts for the fact that the stock data are observed sequentially over time. Based on this approach, we obtain {\color{black}  $\gamma_1 = 0.009$ and $\gamma_2 = 0.001$} as the estimates of the ambiguity sizes for the  uncertainty set $\mathbb{D}_n(\gamma_1, \gamma_2)$.

To compare the portfolio performance  based on the uncertainty set $\mathcal{M}(\boldsymbol{\mu},\boldsymbol{\Sigma})$ with that  based on the  uncertainty set $\mathbb{D}_{n}(\gamma_1, \gamma_2)$, 
and to consider the impact of the parameters $\gamma_1$ and $\gamma_2$ on portfolio performance,  
  we compute the cumulative portfolio wealth curves and  worst-case portfolio CVaRs  based on model \eqref{E-D} {\color{black}  for four  different choices of $(\gamma_1, \gamma_2)$: 
 $(\gamma_1, \gamma_2) = (0, 0)$, 
 $(\gamma_1, \gamma_2) = (0, 0.001)$,  $(\gamma_1, \gamma_2) = (0.009, 0.001)$, and 
  $(\gamma_1, \gamma_2) = (0.009, 0.001)$}, as shown in Figure \ref{Figure3}.
   Note that $\mathbb{D}_{n}(0, 0) =\mathcal{M}(\hat{\boldsymbol{\mu}}, \hat{\boldsymbol{\Sigma}})$. 
  Additionally,  in model \eqref{E-D}, the constraint $\min_{G \in \mathcal{\mathbb{D}}_{n}(\gamma_1,\gamma_2)}\mathbb{E}[-\boldsymbol{\omega}^{\top}\boldsymbol{X}^G]\geq r$ ensures  that the  worst-case expected return over the uncertainty set is no less than $r$. Compared to 
 the target return of $-d$  in  model  \eqref{model m=1 MV uncertainty 1d}, the  expected  portfolio return $r_0$ in model  \eqref{portfolio CVaR}, and based on the discussions about 
   choosing   
 $d$  and $r_0$,  it is natural to set $r<r_0<-d$ since $r_0$ is the lowed expected return.  {\color{black} In our numerical analysis, 
we specify  the parameter $r$  as  the average of the lowest  expected returns among the selected
stocks, which yields $r=-0.098$  in 
 Figure~\ref{Figure3}.
  }

{\color{black} 
From Figure \ref{Figure3}, we observe that the  portfolio based on the uncertainty set $\mathbb{D}_{n}(0, 0.001)$ achieves the  highest   cumulative portfolio  wealth over the test period. Furthermore,
 the  portfolios based on the uncertainty sets $\mathbb{D}_{n}(0, 0.001)$ and  
 $\mathbb{D}_{n}(0.009, 0.001)$ achieve higher  cumulative portfolio  wealth over the test period 
 than those based on $\mathbb{D}_{n}(0, 0)$ and  
 $\mathbb{D}_{n}(0.009, 0)$.  Conversely, portfolios that yield higher wealth also exhibit higher CVaR. These results further highlight the trade-off between expected return and worst-case CVaR.
  
     In addition, the  results  show that 
    portfolio performance is more sensitive to the parameter $\gamma_2$, which represents 
    the divergence level of  the covariance matrix of the loss vector from  the reference matrix, 
     than to  the parameter 
    $\gamma_1$, which   represents the  
     divergence level of  the  mean vector of the loss vector  from  the reference mean vector.  
   Note that 
   $\mathbb{D}_{n}(0, 0.001)$ is a subset of 
        $\mathbb{D}_n(0.009, 0.001)$, but it   contains $\mathbb{D}_{n}(0, 0)$ as a subset. Therefore, 
        Figure \ref{Figure3}
         implies that  a model based on a larger uncertainty set does not necessarily yield better returns or 
         worse worst-case CVaR.     The parameter $\gamma_2$ has 
         a greater   impact  on  portfolio performance   than   $\gamma_1$. 
          Thus, how to model the uncertainty of the covariance matrix of the loss vector is a critical topic in moment-based robust optimization problems and their applications.         
        
     }

\begin{figure}[H]
    \centering
    \begin{subfigure}{0.45\textwidth}
        \centering
        \includegraphics[width=\linewidth]{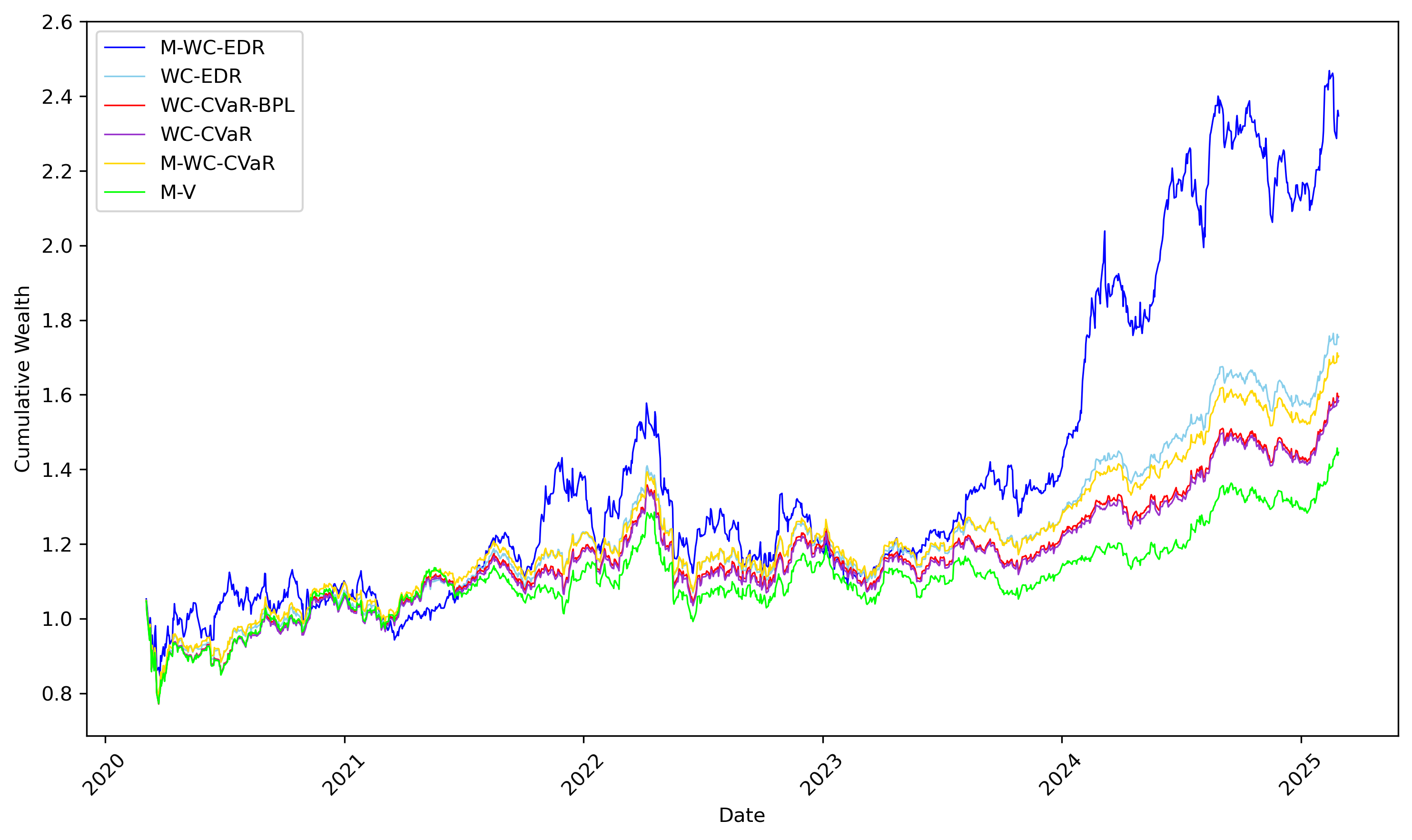}
        \caption{Cumulative wealth when $-d=0.006>r_0$}
    \end{subfigure}
    \begin{subfigure}{0.45\textwidth}
        \centering
        \includegraphics[width=\linewidth]{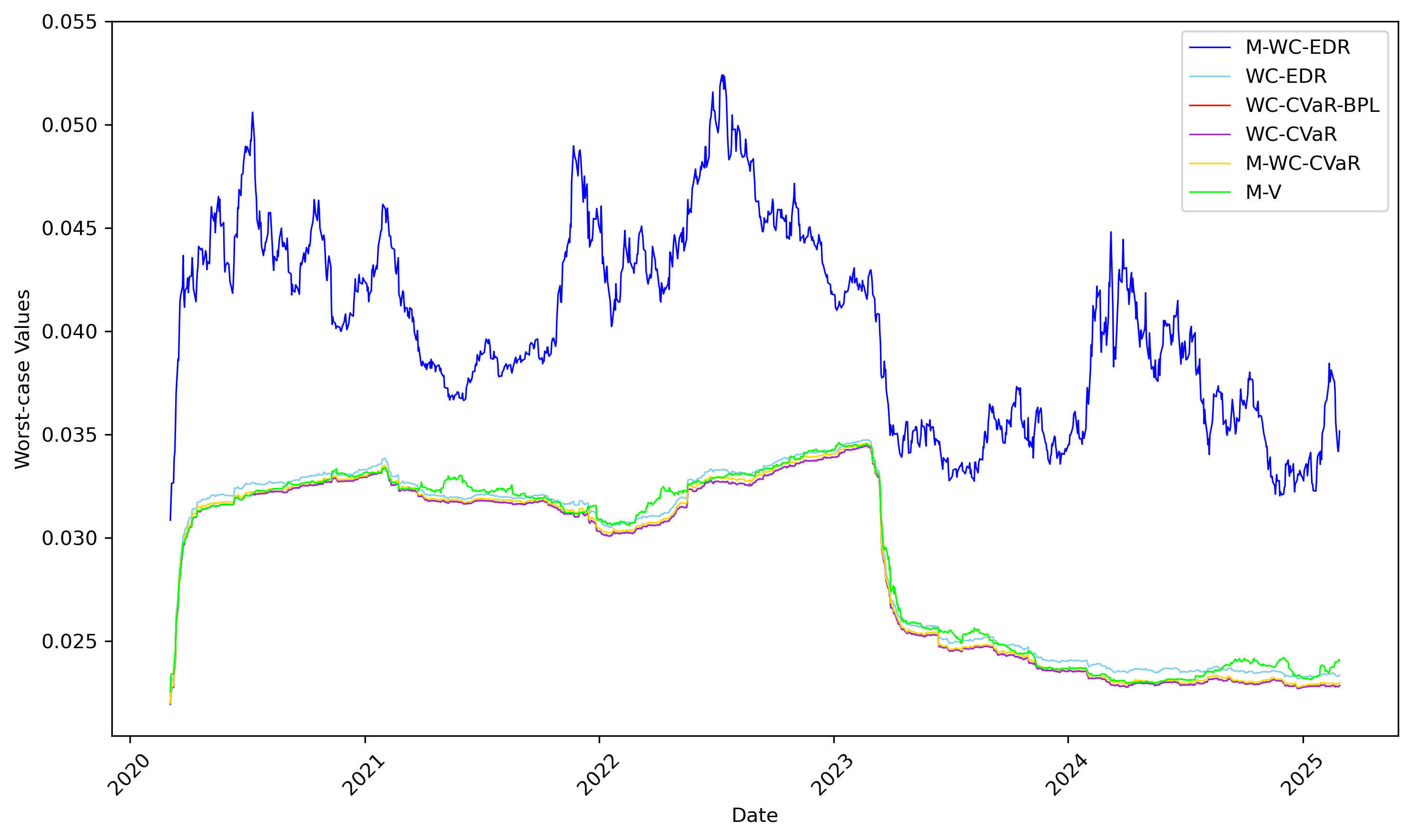}
        \caption{Worst-case CVaR when $-d=0.006>r_0$}
    \end{subfigure}
    
    \begin{subfigure}{0.45\textwidth}
        \centering
        \includegraphics[width=\linewidth]{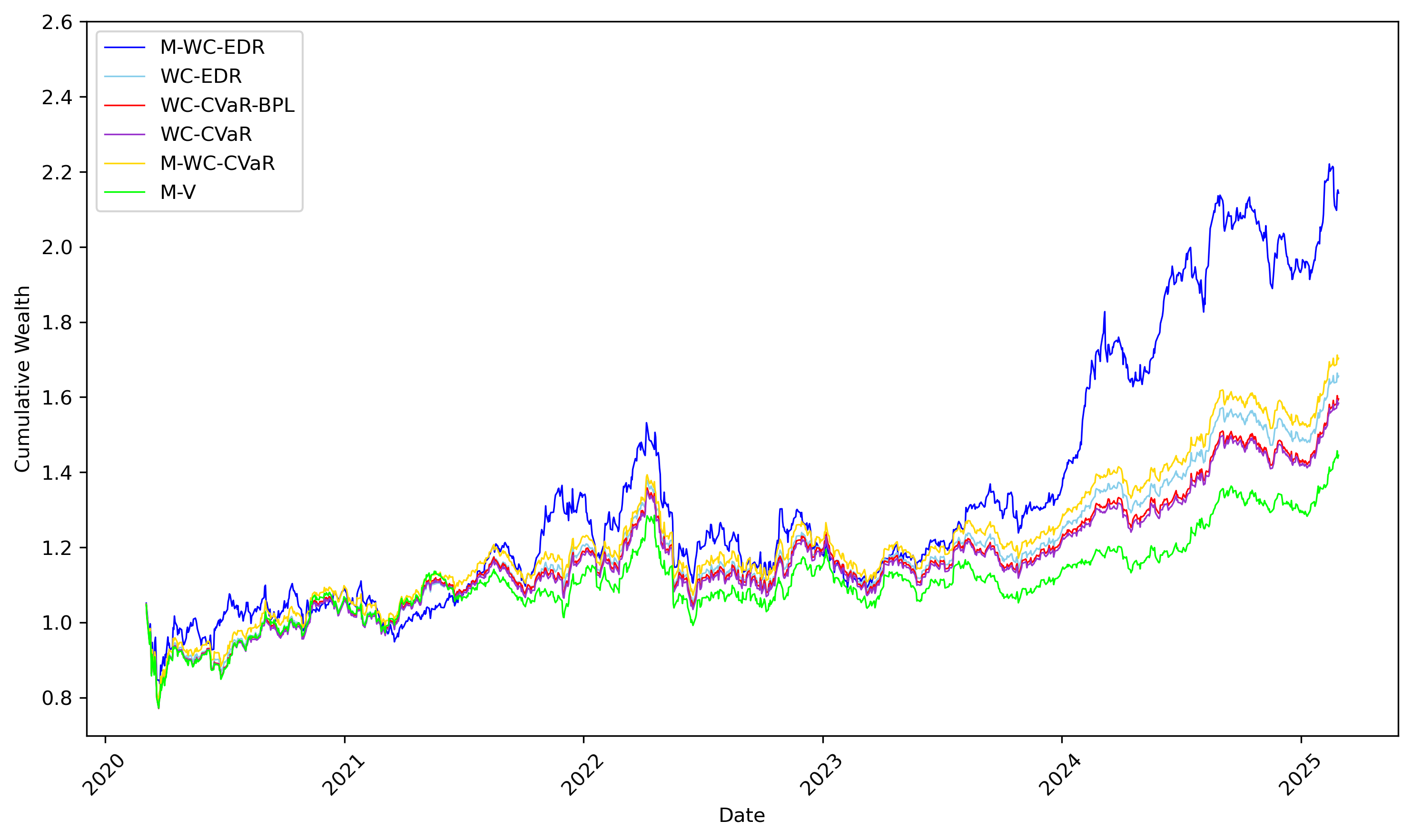}
        \caption{Cumulative wealth  when $-d=0.0003=r_0$}
    \end{subfigure}
    \begin{subfigure}{0.45\textwidth}
        \centering
        \includegraphics[width=\linewidth]{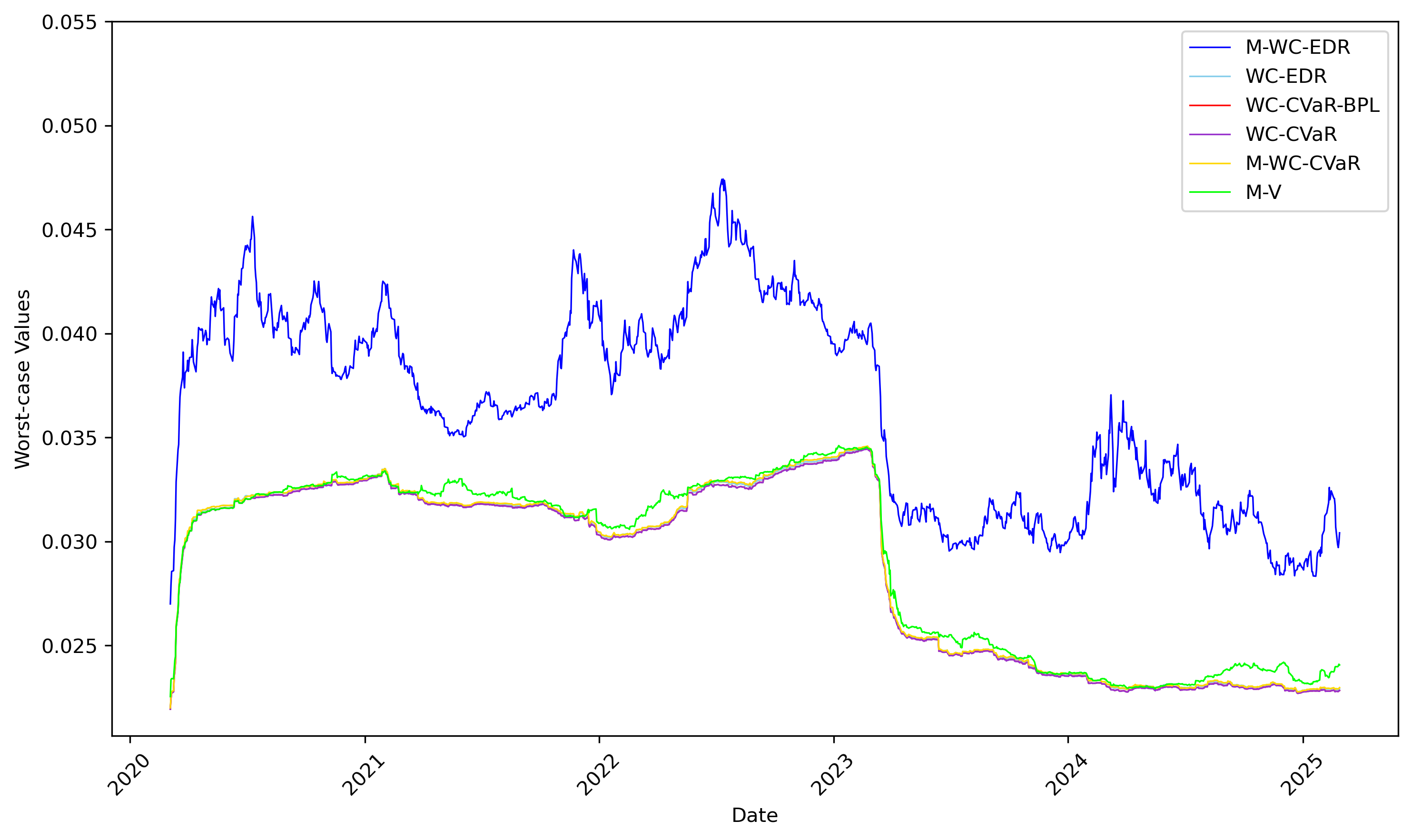}
        \caption{Worst-case CVaR when $-d=0.0003=r_0$}
    \end{subfigure}

    \begin{subfigure}{0.45\textwidth}
        \centering
        \includegraphics[width=\linewidth]{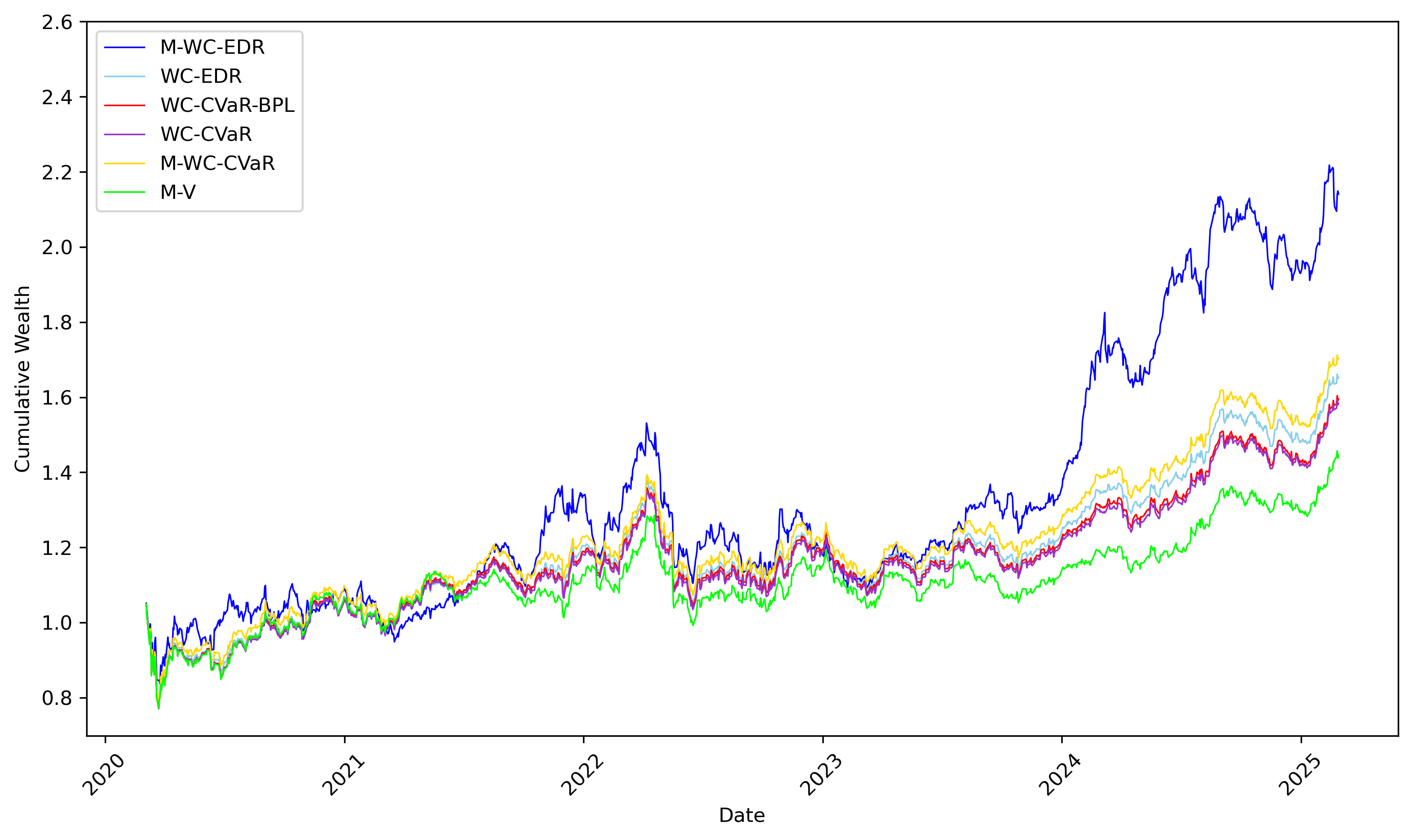}
        \caption{Cumulative wealth when $-d=0.00006<r_0$}
    \end{subfigure}
    \begin{subfigure}{0.45\textwidth}
        \centering
        \includegraphics[width=\linewidth]{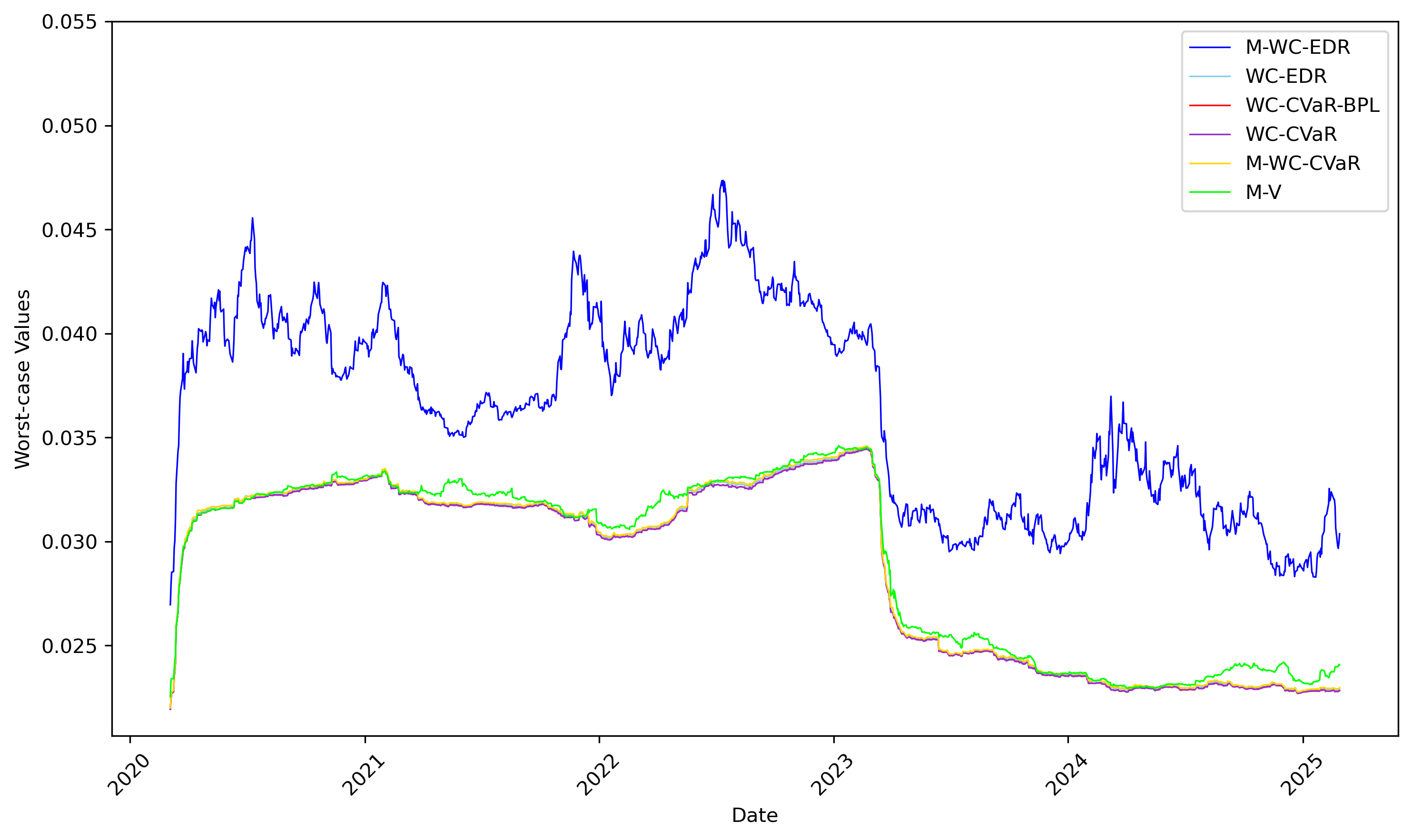}
        \caption{Worst-case CVaR when $-d=0.00006<r_0$}
    \end{subfigure}
    
    \caption{Cumulative wealth and worst-case CVaRs 
    of portfolios based on the M-WC-EDR, WC-EDR, 
    WC-CVaR-BPL, WC-CVaR, M-WC-CVaR, and M-V models, 
         with     $r_0=0.0003$,  $\alpha$=0.9,     {\color{black}$\kappa$=0.5},  and different target returns of $-d$.
        }
    \label{Figure1}
\end{figure}

\begin{figure}[H]
    \centering
    \begin{subfigure}{0.45\textwidth}
        \centering
        \includegraphics[width=\linewidth]{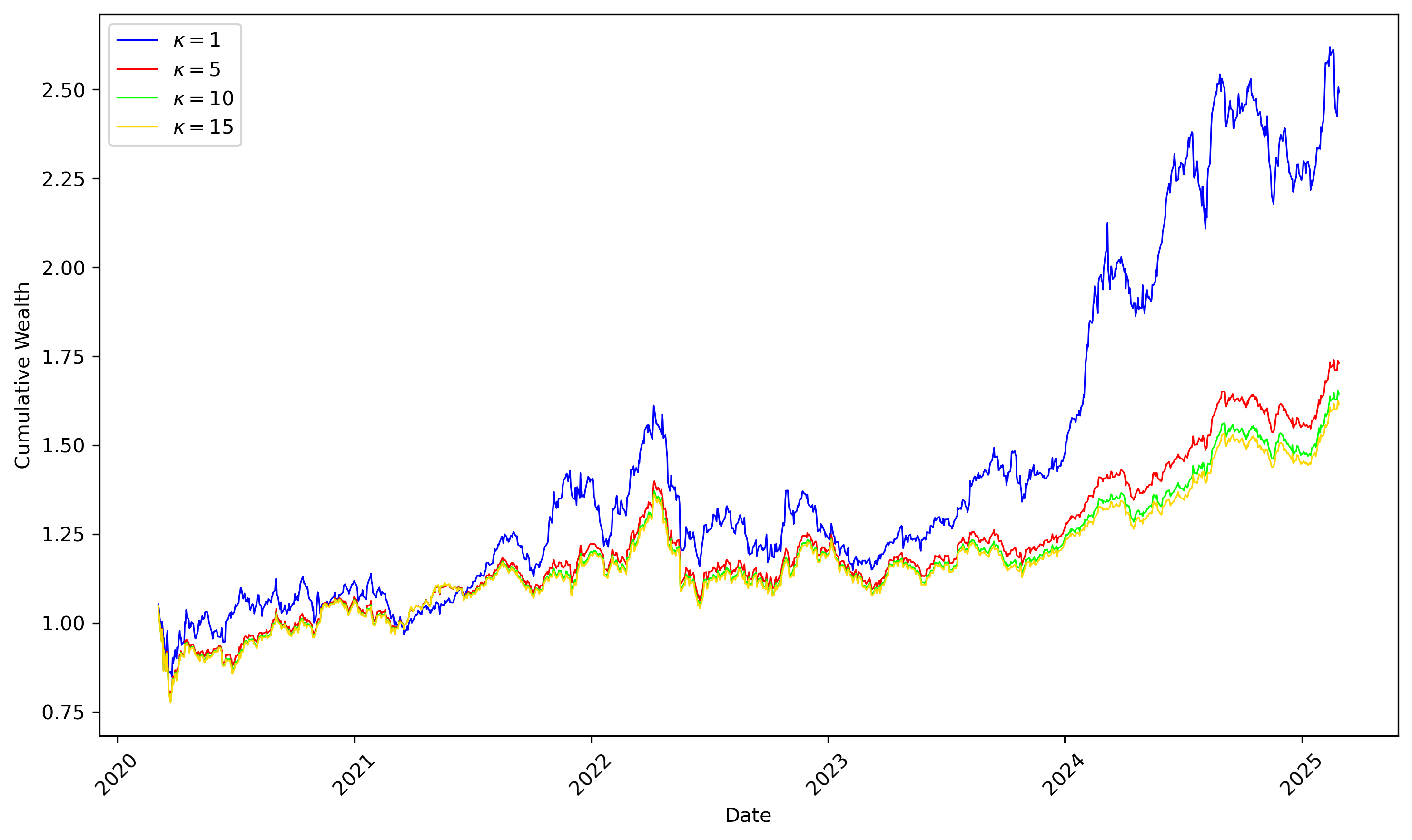}
        \caption{Cumulative wealth of the M-WC-EDR model}
    \end{subfigure}
    \begin{subfigure}{0.45\textwidth}
        \centering
        \includegraphics[width=\linewidth]{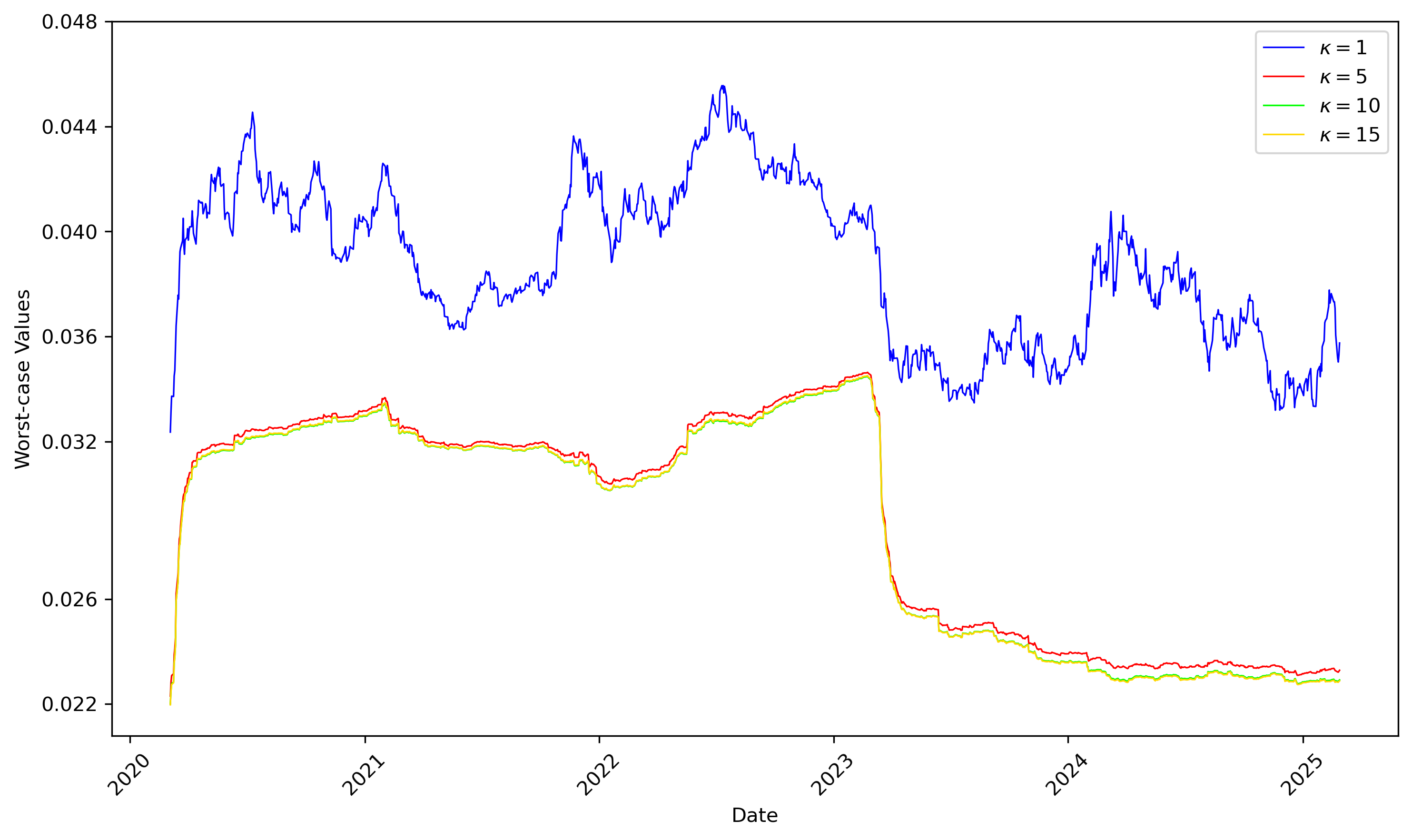}
        \caption{Worst-case CVaRs  of the M-WC-EDR  model}
    \end{subfigure}
    
    \begin{subfigure}{0.45\textwidth}
        \centering
        \includegraphics[width=\linewidth]{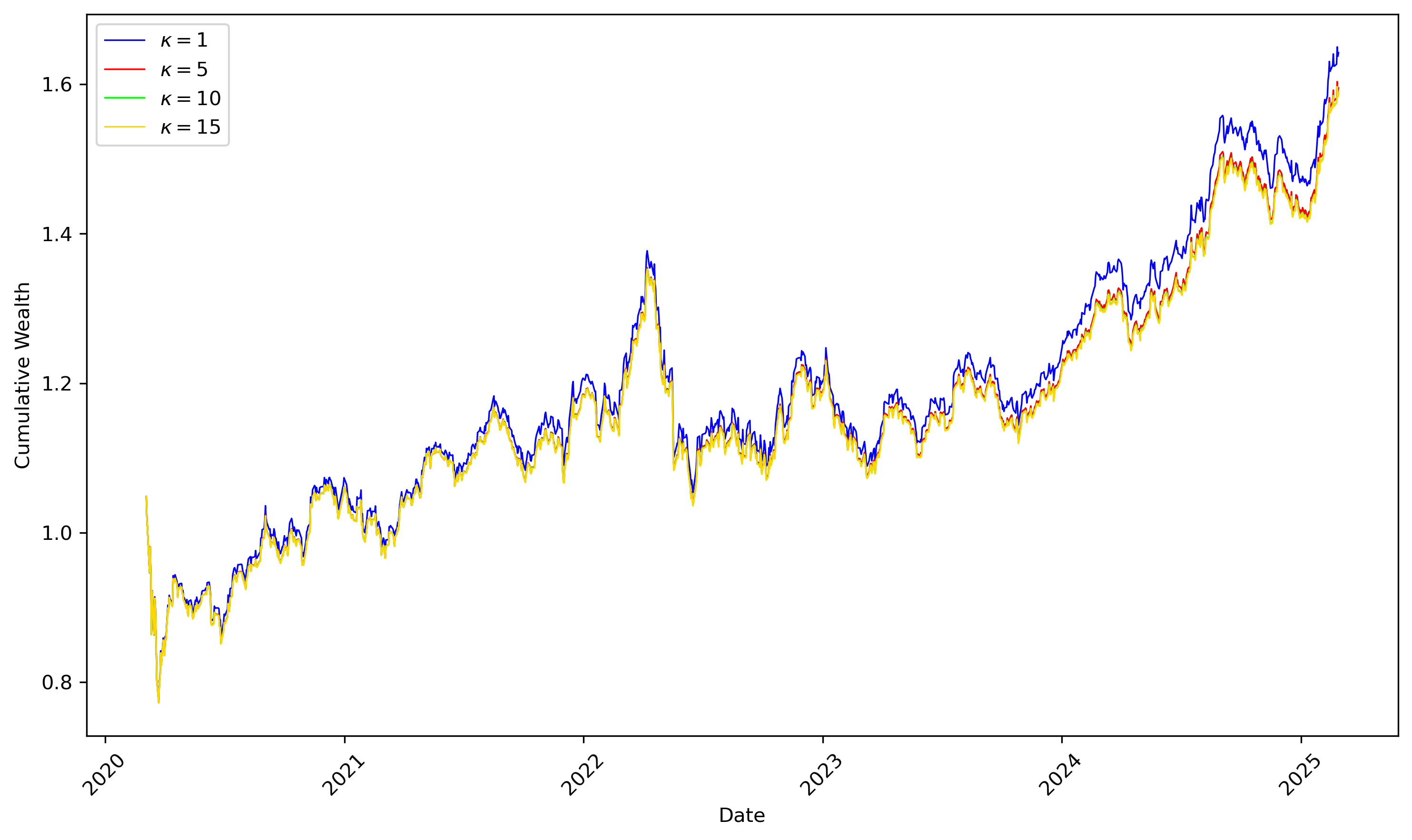}
        \caption{Cumulative wealth of the M-WC-CVaR model
            }
    \end{subfigure}
    \begin{subfigure}{0.45\textwidth}
        \centering
        \includegraphics[width=\linewidth]{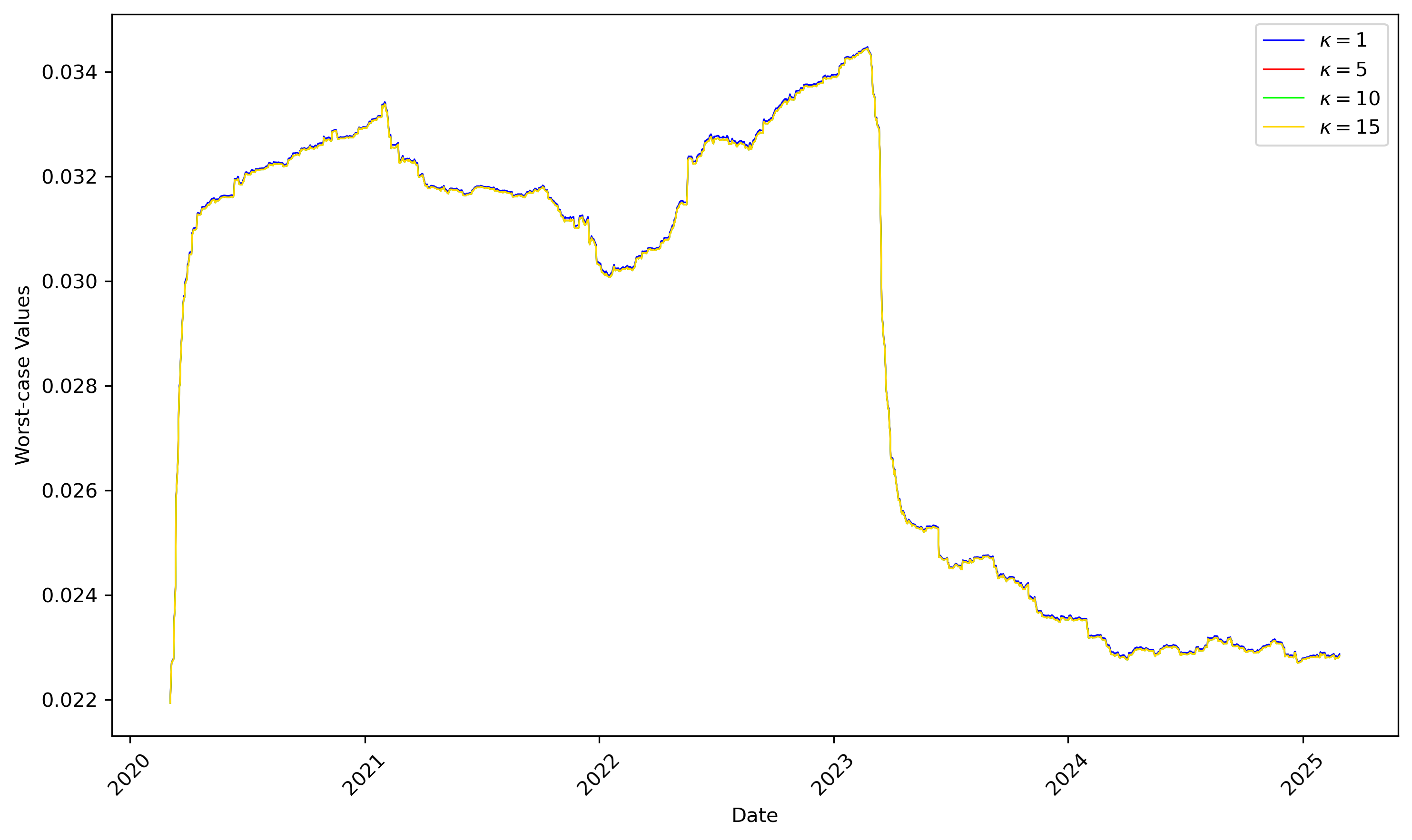}
        \caption{Worst-case CVaRs  of M-WC-CVaR model
        }
    \end{subfigure}

    \begin{subfigure}{0.45\textwidth}
        \centering
        \includegraphics[width=\linewidth]{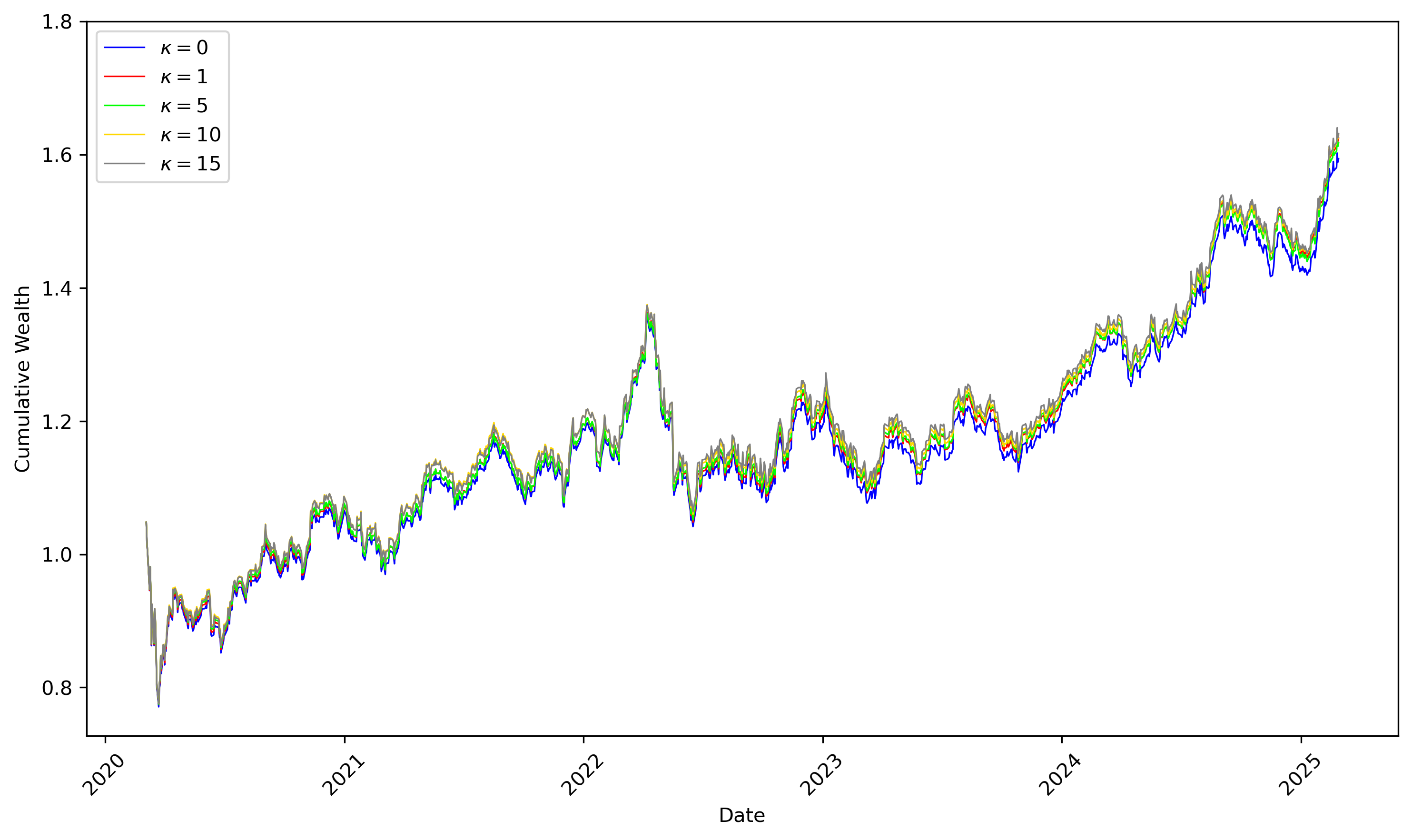}
        \caption{Cumulative wealth of  the WC-CVaR-BPL model 
        }
        
    \end{subfigure}
    \begin{subfigure}{0.45\textwidth}
        \centering
        \includegraphics[width=\linewidth]{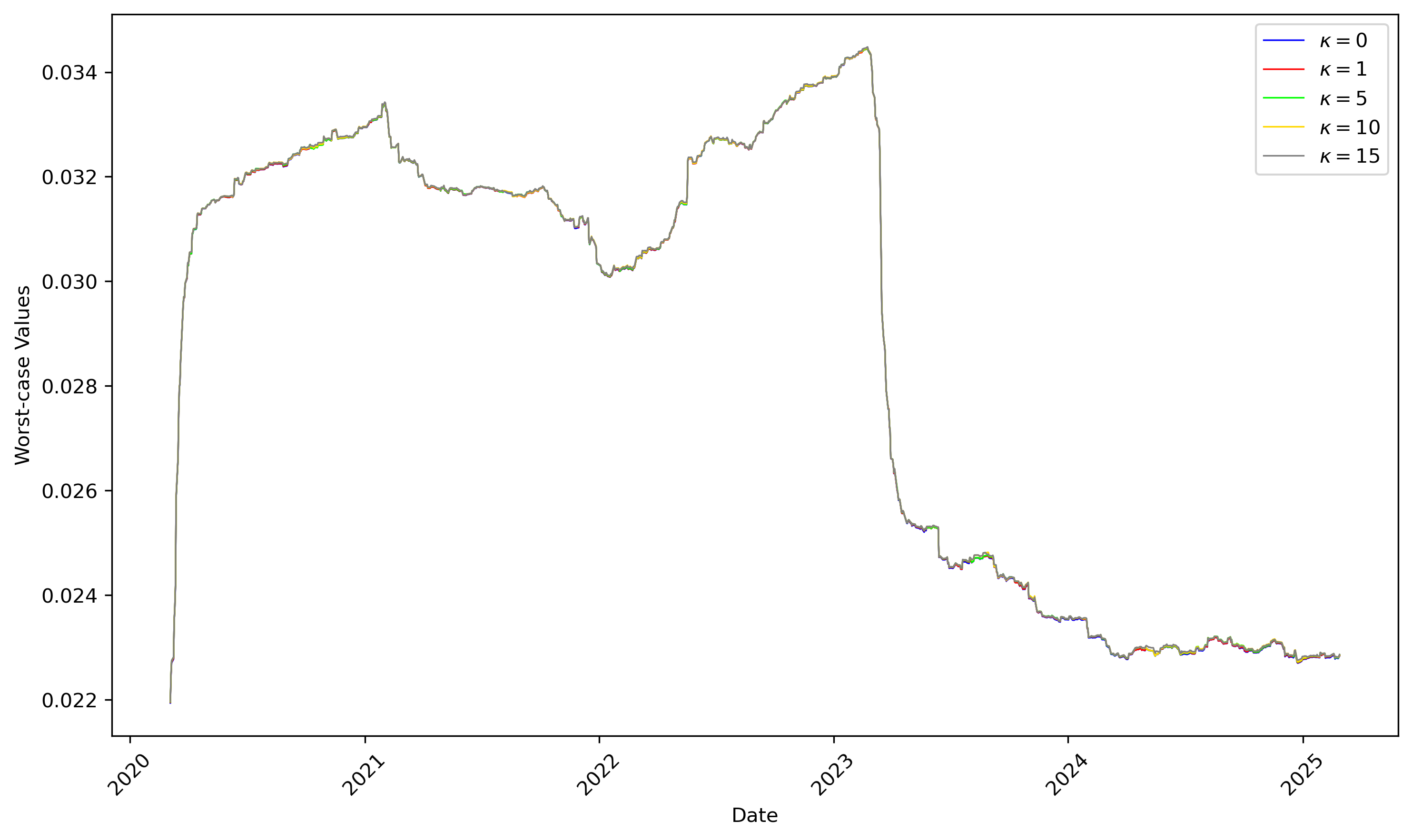}
         \caption{Worst-case CVaRs of the WC-CVaR-BPL model
         }
         \end{subfigure}
    
    \caption{Cumulative wealth   and worst-case CVaRs  
    of portfolios based on the M-WC-EDR,  M-WC-CVaR, and WC-CVaR-BPL  models,   with 
      $\alpha$=0.9, 
      {\color{black}$\kappa$=0.5},   {\color{black} the risk metric} $d=0.03$,  
      and  different weights  $\kappa$.
       }
    \label{Figure2}
\end{figure}

\begin{figure}[H]
    \centering
     \begin{subfigure}{0.45\textwidth}
        \centering
        \includegraphics[width=\linewidth]{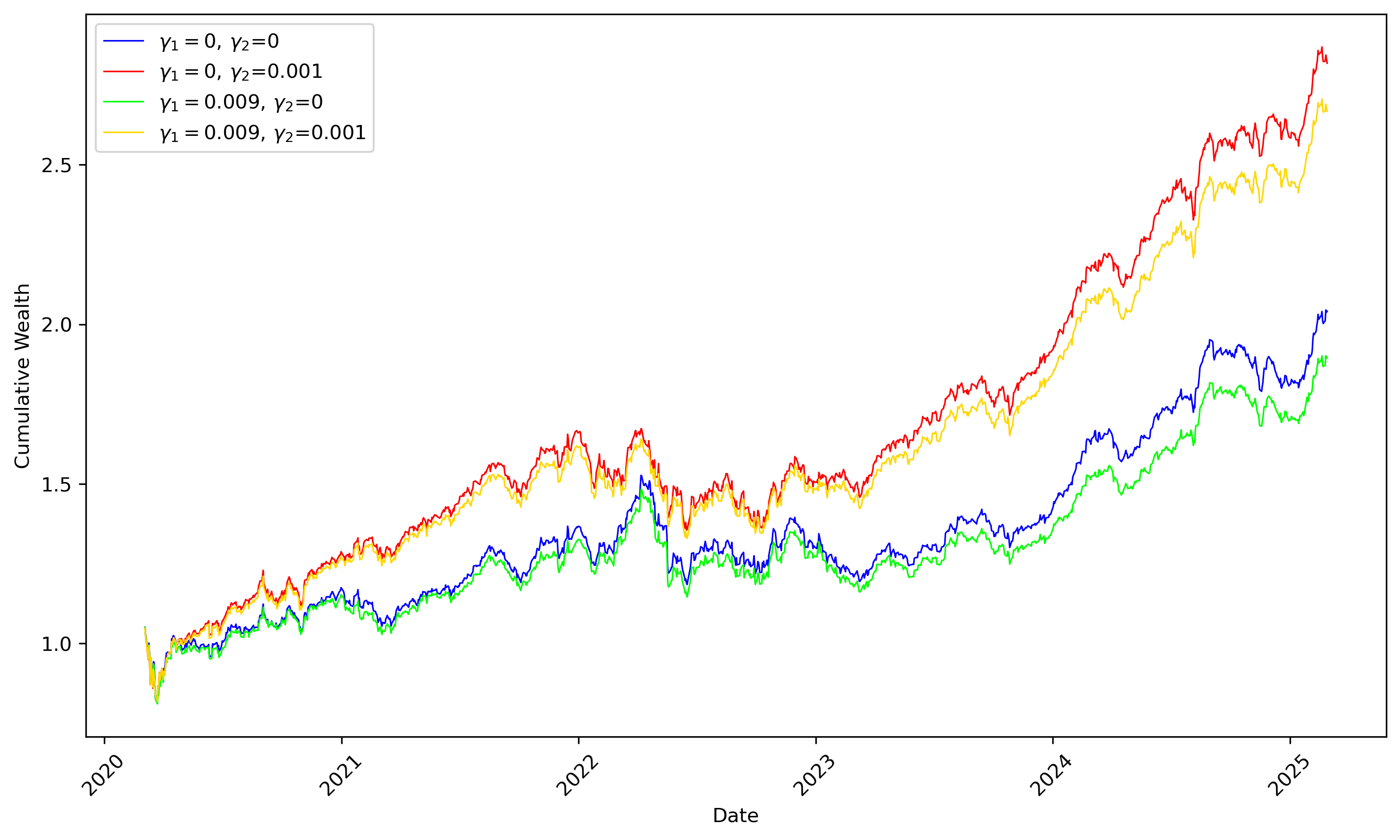}
        \caption{Cumulative wealth of    model   \eqref{E-D}
        }
    \end{subfigure}
    \begin{subfigure}{0.45\textwidth}
        \centering
        \includegraphics[width=\linewidth]{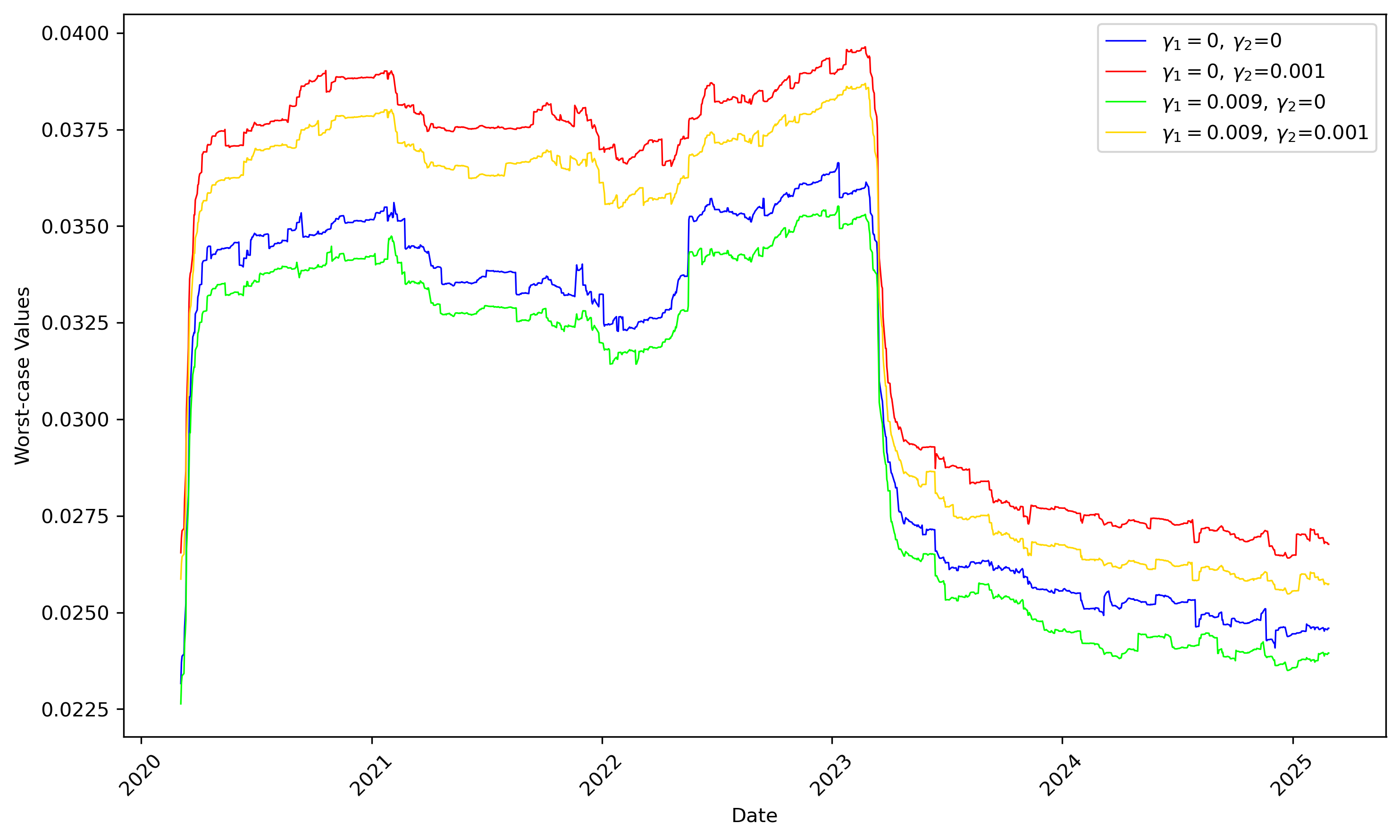}
        \caption{Worst-case CVaRs of  model   \eqref{E-D} 
           }
    \end{subfigure}
    \caption{Cumulative wealth and  worst-case CVaR of portfolios based on model   \eqref{E-D}, with target return $-d= 0.001$,  worst-case expected return  {\color{black}$r=-0.098$},   $\kappa=1$,  and different values of $\gamma_1$ and $\gamma_2$.}
    \label{Figure3}
\end{figure}

\section{Conclusion}\label{sec-con} 

In this paper, we introduce new robust portfolio selection models by incorporating a reward and penalty mechanism into portfolio management. These models also serve as frameworks that balance portfolio loss and the associated downside risk in portfolio management. To determine optimal allocations for these models, we derive an explicit closed-form expression for the worst-case CVaR under the reward-penalty mechanism. This result is of independent interest and generalizes several existing results regarding the worst-case CVaR. Moreover, the approach used to prove this result has the potential to derive the worst-case CVaR for other loss functions, which we plan to explore in future research. 
  The solutions for the optimal allocations in the proposed models are easily implementable in practice, as they either have explicit closed-form expressions under a classical mean-covariance-based uncertainty set or simplify to certain convex optimization problems. Our empirical experiments not only illustrate that portfolio performance could  be improved by using the models proposed in this paper, but also highlight the importance of incorporating downside risk into portfolio loss for effective portfolio risk management.
  
We believe that the ideas and models presented in this paper can inspire further research in
this area. In particular, the models and formulations discussed here may be extended in several
directions. For example, one may incorporate asymmetric thresholds for rewards and penalties across
different contexts, or adopt step-wise or non-linear structures for the reward-penalty mechanism. These
extensions deserve to be explored further in future research.
However, we note that these extensions would lead to an objective function that is neither concave nor convex, making the optimization 
problem challenging. Several methods have been developed in the literature to address optimization problems with non-concave 
or non-convex objectives. When the objective involves quantile-based risk measures such as VaR, CVaR, or distortion risk measures, 
a possible approach is to apply the concave or convex envelope technique. We refer to  \cite{CLM25}  and the references therein for applications of this technique to non-convex optimization problems.
A more challenging and promising direction
is to reformulate the models using divergence-based uncertainty sets instead of the moment-based sets
considered in this paper. Additionally, the worst-case CVaR results developed in this paper   could be applied
to robust optimal insurance or reinsurance problems.

   \section*{Acknowledgment} 
   The authors thank the editor and three anonymous referees for
their constructive comments and suggestions on an earlier version of
this paper.    J. Cai  acknowledges financial support from the Natural Sciences and Engineering Research Council of Canada (RGPIN-2022-03354).   T. Mao acknowledges financial support from the National Natural Science Foundation of China (Grant No. 12371476) and  Key Program for Youth Innovation of the University of Science and Technology of China.

\appendix\label{A}
 \section{Appendix} \label{appendix}

         \begin{lemma}[\citet{chen2011tight}]
\label{le:A=B}
    For any given  $n$-dimensional vectors $\boldsymbol{\omega} \in \R^n$ and   $\boldsymbol{\mu} \in \R^n$,  and  any given  $n \times n$  positive definite matrix  $\bf{\Sigma}$, it holds that $A= B$ for  
    $ A = \big \{ F_{\boldsymbol{\omega}^{\top} \boldsymbol{X}}: \ \mathbb{E}[\boldsymbol{X}]=\boldsymbol{\mu}, \ {\rm Cov}(\boldsymbol{X})=\bf{\Sigma} \big \}
    $
    and 
  $  B = \big \{F_X: \ \mathbb{E}[X]=\boldsymbol{\omega}^{\top}\boldsymbol{\mu}, \ {\rm Var}(X)=\boldsymbol{\omega}^{\top}\bf{\Sigma}\boldsymbol{\omega} \big \}  $.    
%\hfill $\Box$
  \end{lemma}

 \begin{lemma}
\label{le-minimizer} 
 For any $ \boldsymbol{\omega}\not=0$, the following results  hold: 
 \bn[(i)]
\item  The maximizer $\boldsymbol{\Sigma}^*$  for the problem of 
$\sup_{\boldsymbol{\Sigma}\in \mathcal{U}_{\hat{\boldsymbol{\Sigma}}}} \boldsymbol{\omega}^{\top}\boldsymbol{\Sigma}\bm{\omega}$ 
 is given by 
  \be
\label{maxer-Sigma}
\boldsymbol{\Sigma}^* 
=   \arg\sup_{\boldsymbol{\Sigma}\in \mathcal{U}_{\hat{\boldsymbol{\Sigma}}}} \boldsymbol{\omega}^{\top}\boldsymbol{\Sigma}\bm{\omega}   = \frac{\gamma_2\boldsymbol{\omega}\boldsymbol{\omega}^{\top}}{\|\boldsymbol{\omega}\|_{2}^{2}} +\hat{\boldsymbol{\Sigma}}.
 \ee
 
 \item 
  The maximizer  $\boldsymbol{\mu}^* $ for the problem of
$\sup_{\boldsymbol{\mu}\in\mathcal{U}_{\hat{\boldsymbol{\mu}}}} \boldsymbol{\omega}^{\top}\boldsymbol{\mu}$
 is given by
\begin{equation}\label{optimalmu}
    \boldsymbol{\mu}^*= \arg \sup_{\boldsymbol{\mu}\in\mathcal{U}_{\hat{\boldsymbol{\mu}}}} \boldsymbol{\omega}^{\top}\boldsymbol{\mu}  = \hat{\boldsymbol{\mu}}+\frac{\sqrt{\gamma_1} \, \hat{\boldsymbol{\Sigma}}\boldsymbol{\omega}}{\sqrt{\boldsymbol{\omega}^{\top}\hat{\boldsymbol{\Sigma}}\boldsymbol{\omega}}}.
\end{equation} 

\item It holds that 
\be
\label{min-cons}
           \min_{G \in \mathbb{D}_n(\gamma_1,\gamma_2)} \mathbb{E}[-\boldsymbol{\omega}^{\top}\boldsymbol{X}^G] 
            = -\hat{\boldsymbol{\mu}}^{\top} {\boldsymbol{\omega}}-\sqrt{\gamma_1}\sqrt{\boldsymbol{\omega}^{\top} \hat{\boldsymbol{\Sigma}}{\boldsymbol{\omega}}}.
       \ee  \en 
\end{lemma} 
\begin{proof} 
(i)  The proof of \eqref{maxer-Sigma} is given in \citet{kang2019data}. 

(ii) The maximizer  $\boldsymbol{\mu}^* $ for the problem of 
$\sup_{\boldsymbol{\mu}\in\mathcal{U}_{\hat{\boldsymbol{\mu}}}} \boldsymbol{\omega}^{\top}\boldsymbol{\mu}$ is the same as 
the minimizer $\boldsymbol{\mu}^* $ for the problem 
$\min_{\boldsymbol{\mu}\in\mathcal{U}_{\hat{\boldsymbol{\mu}}}}-\boldsymbol{\omega}^{\top}\boldsymbol{\mu}$.  
Note that the problem in (ii) is an optimization problem with an inequality constraint. By applying KKT optimality conditions  (see,  for example,  Chapter 5.5.3 of  \citet{boyd2004convex} for details) to this problem, we define the Lagrangian for the  minimization problem $\min_{\boldsymbol{\mu}\in\mathcal{U}_{\hat{\boldsymbol{\mu}}}}-\boldsymbol{\omega}^{\top}\boldsymbol{\mu}$
 as $L(\boldsymbol{\mu},\lambda)=-\boldsymbol{\omega}^{\top}\boldsymbol{\mu}+\lambda((\boldsymbol{\mu}-\hat{\boldsymbol{\mu}})^\top\hat{\boldsymbol{\Sigma}}^{-1}(\boldsymbol{\mu}-\hat{\boldsymbol{\mu}})-\gamma_1),$  where $\lambda\in\mathbb{R}^+=(0, \infty)$ is the Lagrangian multiplier. The first-order condition then gives
\begin{equation*}
\left\{
\begin{aligned}
 \frac{\partial {L}(\boldsymbol{\mu},\lambda)}{\partial \boldsymbol{\mu}}&=-\boldsymbol{\omega}+2\lambda\hat{\boldsymbol{\Sigma}}^{-1}(\boldsymbol{\mu}-\hat{\boldsymbol{\mu}})=0,\\
 \frac{\partial {L}(\boldsymbol{\mu},\lambda)}{\partial \lambda}&=(\boldsymbol{\mu}-\hat{\boldsymbol{\mu}})^\top \hat{\boldsymbol{\Sigma}}^{-1}(\boldsymbol{\mu}-\hat{\boldsymbol{\mu}})-\gamma_1=0.\\
\end{aligned}
\right.
\end{equation*}
Solving the above two equations yields \eqref{optimalmu} with 
 $\lambda=\frac{\sqrt{\bm{\omega}^{\top}\bm{\Sigma}\bm{\omega}}}{2\sqrt{\gamma_1}}>0$.

(iii)   We have
\by
           \min_{G \in \mathbb{D}_n(\gamma_1,\gamma_2)} \mathbb{E}[-\boldsymbol{\omega}^{\top}\boldsymbol{X}^G] 
        &=& \min_{(\boldsymbol{\mu},\boldsymbol{\Sigma})\in \mathcal{U}_{(\hat{\boldsymbol{\mu}},\hat{\boldsymbol{\Sigma}})}} \min_{G \in \mathcal{M}(\boldsymbol{\mu},\boldsymbol{\Sigma})} \mathbb{E}[-\boldsymbol{\omega}^{\top}\boldsymbol{X}^G]   \nonumber  \\
        &=&\min_{(\boldsymbol{\mu},\boldsymbol{\Sigma})\in \mathcal{U}_{(\hat{\boldsymbol{\mu}},\hat{\boldsymbol{\Sigma}})}}  
        -\boldsymbol{\omega}^{\top}\boldsymbol{\mu}  \nonumber  \\
        &=& \min_{\boldsymbol{\mu}\in\mathcal{U}_{\hat{\boldsymbol{\mu}}}}-\boldsymbol{\omega}^{\top}\boldsymbol{\mu}    
        \ = \  -\boldsymbol{\omega}^{\top} \boldsymbol{\mu}^* \ = \ -\hat{\boldsymbol{\mu}}^{\top}\hat{\boldsymbol{\omega}}-\sqrt{\gamma_1}\sqrt{\boldsymbol{\omega}^{\top} \hat{\boldsymbol{\Sigma}}{\boldsymbol{\omega}}},
       \ey
       where the last equality follows from \eqref{optimalmu}. 
Therefore, we complete the proof.
\end{proof}

\begin{proof}[\textbf{Proof of Lemma \ref{le-3.1}}]  Note that the objective function in problem  \eqref{eq:0401-1} can be re-expressed as the following form:
\be
\label{obj=obj}
 &&  \frac{s-\alpha}{1-\alpha}(1- \beta) x_2 + \frac{1-s}{1-\alpha} \left((1-\beta) d + (1-\lambda)(x_3-d) \right) \nonumber \\
& =&    \frac{1-s}{1-\alpha} (\lambda-\beta) d+ \frac{s-\alpha}{1-\alpha} (1-\beta) x_2 + \frac{1-s}{1-\alpha} (1-\lambda) x_3.  
\ee
To solve problem \eqref{eq:0401-1},  we first consider the following  weaker  version of this problem: 
\begin{align}\label{prob-upper}
\max_{(x_1,x_2,x_3,s) \in \R^3 \times [\alpha, 1]} ~~~&\,  \frac{1-s}{1-\alpha} (\lambda-\beta) d+ \frac{s-\alpha}{1-\alpha} (1-\beta) x_2 + \frac{1-s}{1-\alpha} (1-\lambda) x_3  \\
{\rm s.t.}~~~~~~~~~~~~& \alpha x_1+ (s-\alpha) x_2 + (1-s) x_3=\mu,\notag\\
& \alpha x_1^2+ (s-\alpha) x_2 ^2+ (1-s) x_3^2=\mu^2+\sigma^2.\notag 
\end{align} 
The maximum value of the objective function in problem \eqref{prob-upper} 
 serves as an upper bound for the maximum value of 
the objective function in problem \eqref{eq:0401-1},  since the  constraints in  problem \eqref{eq:0401-1}  are more restrictive 
than those in problem \eqref{prob-upper},  and the objective functions in both problems are the same. 
For any $s \in [\alpha, 1]$, define the  function $h(u)$ for $0 \le u \le 1$ as 
$$h(u) = 0 \mathds{1}_{\{u <\alpha\}} + (1-\beta) \mathds{1}_{\{ \alpha\le u<s\}}  + (1-\lambda) \mathds{1}_{\{s\le u\le 1\}}, \ \ ~u\in[0,1].$$ 
Thus,  the sum of the last two terms in the objective function of  problem \eqref{prob-upper} has the following expression: 
\begin{align} \label{eq:0411-3}
  \frac{s-\alpha}{1-\alpha} (1-\beta) x_2 + \frac{1-s}{1-\alpha}  (1-\lambda) x_3  
 =\frac1{1-\alpha}\int_0^1 {\rm VaR}_u(X) h(u)d u ,   
\end{align}
where $X$ is a three-point  random variable with the probability function $[x_1, \alpha; x_2, s-\alpha; x_3, 1-s]$. 

Note that for any feasible solution $(x_1,x_2,x_3,s)$ of  problem \eqref{prob-upper},
we have 
\begin{align}
   &\int_0^1 {\rm VaR}_u(X) h(u)d u\\   
=&  \int_0^1 ({\rm VaR}_u(X) - \mu)h(u)d u  +    \mu [(1-\beta)(s-\alpha) + (1-\lambda)(1-s)]\notag\\
=&  \int_0^1 ({\rm VaR}_u(X) - \mu) [h(u)-c(s)]d u  +    \mu [(1-\beta)(s-\alpha) + (1-\lambda)(1-s)]\notag\\
\le& \,  \Big ( \int_0^1 ({\rm VaR}_u(X) - \mu)^2 d u \int_0^1 [h(u) -c^*(s)]^2d u \Big )^{1/2}  +    \mu [(1-\beta)(s-\alpha) + (1-\lambda)(1-s)] \notag\\
=& \,  \sigma\sqrt{  (s-\alpha) (1-\beta)^2 + (1-s)(1-\lambda)^2 -(c^*(s))^2 } +    \mu [(1-\beta)(s-\alpha) + (1-\lambda)(1-s)], \label{eq:0411-2}
\end{align}
where $c^*(s)=(1-\beta)(s-\alpha) + (1-\lambda)(1-s)$, 
 the inequality follows from the Cauchy-Schwarz inequality, and the inequality becomes an equality if ${\rm VaR}_u(X)$ and $h(u)$ are linear dependent, namely,  if ${\rm VaR}_u(X) =k h(u) +\ell$ for some $k\ge 0$ and $\ell \in\R$. Note that there exist $k\ge 0$ and $\ell \in\R$ such that  $x_1=\ell$, $x_2=(1-\beta) k+\ell$,   and $x_3=(1-\lambda) k+\ell$ satisfy
\begin{align*}
\begin{cases}
\alpha x_1+ (s-\alpha) x_2 + (1-s) x_3=\mu,  \\
 \alpha x_1^2+ (s-\alpha) x_2 ^2+ (1-s) x_3^2=\mu^2+\sigma^2,
 \end{cases}
 \end{align*}
 that is,
 \begin{align*}
\begin{cases}
\alpha \ell + (s-\alpha) ((1-\beta) k+\ell) + (1-s)  ((1-\lambda) k+\ell)=\mu,  \\
 \alpha \ell^2+ (s-\alpha)  ((1-\beta) k+\ell)  ^2+ (1-s)  ((1-\lambda) k+\ell) ^2=\mu^2+\sigma^2.
 \end{cases}
 \end{align*}
 By solving the above  system of  equations, we have that  the solution $(\ell=\ell(s), k=k(s))$ to this system 
  is given by 
 $$
 \begin{cases}
  \ell = \mu -  \frac{ \displaystyle\sigma [(s-\alpha)  (1-\beta)    + (1-s)   (1-\lambda) ]}{\displaystyle \sqrt{(s-\alpha) (1-\beta) ^2+ (1-s)   (1-\lambda)  ^2-[(s-\alpha)  (1-\beta)    + (1-s)   (1-\lambda) ]^2 }},\\
k=\frac{ \displaystyle\sigma }{ \displaystyle\sqrt{(s-\alpha) (1-\beta) ^2+ (1-s)   (1-\lambda)  ^2 -[(s-\alpha)  (1-\beta)    + (1-s)   (1-\lambda) ]^2 }}.
 \end{cases}
 $$
This  implies that  the upper bound for $\int_0^1 {\rm VaR}_u(X) h(u)d u$  given in  
\eqref{eq:0411-2} can be attained. Substituting  $x_1=\ell$, $x_2=(1-\beta) k+\ell$,   and $x_3=(1-\lambda) k+\ell$ 
into
   the objective function of \eqref{prob-upper}, we have
that    problem \eqref{prob-upper}  simplifies to the  one-variable   unconstrained  maximization problem \eqref{max-f}.   
 \end{proof} 
  
\begin{proof}[\textbf{Proof of Theorem \ref{Prop-wcTVaR-institutional}}]
         \begin{itemize}
\item [(a)]  Assume $\lambda \ge \beta$ or equivalently  $0 < 1-\lambda\le 1-\beta \le 1$. In this case,  $\ell_d(x)=(1-\beta) x+ (\beta-\lambda) (x - x \wedge d) = (1-\lambda) x + (\lambda-\beta) (x \wedge d)$ is increasing and concave in $x\in\R$.  We  show that 
the supremum value in  \eqref{f-wcTVaR} can be attained  at  a two-point distribution. 
 In doing so, we first show that for any distribution  $F\in   \mathcal S (\mu, \sigma)$, there exists a two-point  random variable $X_\epsilon$  with a two-point distribution  
$F_\epsilon \in S (\mu, \sigma)$ satisfying 
\begin{align}
\label{TVaR<two-point}
  \text{CVaR}_\alpha (\ell_d(X^F)) & \le  \text{CVaR}_\alpha (\ell_d(X_\epsilon^{F_\epsilon})).
\end{align}
To do so,  we consider  the following two cases: 
          \begin{itemize}
\item [(a1)] 
 If $d\le {\rm VaR}_\alpha(X)$, define  $U$ as a uniform random variable over interval $(0, 1)$  such that $U$ and $X$ are comonotonic, and define  the following  two-point random variable:  
\begin{equation}
\label{X_e}
 X_\epsilon =  (e_1- \epsilon \,  \overline{\alpha}) 1_{\{U\le \alpha \}} +  (e_2+ \epsilon  \alpha) 1_{\{U>\alpha \}} ,
\end{equation} 
where $e_1 = \E[X|U\le \alpha] $, $e_2= \E[X|U>\alpha]\ge  {\rm VaR}_\alpha(X) \ge d$, $\overline{\alpha} = 1-\alpha$, and  $\epsilon\ge 0$.  It is not hard to verify  that  
$\E[X_\epsilon] =\mu$ for any $\epsilon\ge 0$,   and there exists an $\epsilon\ge 0$ such  that $\E[X_\epsilon^2] =\mu^2+\sigma^2$.  Then, we have 
   %\begin{align*}
 $ \text{CVaR}_\alpha (\ell_d(X^F))   = \ell_d(e_2) \le \ell_d(e_2+ \epsilon  \alpha) = \text{CVaR}_\alpha (\ell_d(X_\epsilon^{F_\epsilon})),
 $
%\end{align*}
where $F_\epsilon$ is the two-point  distribution of $X_\epsilon$, and  the first equality follows from  the
 facts that $\p(\ell_d(X) =\ell_d(e_2)) = 1-\alpha$ and  $\ell_d(x)$ is increasing in $x$. 
       
       \item [(a2)]
 If $d> {\rm VaR}_\alpha(X)$, 
 with  the two-point  random variable $X_\epsilon$ defined in \eqref{X_e},  we have 
 \begin{align*}
  \text{CVaR}_\alpha (\ell_d(X^F))  &=   \E[\ell_d(X) |U>\alpha ]\le \ell_d( \E[X|U>\alpha ]) \\
  &=  \ell_d(e_2)\le\ell_d(e_2+ \epsilon  \alpha)= \text{CVaR}_\alpha (\ell_d(X_\epsilon^{F_\epsilon})),
\end{align*}
where the first inequality follows from  Jensen's inequality 
 and the fact that $\ell_d(x)$ is  concave in $x$.
           \end{itemize}
    Combining the above two cases, we see that \eqref{TVaR<two-point} holds.   Therefore, 
    \begin{align*}
\sup_{F \in \mathcal{S}(\mu, \sigma)}  \text{CVaR}_\alpha (\ell_d(X^F)) & = \sup_{F \in  \mathcal{S}_2^*(\mu, \sigma)}  \text{CVaR}_\alpha (\ell_d(X^{F})),
\end{align*}    
    where   $\mathcal{S}_2^*(\mu, \sigma)$ is a subset of $ \mathcal{S}(\mu, \sigma)$ and contains only two-point distributions. It is easy to 
    verify that        
\begin{equation}
\label{S_2}
\mathcal{S}_2^*(\mu, \sigma) = 
\big \{ [\mu-\sigma \sqrt{(1-\alpha)/\alpha}, \alpha;\, \mu+\sigma \sqrt{\alpha/(1-\alpha)}, 1-\alpha] \big \}
\end{equation}
 contains  only one two-point  distribution with the probability function specified in \eqref{S_2}.  
 Thus,   $\sup_{F \in  \mathcal{S}_2^*(\mu, \sigma)}  \text{CVaR}_\alpha (\ell_d(X^{F})) = \text{CVaR}_\alpha (\ell_d(X^{F^*}))$, where 
 $F^*$ is a two-point distribution with the two-point probability function specified in \eqref{S_2}.  A simple calculation of $\text{CVaR}_\alpha (\ell_d(X^{F^*}))$ yields $ \text{CVaR}_\alpha (\ell_d(X^{F^*}))  = \ell_d\big (\mu+\sigma\sqrt{\frac \alpha{1-\alpha} } \,  \big )$, which means that 
\eqref{ga} holds.

\item [(b)] Assume $\lambda < \beta$ or equivalently  $0 \le 1-\beta < 1-\lambda \le  1$. In this case,  $\ell_d(x)=(1-\beta)\min\left\{x,d\right\}+(1-\lambda)(x-d)_{+}$ is increasing and convex in $x\in\R$, and we  first  show that for any distribution  $F\in   \mathcal S (\mu, \sigma)$, there exists a 
 three-point distribution
$F^* \in S (\mu, \sigma)$ satisfying 
\begin{align}
\label{TVaR<three-point}
  \text{CVaR}_\alpha (\ell_d(X^F)) & \le  \text{CVaR}_\alpha (\ell_d(X^{F^*})). 
\end{align}

In fact, 
 if $d\le {\rm VaR}_\alpha(X)$, following  the same arguments used  in (a1) above, we  see   that there exists a   two-point distribution 
$F^* \in \mathcal{S}(\mu, \sigma)  $ with a probability function 
$ [\mu-\sigma \sqrt{(1-\alpha)/\alpha}, \,  \alpha;\, \mu+\sigma \sqrt{\alpha/(1-\alpha)}, \, 1-\alpha]$ 
 such that  $\text{CVaR}_\alpha (\ell_d(X^F)) \le  \text{CVaR}_\alpha (\ell_d(X^{F^*}))$.
   Additionally,  if $d> {\rm VaR}_\alpha(X)$, 
define   a uniform random variable  $U$ over interval $(0, 1)$ such that $U$ and $X$ are comonotonic, and define the following 
 three-point random variable:  
 \begin{align*}
 X_\epsilon 
 &=  (e_1- \epsilon  \, \overline{\alpha}_d) 1_{\{U\le \alpha\}} +   e_2  1_{\{U>\alpha, \, X\le d \}} +  (e_3+\epsilon  \alpha ) 1_{\{ X>d \}}   ,
\end{align*}
where $e_1 = \E[X|U\le \alpha]$, $e_2= \E[X|U>\alpha, X\le d]$,   $e_3= \E[X|X>d]$,
$\overline{\alpha}_d = \p(X>d)$, and  $\epsilon\ge 0$.  It is not hard to verify that  $\E[X_\epsilon] =\mu$ for any $\epsilon\ge0$ and   there exists $\epsilon\ge0$ such that $\E[X_\epsilon^2] =\mu^2+\sigma^2$ and  
%\begin{align*}
$
  \text{CVaR}_\alpha (\ell_d(X^F))   \le \text{CVaR}_\alpha (\ell_d(X_\epsilon^{F_\epsilon})).
  $
%\end{align*}
      Note that  a two-point point distribution can be also viewed as a three-point  distribution by  Definition  \ref{de-k-ds}. Thus, 
   combining the above discussions,  
  we see that  \eqref{TVaR<three-point} holds. Therefore,  
    \begin{align*}
\sup_{F \in \mathcal{S}(\mu, \sigma)}  \text{CVaR}_\alpha (\ell_d(X^F)) & = \sup_{F \in  \mathcal{S}_3^*(\mu, \sigma)}  \text{CVaR}_\alpha (\ell_d(X^{F})),
\end{align*}  
  where   $\mathcal{S}_3^*(\mu, \sigma)$ is a subset of $\mathcal{S}(\mu, \sigma)$ and consists of all three-point distributions 
with  a  probability function $[x_1,\alpha;  x_2,s-\alpha;  x_3, 1-s]$ that satisfies  $x_1\le x_2\le d\le x_3$ and $s\in [\alpha,1]$. It is easy to see that 
 the problem $\sup_{F\in\mathcal{S}^*_3(\mu, \sigma)}{\rm CVaR}_{\alpha}((1-\beta) X^F+(\beta-\lambda)(X^F-d)_{+})$ is equivalent to problem \eqref{eq:0401-1}.

By calculating the derivative of  the objective function $f(s)$ in problem \eqref{max-f}, we see that the function $f(s)$ 
 is concave in $s\in [\alpha,1]$. Thus, problem \eqref{max-f} can be treated as     a convex minimization problem since 
 $\max_{ s \in [\alpha,1]} f(s) = -\min_{ s \in [\alpha,1]}  -f(s)$.    
   Denote by $(x_1^*,x_2^*,x_3^*, s^*)$ an optimal solution to   problem \eqref{prob-upper}, that is, $s^*$ is an optimal solution to the above problem, $x_1^* = \ell(s^*)$, $x_2^* =k(s^*) (1-\beta)+ \ell(s^*)$ and $x_1^* =k(s^*)(1-\lambda)+ \ell(s^*)$. 
Define $X$ as a  three-point random variable such that $\p(X=x_1^*) =\alpha,$ $\p(X=x_2^*) =s^*-\alpha$,  and $ \p(X=x_3^*) =1-s^*.$
   By the constraints of problem \eqref{eq:0401-1},  we have $\E[X]= \mu$ and ${\rm Var}(X)=\sigma^2$.  We next consider the following three cases.

\begin{itemize}
\item[(b1)]  If $d-\mu \le \sigma d_1$, noting that in this case,  $f(s)$ is decreasing in $[\alpha,1]$, we observe that the optimal $s^*=\alpha$, and $x_3^* = \mu+ \sigma\sqrt{\frac{\alpha}{1-\alpha}} \ge d$.  
This implies that the optimal solution is also a feasible solution to  problem \eqref{eq:0401-1} and thus is the optimal solution  to  problem \eqref{eq:0401-1}. Moreover, the maximum value of  problem \eqref{eq:0401-1} is
$$
\max\limits_{s\in [\alpha,1]} \, f(s)=f(s^*)=(\lambda-\beta)d+	
(1-\lambda)\Big  (\mu+\sigma \, \sqrt{\frac{\alpha}{1-\alpha}} \, \Big ).
$$

 \item[(b2)]  If $d-\mu \in   ( \sigma d_1,  \,  \sigma d_2 )$, noting that in this case,   $f'(\alpha)>0$ and $f'(1)<0$, by the concavity of $f(s)$, there exists a unique maximizer to $f(s)$ in $(\alpha,1)$ denoted by $s^*$. Solving the equation  $f'(s^*)=0$, we obtain  the optimal solution $$s^* = \frac{(d-\mu)\Delta(\mu, \sigma)}{2(\sigma^2+(d-\mu)^2)(\beta-\lambda)} -\frac{(1+2\alpha)(1-\beta)-(1-\lambda)}{2(\beta-\lambda)} \in (\alpha,1), $$

where $\Delta(\mu, \sigma)$ is defined in \eqref{Delta}.
% =\sqrt{\delta((\beta-\lambda)^2+4\alpha(1-\beta)(1-\lambda))}$ and $\delta=\sigma^2+(d-\mu)^2$.
 Further, we have
$$\begin{aligned}
x_3^*&=\mu+\frac{\sigma(s^*(1-\lambda)-(s^*-\alpha)(1-\beta))}{\sqrt{(s^*-\alpha) (1-\beta) ^2+ (1-s^*)   (1-\lambda)  ^2-[(s^*-\alpha)  (1-\beta)    + (1-s^*)   (1-\lambda) ]^2}}\\&=\mu+\frac{2(d-\mu)(s^*(1-\lambda)-(s^*-\alpha)(1-\beta))}{\beta(1-2s^*+2\alpha)+\lambda(2s^*-1)}\\&=d+\frac{(\beta-\lambda)(d-\mu)}{(1-\beta)(1-2s^*+2\alpha)+(1-\lambda)(2s^*-1)} \ \ge \ d,
\end{aligned}$$

where the second equality and the  inequality follow from the first-order derivative condition satisfied by $s^*$, i.e., 
\begin{equation}\notag
    \begin{aligned}
        &\sigma/\sqrt{(s^*-\alpha) (1-\beta) ^2+ (1-s^*)   (1-\lambda)  ^2-[(s^*-\alpha)  (1-\beta)    + (1-s^*)(1-\lambda)]^2}\\
        =&\frac{2(d-\mu)}{(1-\beta)(1-2s^*+2\alpha)+(1-\lambda)(2s^*-1)}\geq 0.
\end{aligned}
\end{equation}
Similarly,  it is easy to  verify that $$x_2^*= d+\frac{(\beta-\lambda)(\mu-d)}{(1-\beta)(1-2s^*+2\alpha)+(1-\lambda)(2s^*-1)} \le d.$$ This implies that the optimal solution is also a feasible solution to  problem \eqref{eq:0401-1} and thus is the optimal solution  to  problem \eqref{eq:0401-1}. Moreover, by substituting the optimal solution $(x_1^*,x_2^*,x_3^*, s^*)$ into the objective function of problem \eqref{eq:0401-1}, we obtain that 
 the maximum value of the objective function of   problem \eqref{eq:0401-1} is given by 
 $$
\begin{aligned}
%\eqref{eq:0401-1}&=
&\frac{1-s^*}{1-\alpha} (\lambda-\beta) d+ \frac{s^*-\alpha}{1-\alpha} (1-\beta) x_2^* + \frac{1-s^*}{1-\alpha} (1-\lambda) x_3^*\\
=&(1-\beta) d+\frac{(2-\lambda-\beta)(d-\mu)}{2(1-\alpha)} \times \frac{(1-s^*)(1-\lambda)-(s^*-\alpha)(1-\beta)}{(1-\beta)(1-2s^*+2\alpha)+(1-\lambda)(2s^*-1)}\\
=&(1-\beta)d+\frac{\Delta(\mu, \sigma)-(2-\lambda-\beta)(d-\mu)}{2(1-\alpha)}.
\end{aligned}
$$ 
\item[(b3)]   If $d-\mu \ge   \sigma  d_2$, noting that  in this case, $f(s)$ is increasing in $[\alpha,1]$, the optimal $s^*=1$, and $x_2^* = \mu +  \sigma\sqrt{\frac{\alpha}{1-\alpha}} \le  d.$ This solution is feasible for problem \eqref{eq:0401-1} and thus is optimal for it. Moreover, the maximum value of  problem \eqref{eq:0401-1} is
$
\max\limits_{s\in [\alpha,1]} \, f(s)=f(s^*)=(1-\beta)\big  (\mu+\sigma \, \sqrt{\frac{\alpha}{1-\alpha}} \big ). $
       Note that if $\beta=1$, then $d_2=\infty$ and thus the case of $d-\mu \ge   \sigma  d_2$ vanishes.
      \end{itemize}    
\end{itemize}
Combining the above two cases, we complete the proof.
\end{proof} 

 \begin{proof}[\textbf{Proof of Proposition \ref{f-convex}}]
  By the closed-form of $f_{\lambda, \beta}$ given by Theorem \ref{Prop-wcTVaR-institutional},  we  can verify  that function $f_{\lambda, \beta}(\mu,\sigma)$ is convex by 
  using  the standard manipulation. Here we present another  proof of this result below. To do so, 
define $h(d)=\sup_{F\in\mathcal{S}(0, 1)}{\rm CVaR}_{\alpha}((1-\beta) X^F+(\beta-\lambda)(X^F-d)_{+})$ for $d\in\R$.
Since $(X^F-d)_{+} $ is a convex function in $d \in \R$ and CVaR is a convex risk measure, we see that 
 $h$ is a convex function in $d \in \R$. Moreover,  we have 
\begin{align*} 
f_{\lambda, \beta}(\mu, \sigma)&=\sup_{F\in\mathcal{S}(\mu, \sigma)}{\rm CVaR}_{\alpha}((1-\beta) X^F+(\beta-\lambda)(X^F-d)_{+}) \\
&=\sup_{F\in\mathcal{S}(\mu, \sigma)}\sigma \, 
{\rm CVaR}_{\alpha} \left( (1-\beta) \frac{X^F-\mu}\sigma+(\beta-\lambda)\left(\frac{X^F-\mu}\sigma-\frac{d -\mu}\sigma\right)_{+} \, \right)  + (1-\beta)\mu\\
&=\sup_{F\in\mathcal{S}(0,1)}\sigma \, 
{\rm CVaR}_{\alpha} \left( (1-\beta) X^F+(\beta-\lambda)\left(X^F-\frac{d -\mu}\sigma\right)_{+} \, \right)  + (1-\beta)\mu\\
&= \sigma h\left( \frac{d -\mu}\sigma \right)  + (1-\beta)\mu.
\end{align*} 
By the convexity of $h$ and Example 3.20 in  Section 3.2.6 of \cite{boyd2004convex}, we have $(\mu,\sigma)\mapsto  \sigma h\big ( \frac{d -\mu}\sigma \big ) $ is convex,  and thus, $f_{\lambda, \beta}(\mu, \sigma)$ is convex.  This completes the proof. 
  \end{proof} 

\begin{proof}[\textbf{Proof of Lemma \ref{le-for-g(s)}}]      First, By Proposition \ref{f-convex}, we have that $f_{\lambda, \beta}(s, t)$ is convex in $(s, t) \in \R \times \R^+$. Furthermore,  by Theorem 
\ref{Prop-wcTVaR-institutional}, we have that $f_{\lambda, \beta}(s,t)$ is also increasing in $t\in \R^+$. Therefore, by the composition rule  in Section 3.2.4 of \cite{boyd2004convex},  we have 
  $g_{\lambda, \beta}(s) =f_{\lambda, \beta}(s,\sqrt{b_0s^2-2b_1s+b_2})$ is  convex  in $s \in \R$ since $ \sqrt{b_0s^2-2b_1s+b_2}= \sqrt{b_0}\times\sqrt{(s-b_1/b_0)^2+(b_2b_0-b_1^2)}$ is convex in $s \in \R$.  Next,   from Theorem 
\ref{Prop-wcTVaR-institutional}(b), we see that the function $g_{\lambda, \beta}$ has the following expression: 
\by
%  \label{gs-TVaR}
       g_{\lambda, \beta}(s)  
       =   
    \begin{cases}
     g_1(s)\triangleq (\lambda-\beta)d+(1-\lambda)\big(s+ \sqrt{b_0s^2-2b_1s+b_2}  \sqrt{\frac{\alpha}{1-\alpha}}\big), \quad &d-s\leq d_1 \sqrt{b_0s^2-2b_1s+b_2} ,\\
      g_2(s) \triangleq    (1-\beta)d+\frac{\Delta\big (s, \, \sqrt{b_0s^2-2b_1s+b_2} \big )-(2-\lambda-\beta)(d-s)}{2(1-\alpha)}, \quad &d_1 \sqrt{b_0s^2-2b_1s+b_2}<d-s \\
        & \ \ \ \ \ \ \ \  <d_2 \sqrt{b_0s^2-2b_1s+b_2},\\
    g_3(s) \triangleq     (1-\beta)\big(s+ \sqrt{b_0s^2-2b_1s+b_2}  \sqrt{\frac{\alpha}{1-\alpha}} \big), \quad &d-s\geq d_2 \sqrt{b_0s^2-2b_1s+b_2} ,
    \end{cases} 
\ey 
where the function $\Delta(s, t)$ is defined in \eqref{Delta}, and the parameters $d_1$ and $d_2$ are defined in \eqref{d1d2}. 
Thus, we can express the function $g_{\lambda, \beta}$ in the following form: 
\be
  \label{gs-TVaR}
       g_{\lambda, \beta}(s)  
       =   
    \begin{cases}
     g_1(s),  \quad &  s \in A_1 \triangleq  \{s \in \R:   d-s\leq d_1 \sqrt{b_0s^2-2b_1s+b_2}\},\\
      g_2(s),  \quad & s \in A_2  \triangleq  \{s \in \R: d_1 \sqrt{b_0s^2-2b_1s+b_2}<d-s <d_2 \sqrt{b_0s^2-2b_1s+b_2} \},\\
    g_3(s), \quad & s \in A_3  \triangleq   \{s \in \R: d-s\geq d_2 \sqrt{b_0s^2-2b_1s+b_2} \}.
    \end{cases} 
              \ee  
By differentiating, we confirm that $g_{\lambda, \beta}^{\prime \prime}(s) >0$ for all $s \in \R$. Thus, $g_{\lambda, \beta}(s)$ is strictly convex. 

Furthermore, by solving the inequalities in the sets $A_1$, $A_2$, and $A_3$, we confirm that  for  each possible relationship among $d$, $b_0$, $b_1$, and $b_2$, there exist $s_1 \le  s_2$ such that  $g_{\lambda, \beta}(s)$  coincides with exactly one of $g_1(s)$, $g_2(s)$, or $g_3(s)$ on $(-\infty, s_1)$, and with exactly one of $g_1(s)$, $g_2(s)$, or $g_3(s)$ on $(s_2,  \infty)$.   
  Note that  the values of $s_1 $ and $s_2$ depend on  the relationships  among $d$, $b_0$, $b_1$, and 
               $b_2$, 
   and  may  have different expressions for different relationships. We omit the     tedious and  detailed discussion  of these  relationships.   
  
  We now show that  the function $g_{\lambda, \beta}(s)$ is coercive on $\R$, i.e.,
$\lim_{|s|\rightarrow\infty}g_{\lambda, \beta}(s)=\infty.$ 
Consider functions $g_1$ and $g_3$ in \eqref{gs-TVaR}, we have that 
 $g_1(s) = 
 %=  (\lambda-\beta)d+(1-\lambda)\big(s+\sqrt{\frac{\alpha}{1-\alpha}}\sqrt{b_0s^2-2b_1s+b_2}\big) =
(\lambda-\beta)d+(1-\lambda)g_0(s)
$ and $g_3(s) =    
 %(1-\beta)\big(s+ \sqrt{b_0s^2-2b_1s+b_2}  \sqrt{\frac{\alpha}{1-\alpha}} \big) =
 (1-\beta) g_0(s)$, where 
$$g_0(s) = s+\sqrt{\frac{\alpha}{1-\alpha}} \, \sqrt{b_0s^2-2b_1s+b_2}=s+\sqrt{\frac{\alpha}{1-\alpha}} \, 
\sqrt{b_0\Big (s-\frac{b_1}{b_0} \Big )^2+ b_2 - \frac{b_1^2}{b_0}}, $$ 
which, together with $b_0 b_2 - b_1^2 >0$ and the  triangle inequality of $|x-y| \ge |x|-|y|$, implies   that  for any $s \in \R$,
\begin{align*}
    g_0(s)& \ > \ -|s|+ \sqrt{\frac{\alpha}{1-\alpha}} \, \sqrt{b_0}\, \Big | s-\frac{b_1}{b_0} \Big |    
  % \\
%    &\geq -|s|+\sqrt{\frac{\alpha}{1-\alpha}} \, \sqrt{b_0}\Big (|s|- \Big |\frac{b_1}{b_0} \Big | \Big )\\
  \  \ge \  \Big(-1+\sqrt{\frac{\alpha}{1-\alpha}}\sqrt{b_0}\Big)|s| - \sqrt{\frac{\alpha}{1-\alpha}} \, \sqrt{b_0}\,  \Big |\frac{b_1}{b_0} \Big |.
\end{align*}
Hence,   $g_0(s)$ is corecive, i.e.,  $\lim_{|s|\rightarrow\infty}g_0(s)=\infty $ if $-1+\sqrt{\frac{\alpha}{1-\alpha}}\sqrt{b_0} >0$ or equivalently   $b_0+1 > \frac{1}{\alpha}$. Therefore, 
both of $g_1(s)=
(\lambda-\beta)d+(1-\lambda)g_0(s)
$ and $g_3(s) = (1-\beta) g_0(s)$ are coercive  if    $b_0+1 > \frac{1}{\alpha}$.

Similarly, consider function $g_2$ in \eqref{gs-TVaR},   by the definition of  the function $\Delta(s, t)$ given  in \eqref{Delta} and denoting  $C$ by 
$
C=(2(1-\alpha))^{-1}   \sqrt{(\beta-\lambda)^2+4\alpha(1-\beta)(1-\lambda)},$   we have $C>0$ and  
\by
g_2(s) 
%&=&  (1-\beta)d+\frac{\Delta\big (s, \, \sqrt{b_0s^2-2b_1s+b_2} \big )-(2-\lambda-\beta)(d-s)}{2(1-\alpha)} \\
&=& (1-\beta)d+ C \sqrt{b_0s^2-2b_1s+b_2+(d-s)^2}
    -\frac{(2-\lambda-\beta)(d-s)}{2(1-\alpha)}\\
    &=& (1-\beta)d -   \frac{(2-\lambda-\beta)d}{2(1-\alpha)}+g^*(s),
        \ey 
   where
     \begin{align*}
    g^*(s) & =  C \sqrt{b_0s^2-2b_1s+b_2+(d-s)^2}
    +\frac{(2-\lambda-\beta)s}{2(1-\alpha)} \\
  %  &= C \sqrt{(b_0+1)s^2-2(b_1+d)s+b_2+d^2}
  %  +\frac{(2-\lambda-\beta)s}{2(1-\alpha)} \\
        &= C \sqrt{(b_0+1) \Big (s-\frac{b_1+d}{b_0+1} \Big )^2+\frac{b_0b_2-b_1^2+b_0d^2-2b_1d+b_2}{b_0+1}}+\frac{(2-\lambda-\beta)s}{2(1-\alpha)}.
        \end{align*}
 Note that $b_0b_2-b_1^2>0$, thus for all $d \in\R$, it holds that $b_0d^2-2b_1d +b_2>0$.  Thus, 
 %Together, these conditions 
%$b_0>0$ and $b_2>0$, 
%imply}
\begin{align*}
    g^*(s)  &  >  C \sqrt{(b_0+1)} \, \Big | s-\frac{b_1+d}{b_0+1} \Big |+\frac{(2-\lambda-\beta)s}{2(1-\alpha)} \\
  %  &\ge C \sqrt{(b_0+1)} \, \Big  ( | s| - \Big | \frac{b_1+d}{b_0+1} \Big | \Big )+\frac{(2-\lambda-\beta)s}{2(1-\alpha)} \\
   %   & > C   \, \sqrt{(b_0+1)\,  \Big (s-\frac{b_1+d}{b_0+1} \Big )^2}+\frac{(2-\lambda-\beta)s}{2(1-\alpha)}\\
    & \ge   \Big ( C \sqrt{(b_0+1)}-\frac{2-\lambda-\beta}{2(1-\alpha)} \Big ) |s| - C \sqrt{(b_0+1)} \,  \Big | \frac{b_1+d}{b_0+1} \Big |.
 %   &=\Big(\frac{\sqrt{(\beta-\lambda)^2+4\alpha(1-\beta)(1-\lambda)}}{2(1-\alpha)}\sqrt{(b_0+1)}-\frac{2-\lambda-\beta}{2(1-\alpha)}\Big)|s|.
\end{align*}
Hence, $g^*(s)$ is coercive, i.e.,  $\lim_{|s|\rightarrow\infty}g^*(s)=\infty $ if  $C \sqrt{(b_0+1)}-\frac{2-\lambda-\beta}{2(1-\alpha)}>0$ or equivalently 
 $b_0+1>\frac{(2-\lambda-\beta)^2}{(\beta-\lambda)^2+4\alpha(1-\beta)(1-\lambda)}$. Therefore, $g_2(s) =(1-\beta)d -   \frac{(2-\lambda-\beta)d}{2(1-\alpha)}+g^*(s) $ is coercive if  $b_0+1>\frac{(2-\lambda-\beta)^2}{(\beta-\lambda)^2+4\alpha(1-\beta)(1-\lambda)}$.

In conclusion, as long as both  conditions   $b_0+1 > \frac{1}{\alpha}$     and $b_0+1>\frac{(2-\lambda-\beta)^2}{(\beta-\lambda)^2+4\alpha(1-\beta)(1-\lambda)}$ hold, or equivalently the condition in  \eqref{cons-for-g} holds, 
 it follows   that $\lim_{|s|\rightarrow\infty}g_{\lambda, \beta}(s)=\infty$, i.e., $g_{\lambda, \beta}(s)$ is coercive. Therefore, the strictly convex function 
  $g_{\lambda, \beta}(s)$ has a unique minimizer on $\R$. 
   \end{proof}

\begin{proof}[\textbf{Proof of Theorem \ref{CVaR-MV-PR}}]
By  Lemma \ref{le:A=B} and the definition of the function $\ell_d$,   we have that the inner optimization problem in \eqref{portfolio CVaR} has the following equivalence:
\begin{equation}
\label{port-tvar-f}
\begin{aligned}
     &\sup_{G \in \mathcal{M}(\boldsymbol{\mu},\boldsymbol{\Sigma})}{\rm CVaR}_{\alpha} 
     \big((1-\beta)\boldsymbol{\omega}^{\top}\boldsymbol{X}^G+(\beta-\lambda)(\boldsymbol{\omega}^{\top}\boldsymbol{X}^G-d)_{+}\big)
     \\
     = &\sup_{F \in \mathcal{S}(\boldsymbol{\omega}^{\top}\boldsymbol{\mu}, \ \boldsymbol{\omega}^{\top}\boldsymbol{\Sigma} \boldsymbol{\omega})}{\rm CVaR}_{\alpha}\big(\ell_{d}(X^F)\big) 
     =  \ f_{\lambda, \beta}(\boldsymbol{\omega}^{\top}\boldsymbol{\mu}, \ \boldsymbol{\omega}^{\top}\boldsymbol{\Sigma} \boldsymbol{\omega}), 
  \end{aligned}
\end{equation}                                                                
where the function $f_{\lambda, \beta}(s, t)$ is defined in \eqref{TVaR-f} and   solved in Theorem 
\ref{Prop-wcTVaR-institutional}(b).
 Thus, problem \eqref{port-tvar-f} simplifies to the following 
finite-dimensional optimization problem: 
\begin{equation*}
    \begin{aligned}
        &\min_{\boldsymbol{\omega}\in\mathbb{R}^n}\left\{f_{\lambda, \beta}(\boldsymbol{\omega}^{\top}\boldsymbol{\mu},\sqrt{\boldsymbol{\omega}^{\top}\boldsymbol{\Sigma}\boldsymbol{\omega}}):\boldsymbol{\omega}^{\top}\boldsymbol{e}=1\right\}
        =\min_{s\in\mathbb{R}}\min_{\boldsymbol{\omega}}\left\{f_{\lambda, \beta}(s,\sqrt{\boldsymbol{\omega}^{\top}\boldsymbol{\Sigma}\boldsymbol{\omega}}):\boldsymbol{\omega}^{\top}\boldsymbol{e}=1, \boldsymbol{\omega}^{\top}\boldsymbol{\mu}=s\right\}.
    \end{aligned}
\end{equation*}

From Theorem \ref{Prop-wcTVaR-institutional}(b), we see that 
for any given $s$, the function $f_{\lambda, \beta}(s, t)$ is increasing in $t$, and 
$\sqrt{\boldsymbol{\omega}^{\top}\boldsymbol{\Sigma}\boldsymbol{\omega}}$ is increasing in 
$\boldsymbol{\omega}^{\top}\boldsymbol{\Sigma}\boldsymbol{\omega}$,  thus, the minimizer $\boldsymbol{\omega}^*_s$ for the problem 
 $\min \{f_{\lambda, \beta}(s,\sqrt{\boldsymbol{\omega}^{\top}\boldsymbol{\Sigma}\boldsymbol{\omega}}):\boldsymbol{\omega}^{\top}\boldsymbol{e}=1,\boldsymbol{\omega}^{\top}\boldsymbol{\mu}=s\} $ is reduced to 
\be
        \boldsymbol{\omega}^*_s
      &=& \argmin_{\boldsymbol{\omega}}\left\{\boldsymbol{\omega}^{\top}\boldsymbol{\Sigma}\boldsymbol{\omega}:\boldsymbol{\omega}^{\top}\boldsymbol{e}=1,\boldsymbol{\omega}^{\top}\boldsymbol{\mu}=s\right\} 
           =  \begin{pmatrix}
            \boldsymbol{\Sigma}^{-1}\boldsymbol{\mu} & \boldsymbol{\Sigma}^{-1}\boldsymbol{e}
        \end{pmatrix}\begin{pmatrix}
            b_0 & -b_1\\-b_1 & b_2
        \end{pmatrix}\begin{pmatrix}
            s \\ 1
        \end{pmatrix}, \label{w-solution-MV-model-TVaR}
    \ee
  where expression \eqref{w-solution-MV-model-TVaR} follows from Lemma \ref{sol-MV-model}.

Note that the  variance of the portfolio loss  under the allocation $ \boldsymbol{\omega}^*_s$ is  ${\boldsymbol{\omega}^*_s}^{\top}\boldsymbol{\Sigma}\boldsymbol{\omega}^*_s=b_0s^2-2b_1s+b_2$, which is a function of $s$. Therefore, the optimal solution  $ \boldsymbol{\omega}^*$ to the original problem reduces to 
$ \boldsymbol{\omega}^*_{s^*}$, where $s^*$ is the solution to  the following single-variable optimization problem: 
$$
\min_{s\in\mathbb{R}} f_{\lambda, \beta}(s,\sqrt{b_0s^2-2b_1s+b_2}),
$$
 and   given by
$
        s^*=  \argmin_{s\in\mathbb{R}} f_{\lambda, \beta}(s,\sqrt{b_0s^2-2b_1s+b_2})  =   \argmin_{s\in\mathbb{R}} g_{\lambda, \beta}(s).  
        $
%   where    
%  $   g_{\lambda, \beta}(s)  =   f_{\lambda, \beta}(s,\sqrt{b_0s^2-2b_1s+b_2})$.
        By Lemma \ref{le-for-g(s)}, we have that $g_{\lambda, \beta}(s)$ is coercive  and strictly convex in $s \in \R$. Thus, $g_{\lambda, \beta}(s)$ has a unique minimizer  $s^*$ 
       on $\R$. In fact, $s^*$ is the unique solution to the equation of $g_{\lambda, \beta}^\prime (s)=0$.       
         \end{proof}

\begin{proof}[\textbf{Proof of Theorem \ref{TVaR-MV-PR-a=0}}]
 Note  that $1-\lambda= (\kappa+1)(1-\beta)$ and that $d_1 \to -\infty$ and  $d_2 \to \infty$ as $\alpha \to 0$. Thus, by 
Theorem \ref{Prop-wcTVaR-institutional}(b), we see that as $\alpha \to 0$, the function $g_{\lambda, \beta}(s)$ defined in Theorem \ref{sol-MV-model}  simplifies to 
the expression of the function $g_2(s)$ in \eqref{gs-TVaR}  for   any $s \in \R$, namely, for any $s \in \R$,   
\by
  g_{\lambda, \beta}(s)  &=& g_2(s) =     (1-\beta)d+\frac{\Delta\big (s, \, \sqrt{b_0s^2-2b_1s+b_2} \big )-(2-\lambda-\beta)(d-s)}{2} \\
  &=&   (1-\beta)d+\frac{\sqrt{(b_0s^2-2b_1s+b_2+(d-s)^2)(\beta-\lambda)^2} -(2-\lambda-\beta)(d-s)}{2} \\
  &=& (1-\beta)d+\frac{(1-\beta)\kappa \sqrt{(b_0s^2-2b_1s+b_2+(d-s)^2}}{2}  -\frac{ ((\kappa+1)(1-\beta) 
   + 1-\beta)(d-s)}{2}.
     \ey
 By the proof of Lemma  \ref{le-for-g(s)}, we see that $ g_{\lambda, \beta}(s)=g_2(s) $ is coercive and strictly convex if  $b_0+1>\frac{(2-\lambda-\beta)^2}{(\beta-\lambda)^2}$  or equivalently  $\kappa > \kappa^*$.   Thus, $g_{\lambda, \beta}(s)$ has a unique minimizer on $\R$, which is the unique solution to 
 the equation  of $g_{\lambda, \beta}^\prime(s)=0$, which yields  the expression of  $s^*$ given in   \eqref{s*}.
     \end{proof} 

\begin{proof}[\textbf{Proof of Theorem \ref{th:ROPS-TVaR}}] 
First, for the inner maximization problem in \eqref{CVaR-D}, we have
\be
\label{model2of-TVaR}
        & & \sup_{G \in \mathbb{D}_n(\gamma_1,\gamma_2)}{\rm CVaR}_{\alpha}((1-\beta)\boldsymbol{\omega}^{\top}\boldsymbol{X}^G+(\beta-\lambda)(\boldsymbol{\omega}^{\top}\boldsymbol{X}^G-d)_{+}) \nonumber \\
        &=&\sup_{(\boldsymbol{\mu},\boldsymbol{\Sigma})\in \mathcal{U}_{(\hat{\boldsymbol{\mu}},\hat{\boldsymbol{\Sigma}})}} \sup_{G \in \mathcal{M}(\boldsymbol{\mu},\boldsymbol{\Sigma})} {\rm CVaR}_{\alpha}((1-\beta)\boldsymbol{\omega}^{\top}\boldsymbol{X}^G+(\beta-\lambda)(
        \boldsymbol{\omega}^{\top}\boldsymbol{\mu}-d)_{+}) \nonumber  \\
        &=&\sup_{(\boldsymbol{\mu},\boldsymbol{\Sigma})\in \mathcal{U}_{(\hat{\boldsymbol{\mu}},\hat{\boldsymbol{\Sigma}})}} \sup_{F \in  \mathcal S(\boldsymbol{\omega}^{\top}\boldsymbol{\mu}, \ 
       \boldsymbol{\omega}^{\top}\boldsymbol{\Sigma}\boldsymbol{\omega})}
       {\rm CVaR}_{\alpha}((1-\beta) X^F +(\beta-\lambda)(
        X^F-d)_{+})       
       \nonumber \\
       & =& \sup_{(\boldsymbol{\mu},\boldsymbol{\Sigma})\in \mathcal{U}_{(\hat{\boldsymbol{\mu}},\hat{\boldsymbol{\Sigma}})}} 
       f_{\lambda, \beta}(\boldsymbol{\omega}^{\top}\boldsymbol{\mu}, \ 
       {\color{black} \sqrt{\boldsymbol{\omega}^{\top}\boldsymbol{\Sigma}\boldsymbol{\omega}}}) 
    \nonumber \\
        &=& \sup_{\boldsymbol{\mu}\in\mathcal{U}_{\hat{\boldsymbol{\mu}}}} \sup_{\boldsymbol{\Sigma}\in \mathcal{U}_{\hat{\boldsymbol{\Sigma}}}} 
        f_{\lambda, \beta}(\boldsymbol{\omega}^{\top}\boldsymbol{\mu}, \ {\color{black} \sqrt{\boldsymbol{\omega}^{\top}\boldsymbol{\Sigma}\boldsymbol{\omega}}}), 
    \ee
   where the second equality follows from  Lemma \ref{le:A=B}. 
    
    From Theorem \ref{Prop-wcTVaR-institutional}(b), we see that for any given  $ \boldsymbol{\mu}\in\mathcal{U}_{\hat{\boldsymbol{\mu}}}$, 
    the function  $ f_{\lambda, \beta}(\boldsymbol{\omega}^{\top}\boldsymbol{\mu}, \ 
       {\color{black}\sqrt{\boldsymbol{\omega}^{\top}\boldsymbol{\Sigma}\boldsymbol{\omega}}})$ is increasing in  $\boldsymbol{\omega}^{\top}\boldsymbol{\Sigma}\boldsymbol{\omega}$. Thus,  
       \by
        \sup_{\boldsymbol{\Sigma}\in \mathcal{U}_{\hat{\boldsymbol{\Sigma}}}} 
        f_{\lambda, \beta}(\boldsymbol{\omega}^{\top}\boldsymbol{\mu}, \ 
       {\color{black}\sqrt{\boldsymbol{\omega}^{\top}\boldsymbol{\Sigma}\boldsymbol{\omega}}}) &=&
          f_{\lambda, \beta}(\boldsymbol{\omega}^{\top}\boldsymbol{\mu}, \ 
        \sup_{\boldsymbol{\Sigma}\in \mathcal{U}_{\hat{\boldsymbol{\Sigma}}}} {\color{black}\sqrt{\boldsymbol{\omega}^{\top}\boldsymbol{\Sigma} 
        \boldsymbol{\omega}}}) \\
        &=& f_{\lambda, \beta}(\boldsymbol{\omega}^{\top}\boldsymbol{\mu}, \ 
        {\color{black}\sqrt{\boldsymbol{\omega}^{\top}\boldsymbol{\Sigma}^*
        \boldsymbol{\omega}}})=f_{\lambda, \beta}(\boldsymbol{\omega}^{\top}\boldsymbol{\mu}, \ 
      {\color{black}\sqrt{\boldsymbol{\omega}^{\top}(\hat{\boldsymbol{\Sigma}}+\gamma_2 \boldsymbol{I}_n)\boldsymbol{\omega}}}),            
    \ey
where $\boldsymbol{\Sigma}^*$ is given in  \eqref{maxer-Sigma} and $\boldsymbol{I}_n$ is an $n \times n$ identity matrix. 

Therefore, 
problem \eqref{model2of-TVaR} simplifies to the following  optimization problem: 
\be
\label{one-sum-g}
   \sup_{\boldsymbol{\mu}\in\mathcal{U}_{\hat{\boldsymbol{\mu}}}}  
        f_{\lambda, \beta}(\boldsymbol{\omega}^{\top}\boldsymbol{\mu}, \ 
       {\color{black}\sqrt{\boldsymbol{\omega}^{\top} {\boldsymbol{\Sigma}}^*     \boldsymbol{\omega}}}) =  \sup_{\boldsymbol{\mu}\in\mathcal{U}_{\hat{\boldsymbol{\mu}}}} f_{\lambda, \beta}(\boldsymbol{\omega}^{\top}\boldsymbol{\mu}, \ 
      {\color{black}\sqrt{\boldsymbol{\omega}^{\top}(\hat{\boldsymbol{\Sigma}}+\gamma_2 \boldsymbol{I}_n)\boldsymbol{\omega} }} 
    ). 
       \ee

From the expressions of $f_{\lambda, \beta}(\mu, \sigma)$ given in Theorem \ref{Prop-wcTVaR-institutional},   it is easy to see that for a given $\sigma >0$, 
the function $f_{\lambda, \beta}(\mu, \sigma)$ is increasing in $\mu \in \R$. Thus,  $f_{\lambda, \beta}(\boldsymbol{\omega}^{\top}\boldsymbol{\mu}, \ 
      {\color{black}\sqrt{\boldsymbol{\omega}^{\top}(\hat{\boldsymbol{\Sigma}}+\gamma_2 \boldsymbol{I}_n)\boldsymbol{\omega}}} )$ is increasing in  $\boldsymbol{\omega}^{\top}\boldsymbol{\mu} \in \R$. 
      Therefore, 
      \be
      \label{sup-g=g}
    \sup_{\boldsymbol{\mu}\in\mathcal{U}_{\hat{\boldsymbol{\mu}}}}f_{\lambda, \beta}(\boldsymbol{\omega}^{\top}\boldsymbol{\mu}, \ 
      {\color{black}\sqrt{\boldsymbol{\omega}^{\top}(\hat{\boldsymbol{\Sigma}}+\gamma_2 \boldsymbol{I}_n)\boldsymbol{\omega}}}) 
      =f_{\lambda, \beta} \big (\sup_{\boldsymbol{\mu}\in\mathcal{U}_{\hat{\boldsymbol{\mu}}}}\boldsymbol{\omega}^{\top}\boldsymbol{\mu}, \ 
      {\color{black}\sqrt{\boldsymbol{\omega}^{\top}(\hat{\boldsymbol{\Sigma}}+\gamma_2 \boldsymbol{I}_n)\boldsymbol{\omega}}} \big ).   
       \ee
  Note that  $\sup_{\boldsymbol{\mu}\in\mathcal{U}_{\hat{\boldsymbol{\mu}}}} \boldsymbol{\omega}^{\top}\boldsymbol{\mu}=-\min_{\boldsymbol{\mu}\in\mathcal{U}_{\hat{\boldsymbol{\mu}}}}-\boldsymbol{\omega}^{\top}\boldsymbol{\mu}.$
 Therefore, it follows from \eqref{sup-g=g} and  \eqref{optimalmu} that 
 \by
    \sup_{\boldsymbol{\mu}\in\mathcal{U}_{\hat{\boldsymbol{\mu}}}} f_{\lambda, \beta}(\boldsymbol{\omega}^{\top}\boldsymbol{\mu}, \ 
      {\color{black}\sqrt{\boldsymbol{\omega}^{\top}(\hat{\boldsymbol{\Sigma}}+\gamma_2 \boldsymbol{I}_n)\boldsymbol{\omega}}}) 
      &= & f_{\lambda, \beta}\big (\boldsymbol{\omega}^{\top}\boldsymbol{\mu}^*, \ 
      {\color{black}\sqrt{\boldsymbol{\omega}^{\top}(\hat{\boldsymbol{\Sigma}}+\gamma_2 \boldsymbol{I}_n)\boldsymbol{\omega}}} \big ) \\
     &= & f_{\lambda, \beta}\big (\boldsymbol{\omega}^{\top} \hat{\boldsymbol{\mu}} + \sqrt{\gamma_1}\sqrt{\boldsymbol{\omega}^{\top}\hat{\boldsymbol{\Sigma}}{\boldsymbol{\omega}}}, \ 
      {\color{black}\sqrt{\boldsymbol{\omega}^{\top}(\hat{\boldsymbol{\Sigma}}+\gamma_2 \boldsymbol{I}_n)\boldsymbol{\omega}}} \big ), 
           \ey
  which, together with \eqref{min-cons}, implies that  problem  \eqref{CVaR-D}  simplifies to the following finite-dimensional optimization problem: 
       \begin{equation*}
\label{CVaR-D-simple}
    \begin{aligned}
        \min_{\boldsymbol{\omega}\in\mathbb{R}^n}&    \  f_{\lambda, \beta}\big (\boldsymbol{\omega}^{\top} \hat{\boldsymbol{\mu}} + \sqrt{\gamma_1}\sqrt{\boldsymbol{\omega}^{\top}\hat{\boldsymbol{\Sigma}}{\boldsymbol{\omega}}}, \ 
      {\color{black}\sqrt{\boldsymbol{\omega}^{\top}(\hat{\boldsymbol{\Sigma}}+\gamma_2 \boldsymbol{I}_n)\boldsymbol{\omega}}} \big ) \\
         s.t. &  \, -\hat{\boldsymbol{\mu}}^{\top}\hat{\boldsymbol{\omega}}-\sqrt{\gamma_1}\sqrt{\boldsymbol{\omega}^{\top} \hat{\boldsymbol{\Sigma}}{\boldsymbol{\omega}}}  \geq r,\ \  \ \  \boldsymbol{\omega}^{\top}\boldsymbol{e}=1.
    \end{aligned}
\end{equation*}
Note that the function $f_{\lambda, \beta}(x,y)$  defined in \eqref{TVaR-f}  is increasing in $(x,y)$. Thus,  for any $a_0, x, y\in\mathbb{R}$, we can write 
$f_{\lambda, \beta}(a_0+x,y)=\min_{s\geq x, \, t\geq y}f_{\lambda, \beta}(a_0+s,t).$
Therefore, by introducing auxiliary variables $s,t \in\mathbb{R}^+$,  we can reformulate this optimization problem  \eqref{CVaR-D-simple} as problem \eqref{TVaR-convex-expression}. 

The objective function   $f_{\lambda, \beta}\big (\boldsymbol{\omega}^{\top} \hat{\boldsymbol{\mu}} + \sqrt{\gamma_1} \,  t, \  s \big ) $  is  a convex function of $(\bm{\omega} ,s,t)$,  since  the function $f_{\lambda, \beta}(\mu,\sigma)$ 
 is convex in $(\mu,\sigma)\in\R\times\R^+$ by Proposition \ref{f-convex}. Moreover,  the convexity is preserved under an affine mapping (see, 
 for example,   Section 3.2.2 of \cite{boyd2004convex}).   Thus this completes the proof. 
    \end{proof}

\begin{proof}[\textbf{Proof of Proposition \ref{Pro-5.1}}]
 By letting $\alpha \to 0$ in Theorem  \ref{Prop-wcTVaR-institutional}, we note that $d_1 \to -\infty$, 
$d_2 \to \infty$, and that $\Delta(\mu,\sigma) \to (\beta-\lambda) \sqrt{\sigma^2+(d-\mu)^2}  $  as $\alpha \to 0$ for $\beta > \lambda$  in \eqref{gb2}. As a result, 
we observe that   
the function $
f_{\lambda, \beta}\big (\boldsymbol{\omega}^{\top} \hat{\boldsymbol{\mu}} + \sqrt{\gamma_1} \,  t, \  s \big )$ 
in Theorem \ref{th:ROPS-TVaR} simplifies to 
\by
&& f_{\lambda, \beta}\big (\boldsymbol{\omega}^{\top} \hat{\boldsymbol{\mu}} + \sqrt{\gamma_1} \,  t, \  s \big ) \\
&=& (1-\beta) d+\frac{ \Delta \big (\boldsymbol{\omega}^{\top} \hat{\boldsymbol{\mu}} + \sqrt{\gamma_1} \,  t, \  s \big ) 
 -(2-\lambda-\beta)(d- (\boldsymbol{\omega}^{\top} \hat{\boldsymbol{\mu}} + \sqrt{\gamma_1} \,  t))}{2} \\
&=& (1-\beta) d+\frac{(\beta-\lambda) \sqrt{s^2+(\boldsymbol{\omega}^{\top} \hat{\boldsymbol{\mu}} + \sqrt{\gamma_1} \,  t-d)^2}
 -(2-\lambda-\beta)(d- (\boldsymbol{\omega}^{\top} \hat{\boldsymbol{\mu}} + \sqrt{\gamma_1} \,  t))}{2} \\
 &=& (1-\beta) \Big ((1+\frac{\kappa}{2})(\hat{\boldsymbol{\mu}}^{\top}\boldsymbol{\omega}+\sqrt{\gamma_1}t)-\frac{\kappa}{2}d+\frac{\kappa}{2}\sqrt{s^2+(\hat{\boldsymbol{\mu}}^{\top}\boldsymbol{\omega}+\sqrt{\gamma_1}t-d)^2} \,  \Big), 
 \ey
where the last equality follows from  $1-\lambda= (\kappa+1)(1-\beta)$ and $\beta-\lambda = (1-\beta) \kappa$. Since $1-\beta$ is a constant and  $ f_{\lambda, \beta} \big (\boldsymbol{\omega}^{\top} \hat{\boldsymbol{\mu}} + \sqrt{\gamma_1} \,  t, \  s \big )$ is proportional to 
 $(1+\frac{\kappa}{2})(\hat{\boldsymbol{\mu}}^{\top}\boldsymbol{\omega}+\sqrt{\gamma_1}t)-\frac{\kappa}{2}d+\frac{\kappa}{2}\sqrt{s^2+(\hat{\boldsymbol{\mu}}^{\top}\boldsymbol{\omega}+\sqrt{\gamma_1}t-d)^2}$, Theorem \ref{th:ROPS-TVaR} implies that  
 problem \eqref{E-D} can be solved by solving problem \eqref{prop 1}. To verify that this problem is a convex optimization problem, first note that $(\boldsymbol{\omega},s,t)\mapsto\sqrt{s^2+(\hat{\boldsymbol{\mu}}^{\top}\boldsymbol{\omega}+\sqrt{\gamma_1}t-d)^2}$ is a convex function. Moreover, all functions on the left-hand side of the inequality constraints are convex. 
  This follows from the fact that $({\bf x},b)\mapsto h({\bf a}^\top {\bf x}+b)$ is convex for any ${\bf a}\in\R^n$ whenever $h$ is convex, since convexity is preserved under affine mappings (see e.g. Section 3.2.2 of \cite{boyd2004convex}), and the $\ell_2$-norm $(x_1,x_2)\mapsto \sqrt{x_1^2+x_2^2}$ is  convex.
 \end{proof}

\end{document}